\def\llncs{0}
\def\fullpage{1}
\def\anonymous{0}
\def\authnote{1}
\def\notxfont{0}
\def\submission{0}
\def\cameraready{0}

\ifnum\submission=1
\def\anonymous{1}
\def\llncs{1}
\def\authnote{0}
\else
\fi

\ifnum\cameraready=1
\def\submission{1}
\def\llncs{1}
\def\authnote{0}
\def\anonymous{0}
\else
\fi

\ifnum\anonymous=1
\def\llncs{1}
\def\authnote{0}
\else
\fi

\ifnum\llncs=1
	\documentclass[envcountsect,a4paper,runningheads,10pt]{llncs}
\else
	\documentclass[letterpaper,hmargin=1.05in,vmargin=1.05in]{article}
			\ifnum\fullpage=1
		\usepackage{fullpage}
		\fi
\fi

\usepackage[%
  colorlinks=true,
  citecolor=blue,
  pagebackref=true
]{hyperref}

\usepackage{amsmath, amsfonts, amssymb, mathtools,amscd}

\usepackage{amsthm}

\usepackage{lmodern}
\usepackage[T1]{fontenc}
\usepackage[utf8]{inputenc}

\usepackage{arydshln} 
\usepackage{url}
\usepackage{ifthen}
\usepackage{bm}
\usepackage{multirow}
\usepackage[dvips]{graphicx}
\usepackage[usenames]{color}
\usepackage{xcolor,colortbl} 
\usepackage{threeparttable}
\usepackage{comment}
\usepackage{paralist,verbatim}
\usepackage{cases}
\usepackage{booktabs}
\usepackage{braket}
\usepackage{cancel} 
\usepackage{ascmac} 
\usepackage{framed}
\usepackage{physics}
\usepackage{authblk}
\usepackage{pifont}
\usepackage{autonum}
\definecolor{darkblue}{rgb}{0,0,0.6}
\definecolor{darkgreen}{rgb}{0,0.5,0}
\definecolor{maroon}{rgb}{0.5,0.1,0.1}
\definecolor{dpurple}{rgb}{0.2,0,0.65}

\usepackage[capitalise,noabbrev]{cleveref}
\usepackage[absolute]{textpos}
\usepackage[final]{microtype}
\usepackage[absolute]{textpos}
\usepackage{everypage}

\newtheoremstyle{thicktheorem}%
{\topsep}
{\topsep}
{\itshape}{}%
{\bfseries}%
{.}
{ }%
{\thmname{#1}\thmnumber{ #2}%
		\thmnote{ (#3)}%
}

\newtheoremstyle{remark}
{\topsep}
{\topsep}
	{}
	{}
	{}
	{.}
	{ }
	{\textit{\thmname{#1}}\thmnumber{ #2}
			\thmnote{ (#3)}%
	}

\ifnum\llncs=0
	\theoremstyle{thicktheorem}
	\newtheorem{theorem}{Theorem}[section]
	\newtheorem{lemma}[theorem]{Lemma}
	\newtheorem{corollary}[theorem]{Corollary}
	\newtheorem{proposition}[theorem]{Proposition}
	\newtheorem{definition}[theorem]{Definition}
	\newtheorem{game}[theorem]{Game}

	\theoremstyle{remark}
	
	\newtheorem{remark}[theorem]{Remark}

\else
\fi

	\crefname{theorem}{Theorem}{Theorems}
	\crefname{assumption}{Assumption}{Assumptions}
	\crefname{construction}{Construction}{Constructions}
	\crefname{corollary}{Corollary}{Corollaries}
	\crefname{conjecture}{Conjecture}{Conjectures}
	\crefname{definition}{Definition}{Definitions}
	\crefname{exmaple}{Example}{Examples}
	\crefname{experiment}{Experiment}{Experiments}
	\crefname{counterexample}{Counterexample}{Counterexamples}
	\crefname{lemma}{Lemma}{Lemmata}
	\crefname{observation}{Observation}{Observations}
	\crefname{proposition}{Proposition}{Propositions}
	\crefname{remark}{Remark}{Remarks}
	\crefname{claim}{Claim}{Claims}
	\crefname{fact}{Fact}{Facts}
	\crefname{note}{Note}{Notes}

\ifnum\llncs=1
 \crefname{appendix}{App.}{Appendices}
 \crefname{section}{Sec.}{Sections}
\else
\fi

\ifnum\llncs=1
\renewcommand*{\backref}[1]{}
\else
	\renewcommand*{\backref}[1]{(Cited on page~#1.)}
	\ifnum\notxfont=1
	\else
		\usepackage{newtxtext}
	\fi
\fi

\ifnum\authnote=0  
\newcommand{\taiga}[1]{}
\newcommand{\ryo}[1]{}
\newcommand{\takashi}[1]{}
\newcommand{\fuyuki}[1]{}

\else
\newcommand{\taiga}[1]{$\ll$\textsf{\color{magenta} Taiga: { #1}}$\gg$}
\newcommand{\takashi}[1]{$\ll$\textsf{\color{orange} Takashi: { #1}}$\gg$}
\newcommand{\ryo}[1]{$\ll$\textsf{\color{darkgreen} Ryo: { #1}}$\gg$}

\fi

\newcommand{\Unc}{\mathsf{Unclone}}
\newcommand{\unc}{\mathsf{unclone}}

\newcommand{\Univ}{\mathsf{Univ}}

\newcommand{\RE}{\mathsf{RE}}

\newcommand{\Lab}{\mathsf{Lab}}
\newcommand{\lab}{\mathsf{lab}}

\newcommand{\Mint}{\mathsf{Mint}}

\newcommand{\Comb}{\mathsf{Comb}}
\newcommand{\RobComb}{\mathsf{RobComb}}

\newcommand{\StateGen}{\mathsf{StateGen}}
\newcommand{\OWSG}{\mathsf{OWSG}}

\newcommand{\Commit}{\algo{Commit}}

\newcommand{\NCE}{\mathsf{NCE}}
\newcommand{\nce}{\mathsf{nce}}
\newcommand{\Fake}{\algo{Fake}}
\newcommand{\Reveal}{\algo{Reveal}}

\newcommand{\Sim}{\algo{Sim}}

\newcommand{\NP}{\compclass{NP}}

\newcommand{\la}{\leftarrow}
\newcommand{\ra}{\rightarrow}

\newcommand{\seteq}{\coloneqq}

\newcommand{\cA}{\mathcal{A}}
\newcommand{\cB}{\mathcal{B}}
\newcommand{\cC}{\mathcal{C}}
\newcommand{\cD}{\mathcal{D}}

\newcommand{\cF}{\mathcal{F}}
\newcommand{\cG}{\mathcal{G}}

\newcommand{\cI}{\mathcal{I}}

\newcommand{\cM}{\mathcal{M}}

\newcommand{\cO}{\mathcal{O}}

\newcommand{\cT}{\mathcal{T}}

\newcommand{\cW}{\mathcal{W}}

\newcommand{\cZ}{\mathcal{Z}}

\def\makeuppercase#1{
\expandafter\newcommand\csname tl#1\endcsname{\widetilde{#1}}
}

\def\makelowercase#1{
\expandafter\newcommand\csname tl#1\endcsname{\widetilde{#1}}
}

\newcommand{\N}{\mathbb{N}}

\newcommand{\R}{\mathbb{R}}

\newcommand{\Ms}{\mathcal{M}}

\newcommand{\sep}{\lambda}
\newcommand{\secp}{\lambda}

\newcommand{\aux}{\mathsf{aux}}

\newcommand*{\sk}{\keys{sk}}
\newcommand*{\pk}{\keys{pk}}

\newcommand{\ct}{\keys{CT}}

\newcommand*{\MSK}{\keys{MSK}}

\newcommand*{\keys}[1]{\mathsf{#1}}

\newcommand*{\algo}[1]{\ensuremath{\mathsf{#1}}}

\newcommand{\compclass}[1]{\textbf{\textrm{#1}}}

\newenvironment{boxfig}[2]{\begin{figure}[#1]\fbox{\begin{minipage}{0.97\linewidth}
                        \vspace{0.2em}
                        \makebox[0.025\linewidth]{}
                        \begin{minipage}{0.95\linewidth}
            {{
                        #2 }}
                        \end{minipage}
                        \vspace{0.2em}
                        \end{minipage}}}{\end{figure}}

\newcommand{\bit}{\{0,1\}}

\newcommand{\fail}{\mathtt{fail}}

\newcommand{\Setup}{\algo{Setup}}

\newcommand{\keygen}{\algo{KeyGen}}

\newcommand{\Enc}{\algo{Enc}}
\newcommand{\Dec}{\algo{Dec}}

\newcommand{\Vrfy}{\algo{Vrfy}}

\newcommand\SKE{\algo{SKE}}

\newcommand{\ske}{\mathsf{ske}}

\newcommand\PKE{\algo{PKE}}

\newcommand{\pke}{\mathsf{pke}}

\newcommand{\negl}{{\mathsf{negl}}}

\newcommand{\poly}{{\mathrm{poly}}}

\DeclareMathOperator*{\Exp}{\mathbb{E}}

\usepackage{fancyhdr}

\ifnum\llncs=1
\let\oldvec\vec
\let\vec\oldvec
\makeatletter
\renewcommand*\l@author[2]{}
\renewcommand*\l@title[2]{}
\makeatletter

\else
\fi    

\theoremstyle{remark}

\title{
\textbf{Robust Combiners and Universal Constructions for Quantum Cryptography}
}

\begin{document}

\ifnum\anonymous=1
\author{\empty}\institute{\empty}
\else
%
%
\ifnum\llncs=1
\index{Taiga, Hiroka}
\author{
	Taiga Hiroka\inst{1} 
}
\institute{
	Yukawa Institute for Theoretical Physics, Kyoto University, Japan  \and NTT Corporation, Tokyo, Japan
}
\else
%
%
%

\author[$\star$]{Taiga Hiroka}
\author[$\dagger$$\diamondsuit$]{Fuyuki Kitagawa}
\author[$\dagger$$\diamondsuit$]{Ryo Nishimaki}
\author[$\dagger$$\diamondsuit$$\star$]{Takashi Yamakawa}
\affil[$\star$]{{\small Yukawa Institute for Theoretical Physics, Kyoto University, Japan}\authorcr{\small taiga.hiroka@yukawa.kyoto-u.ac.jp}}
\affil[$\dagger$]{{\small NTT Social Informatics Laboratories, Tokyo, Japan}\authorcr{\small \{fuyuki.kitagawa,ryo.nishimaki,takashi.yamakawa\}@ntt.com}}
\affil[$\diamondsuit$]{{\small NTT Research Center for Theoretical Quantum Infomration}}
\renewcommand\Authands{, }

\fi 
\fi

\ifnum\llncs=1
\date{}
\else
\date{\today}
\fi

\maketitle
\thispagestyle{fancy}
\rhead{YITP-23-137}

\begin{abstract}
A robust combiner combines many candidates for a  cryptographic primitive and generates a new candidate for the same primitive.
Its correctness and security hold as long as one of the original candidates satisfies correctness and security.
A universal construction is a closely related notion to a robust combiner.
A universal construction for a primitive is an explicit construction of the primitive that is correct and secure as long as the primitive exists.
It is known that a universal construction for a primitive can be constructed from a robust combiner for the primitive in many cases.

Although robust combiners and universal constructions for classical cryptography are widely studied, robust combiners and universal constructions for quantum cryptography have not been explored so far.
In this work, we define robust combiners and universal constructions for several quantum cryptographic primitives including one-way state generators, public-key quantum money, quantum bit commitments, and unclonable encryption, and provide constructions of them.

On a different note, it was an open problem how to expand the plaintext length of unclonable encryption.
In one of our universal constructions for unclonable encryption, we can expand the plaintext length, which resolves the open problem.
\end{abstract}

\ifnum\cameraready=1
\else
\ifnum\llncs=1
\else
\newpage
  \setcounter{tocdepth}{2}      
  \setcounter{secnumdepth}{2}   
  \setcounter{page}{0}          
  \tableofcontents
  \thispagestyle{empty}
  \clearpage
\fi
\fi


\section{Introduction}\label{Sec:Introduction}
\subsection{Background}

The ultimate goal of theoretical cryptography is to construct interesting cryptographic primitives unconditionally.
Over the past years, many computational assumptions have been proposed, and many interesting cryptographic primitives have been constructed under the computational assumptions.
However, none of the computational assumptions are proven.
Indeed, we do not even know how to prove $\compclass{P}\neq\NP$ while it is a necessary condition to construct interesting classical cryptographic primitives unconditionally.
Moreover, given many candidates for a primitive, we cannot often decide which candidate is the most secure one.
For example, we can construct public-key encryption (PKE) from decisional Diffie-Hellman (DDH)~\cite{DifHel76,ElGamal85} or learning with errors (LWE)~\cite{STOC:Regev05}, but currently, we do not know which computational assumption is the weaker assumption.
This causes the problem in the following realistic scenario.
Suppose we have two candidates for PKE, where one is based on DDH and the other is based on LWE, and we want to decide more secure candidate to use. 
Unfortunately, in the current knowledge, we cannot decide which candidate is the more secure one.

A robust cryptographic combiner~\cite{TC:AA05,EC:HKNRR05}  was introduced to resolve this issue.
Given many candidates for a primitive, a cryptographic combiner combines these candidates and produces a new candidate for the same primitive.
The new candidate is correct and secure as long as at least one of the original candidates satisfies correctness and security.
For example, a robust PKE combiner takes two candidates for PKE, where one's security relies on DDH and the other's security relies on LWE, and produces a new candidate for PKE.
The new candidate is correct and secure as long as the DDH or LWE assumption holds.
Robust combiner is a well-studied topic in classical cryptography.
In fact, robust combiners for many fundamental classical cryptographic primitives such as one-way functions, public-key encryption, and functional encryption are shown to exist~\cite{EC:HKNRR05,C:AJNSY16,EC:AJS17,TCC:ABJMS19,EC:JaiManSah20}.

A closely related notion to a robust combiner is a universal construction~\cite{STOC:Levin85}.
A universal construction for a primitive, say OWFs, is an explicit construction of OWFs that is correct and secure as long as OWFs exist.
The adversary must be able to break all OWF candidates to break a universal construction. 
In this sense, a universal construction for OWFs is the most secure one among all possible OWF candidates. 
In classical cryptography, universal constructions are well-studied topic and are known to exist for many fundamental primitives.
First, the pioneering work by Levin introduces a notion of universal construction and shows how to construct a universal construction for OWFs~\cite{STOC:Levin85}.
After decades, Harnik, Kilian, Naor, Reingold, and Rosen~\cite{EC:HKNRR05} give a universal construction for PKE and they show how to construct a universal construction for a primitive using a robust combiner for the same primitive. 
Goldwasser and Kalai cast questions about universal constructions for cryptographic primitives related to obfuscation~\cite{TCC:GK16}. 
The following sequence of works~\cite{C:AJNSY16,EC:AJS17,TCC:ABJMS19} gives universal constructions for functional encryption under some assumptions, and \cite{EC:JaiManSah20} gives it unconditionally.

Although robust combiners and universal constructions are widely studied topics in classical cryptography, those in the quantum world have not been studied so far, where each party can generate, process, and communicate quantum information.
It is well known that, even in the quantum world, information-theoretical security is impossible to achieve for many interesting quantum cryptographic primitives~\cite{LC97,Mayers97,STOC:Aaronson18}, and currently, many interesting quantum cryptographic primitives are constructed under computational assumptions.
For example, public-key quantum money is one of the most interesting quantum cryptographic primitives, and many candidate constructions are proposed relying on computational assumptions~\cite{STOC:AC12,ITCS:FGHLS12,arXiv:Kane18,EC:Zhandry19b,kane2022quantum,Euro:JHM23,Zhandry23}.
However, none of them have been proven so far, and moreover, we cannot even decide which assumptions are the weakest assumptions.
This inability leads to the problem that we cannot decide the most secure one to use.

If there exists a robust public-key quantum money combiner, then we can combine them and produce a new candidate for public-key quantum money, which is secure as long as at least one of the original candidates is secure.
Therefore, it is natural to ask the following first question:
\begin{center}
\textit{
Is it possible to construct robust combiners for fundamental quantum cryptographic primitives?
}
\end{center}

On a different note, recent works show the possibility that quantum cryptography exists even if classical cryptography does not.
A pseudo-random state generator (PRSG) is a quantum analog of a pseudo-random generator~\cite{C:JiLiuSon18}, and Kretchmer shows the possibility that PRSGs exist even if $\compclass{BQP} = \compclass{QMA}$~\cite{TQC:Kre21}.
Many interesting quantum cryptographic primitives are shown to be constructed from PRSGs~\cite{Crypto:MY22,Eprint:MY22, Crypto:AQY22,TCC:AGQY22,ITCS:BCQ23}.
Among them, one-way state generators (OWSGs) and quantum bit commitments (equivalent to EFI~\cite{Asia:Yan22,ITCS:BCQ23}) are considered to be candidates for the necessary assumptions for the existence of quantum cryptography.
In the case of classical cryptography, many fundamental primitives have the nice feature of the existence of universal constructions.
It is natural to wonder whether quantum cryptographic primitives have universal constructions or not.
In fact, some researchers believe that the existence of universal constructions is a nice feature for fundamental cryptographic primitives~\cite{Talk:Zha23}.
Therefore, we ask the following second question:
\begin{center}
\textit{
Is it possible to construct universal constructions for fundamental quantum cryptographic primitives?
}
\end{center}

\subsection{Our Results}
We solve the two questions above affirmatively for several cryptographic primitives.
Our contributions to the field are as follows:
\begin{enumerate}
    \item We formally define robust combiners and universal constructions for many quantum cryptographic primitives including OWSGs, public-key quantum money, quantum bit commitments, and unclonable encryption.
    \item We construct a robust combiner and a universal construction for OWSGs without any assumptions.
    A universal construction is secure as long as there exist OWSGs.
    In other words, the adversary of a universal construction must be able to break all OWSG candidates.
    In this sense, our construction for OWSG is the most secure one among all possible OWSG candidates.
    Before this work, the candidate constructions for OWSGs were based on OWFs, average-case hardness of semi-classical quantum statistical difference~\cite{Eprint:CX23} or random quantum circuits~\cite{Crypto:AQY22,ITCS:BCQ23}
    \footnote{As discussed in the previous works~\cite{Crypto:AQY22,ITCS:BCQ23}, it is a folklore that a random quantum circuit is PRSGs although there exists no theoretical evidence so far.
    Since we can construct OWSGs from PRSGs~\cite{Crypto:MY22,Eprint:MY22}, we can also construct OWSGs based on random quantum circuits if a random quantum circuit is PRSGs.}.
    \item We construct a robust combiner and a universal construction for public-key quantum money without any assumptions.
    In particular, in this work, we consider the public-key quantum money mini-scheme introduced in~\cite{STOC:AC12}, which can be generically upgraded into full-fledged public-key quantum money by additionally using digital signatures. 
    A universal construction for a public-key quantum money mini-scheme satisfies security as long as a public-key quantum money mini-scheme exists.
    In other words, the adversary of a universal construction must be able to break all candidates for a public-key quantum money mini-scheme.
    In this sense, our construction is the most secure one among all possible public-key quantum money mini-scheme candidates.
    Before this work, many candidate constructions are proposed~\cite{STOC:AC12,ITCS:FGHLS12,arXiv:Kane18,EC:Zhandry19b,kane2022quantum,Euro:JHM23,Zhandry23}.
    \item We construct a robust combiner and a universal construction for quantum bit commitment without any assumptions.
    Note that our results also imply that we can construct a robust combiner and a universal construction for EFI, oblivious transfer, and multi-party computation, which are equivalent to quantum bit commitments~\cite{ITCS:BCQ23}.
    In our robust combiner, given $ n$-candidates of quantum bit commitments, we can construct a new quantum bit commitment that satisfies statistical binding and computational hiding at least one of $n$-candidates satisfies computational hiding and computational binding at the same time.
    A universal construction for quantum bit commitment is secure as long as there exists a quantum bit commitment.
    In other words, the adversary for a universal construction must be able to break all candidates for quantum bit commitment.
    In this sense, our construction for quantum bit commitment is the most secure one among all possible quantum bit commitment candidates.
    Before this work, candidate constructions of quantum bit commitments were based on OWFs, classical oracle~\cite{STOC:KQST23}, or random quantum circuits~\cite{Crypto:AQY22,ITCS:BCQ23}
    \footnote{It is a folklore that a random quantum circuit is PRSGs although there exists no theoretical evidence so far.
    Since we can construct quantum bit commitments from PRSGs~\cite{Crypto:MY22, Crypto:AQY22}, we can also construct quantum bit commitments based on random quantum circuits if a random quantum circuit is PRSGs.}.
    \item We construct robust combiners and universal constructions for various kinds of unclonable encryption as follows:
    \begin{itemize}
    \item We construct robust combiners for (one-time) unclonable secret-key encryption (SKE) and unclonable public-key encryption (PKE) without any computational assumptions.
    \item By using robust combiners, we construct universal constructions for (one-time) unclonable SKE and unclonable PKE without any computational assumptions.
    \end{itemize}
    Although the previous work~\cite{C:AKLLZ23} gives a construction of one-time unclonable SKE with unclonable IND-CPA security in the quantum random oracle model (QROM), it was an open problem to construct it in the standard model.
    Our universal constructions for (one-time) unclonable SKE (resp. PKE) is the first construction of (one-time) unclonable SKE (resp. PKE) that achieves unclonable IND-CPA security in the standard model, where the security relies on the existence of (one-time) unclonable SKE (resp. PKE) with unclonable IND-CPA security.

\item We give another construction of universal construction for one-time unclonable SKE by additionally using the decomposable quantum randomized encoding~\cite{Stoc:BY22}.
Although this construction additionally uses decomposable quantum randomized encoding, it has the following nice three properties that the universal construction via a robust combiner does not have:
\begin{itemize}
    \item 
    It was an open problem whether unclonable encryption with single-bit plaintexts implies unclonable encryption with multi-bit plaintexts because standard transformation via bit-wise encryption does not work as pointed out in \cite{C:AKLLZ23}.
    In our universal construction, we can expand the plaintext length of one-time unclonable SKE by additionally using decomposable quantum randomized encoding. This resolves the open problem left by \cite{C:AKLLZ23}.
    Note that this result implies that reusable unclonable SKE and unclonable PKE can expand plaintext length without any additional assumptions because reusable unclonable SKE and unclonable PKE imply decomposable quantum randomized encoding.
    \item A universal construction via a robust combiner needs to emulate all possible algorithms, and thus a huge constant is included in the running time.
    Therefore, it may not be executed in a meaningful amount of time if we want reasonable concrete security.
    On the other hand, universal construction via decomposable quantum randomized encoding does not emulate all possible algorithms and thus avoids the “galactic inefficiency” tied to such approaches. 
    \item 
    In a universal construction via a robust combiner, the security relies on the existence of one-time unclonable SKE scheme $\Sigma=(\keygen,\Enc,\Dec)$, where $(\keygen,\Enc,\Dec)$ are uniform QPT algorithms.
    On the other hand, in a universal construction via decomposable quantum randomized encoding, the security still holds even if the underlying one-time unclonable SKE $(\keygen,\Enc,\Dec)$ are non-uniform algorithms.
\end{itemize}
\end{enumerate}

\subsection{More on Related Work}
\paragraph{Fundamental Quantum Cryptographic Primitives.}
Ji, Liu, and Song~\cite{C:JiLiuSon18} introduce a notion of PRSGs, and show that it can be constructed from OWFs.
Morimae and Yamakawa~\cite{Crypto:MY22} introduce the notion of OWSGs, and show how to construct them from PRSGs.
In the first definition of OWSGs, the output quantum states are restricted to pure states, and its definition is generalized to mixed states by \cite{Eprint:MY22}.
In this work, we focus on the mixed-state version.

Bennett and Brassard~\cite{BEN84} initiate the study of quantum bit commitment.
Unfortunately, it turns out that statistically secure quantum bit commitments are impossible to achieve~\cite{LC97,Mayers97}.
Therefore, later works study a quantum bit commitment with computational security~\cite{EC:DumMaySal00,EC:CreLegSal01,Asia:Yan22,Crypto:MY22,Eprint:MY22, Crypto:AQY22,TCC:AGQY22,ITCS:BCQ23,EC:HMY23}.
It was shown that quantum bit commitments can be constructed from PRSGs by \cite{Crypto:MY22,Crypto:AQY22}, and that quantum bit commitments are equivalent to EFI, oblivious transfer, and multi-party computation~\cite{EC:GLSV21,C:BCKM21a,Asia:Yan22,ITCS:BCQ23}.

Recently, Khurana and Tomer~\cite{arxiv:KT23} showed that quantum bit commitments can be constructed from OWSGs with pure state.
Although their main result is not a combiner for quantum bit commitment, they construct some sort of a combiner for quantum bit commitments as an intermediate tool for achieving their result.
In their construction, they construct a uniform quantum bit commitment from a non-uniform one.
At this step, they combine quantum bit commitments in the following sense.
In their construction, they combine $(n+1)$-quantum bit commitments and generate a new quantum bit commitment.
Its hiding and binding property holds as long as one of the original candidates satisfies hiding and binding at the same time and other $ n$ candidates also satisfy either hiding or binding.
Compared to their technique, our robust combiner does not need to assume other $n$ candidates satisfy hiding or binding.
Therefore, our robust combiner can be applied in a more general setting than their technique.
Though our construction partially shares a similarity with theirs, we rely on additional ideas to deal with candidate schemes that do not satisfy either binding or hiding.

\paragraph{Unclonable Encryption.} 
Broadbent and Lord~\cite{TQC:BL20} introduced a notion of unclonable encryption.
They considered two security definitions for unclonable encryption.
One is one-wayness against cloning attacks and they achieve information-theoretic one-wayness by using BB84 states.
The other is indistinguishability against cloning attacks (indistinguishable-secure unclonable encryption).
However, they did not achieve it. They constructed indistinguishable-secure unclonable encryption only in a very restricted model by using PRFs.
Ananth, Kaleoglu, Li, Liu, and Zhandry~\cite{C:AKLLZ23} proposed the first indistinguishable-secure
unclonable encryption in the QROM. 
Ananth and Kaleoglu~\cite{TCC:AK21} construct unclonable PKE from unclonable encryption and PKE with ``classical'' ciphertexts.
Note that it is unclear how to apply their technique for PKE with quantum ciphertexts.
The technique of \cite{Asia:HMNY21} can be used to construct unclonable PKE from unclonable encryption and PKE with quantum ciphertexts, which we use in this work.

\paragraph{Combiner for Classical Cryptography.}
It is known that robust combiners are known to exist for many fundamental classical cryptographic primitives.
Oblivious transfer (OT) is an example of exceptions.
It is an open problem how to construct a robust combiner for classical OT and some black-box impossibilities are known~\cite{EC:HKNRR05}.
Interestingly, our result implies that a robust combiner for quantum OT exists although a robust combiner for classical OT is still an open problem.

\subsection{Organization}
In \cref{sec:tech}, we give a technical overview.
In \cref{sec:pre}, we define the notations and preliminaries that we require in this work.
In \cref{sec:owsg}, we define the notions of robust OWSG combiners and a universal construction for OWSGs and provide constructions.
We provide some proof in \cref{sec:padding}.
In \cref{sec:money}, we define the notions of a robust combiner and a universal construction for public-key quantum money mini-scheme and provide constructions.
We provide some proof in \cref{sec:app_money}.
In \cref{sec:commit}, we define the notions of a robust canonical quantum bit commitment combiner and a universal construction for canonical quantum bit commitment and provide constructions.
We provide some proof in \cref{sec:convert}.
In \cref{sec:unc}, we define the notions of robust combiners for unclonable encryption and universal constructions for unclonable encryption and provide constructions.
We provide some proof in \cref{sec:app_unclone_comb,sec:app_unc_pke}.
In \cref{Sec:expanstion_UE}, we provide another universal construction for unclonable encryption.
We provide some proof in \cref{sec:app_perfect_unc}.
In this construction, we can expand the plaintext length of unclonable encryption.

\section{Technical Overview}\label{sec:tech}
First of all, let us recall the definition of robust combiner.
A robust combiner for a primitive $P$ is a deterministic classical polynomial-time Turing machine $\RobComb.\cM_P$ that takes as input $n$-candidates $\{\Sigma[i]\}_{i\in[n]}$ for $P$, and produces a new candidate $\Sigma$ for $P$. 
$\Sigma$ is correct and secure as long as at least one of the candidates $\{\Sigma[i]\}_{i\in[n]}$ for $P$ is correct and secure.
Here, the point is that $\{\Sigma[i]\}_{i\in[n]}$ are not promised to satisfy even correctness other than one of them.
In the following, we will explain the case where only two candidates $\Sigma[1]$ and $\Sigma[2]$ are given for simplicity.
Remark that the same argument goes through in the general case, where $n$ candidates $\{\Sigma[i]\}_{i\in[n]}$ are given.

\subsection{Robust Combiner for One-Way State Generators and Public-Key Quantum Money}
In this section, we explain a robust combiner for OWSGs.
A robust combiner for public-key quantum money can be constructed by partially using the technique by \cite{EC:HKNRR05}. 
\paragraph{Definition of One-Way State Generators.}

OWSG is a quantum generalization of OWFs and consists of a tuple of quantum polynomial-time algorithms $\Sigma_{\OWSG}\seteq(\keygen,\StateGen,\Vrfy)$.
The $\keygen$ algorithm takes as input a security parameter $1^\secp$, and generates a classical key $k$,
the $\StateGen$ algorithm takes as input a classical key $k$ and outputs a quantum state $\psi_k$,
and the $\Vrfy$ algorithm takes as input a classical key $k$ and a quantum state $\psi_k$ and outputs $1$ indicating acceptance or $0$ indicating rejection.
We require that OWSG $\Sigma$ satisfies correctness and security.
The correctness guarantees that $\Vrfy(k,\psi_k)$ outputs $1$ indicating acceptance with overwhelming probability, where $k\la\keygen(1^\secp)$ and $\psi_k\la\StateGen(k)$.
The security guarantees that no QPT adversaries given polynomially many copies of $\psi_k$ cannot generate $k^*$ such that $1\la\Vrfy(k^*,\psi_k)$, where $k\la\keygen(1^\secp)$ and $\psi_k\la\StateGen(k)$.

\paragraph{Robust Combiner.}
First, we consider the simpler case, where given OWSG candidates $\Sigma_{\OWSG}[1]=(\keygen[1],\allowbreak\StateGen[1],\Vrfy[1])$ and $\Sigma_{\OWSG}[2]=(\keygen[2],\StateGen[2],\Vrfy[2])$ are promised to satisfy at least correctness.
In this case, we can construct a combiner for OWSGs in the same way as OWFs.
Namely, a combined protocol $\Comb.\Sigma_{\OWSG}=(\keygen,\StateGen,\Vrfy)$ simply runs $\Sigma[1]$ and $\Sigma[2]$ in parallel.

Does the same strategy work for the general setting, where original candidates are not promised to satisfy correctness?
Unfortunately, the simple parallel protocol works only when both $\Sigma_\OWSG[1]$ and $\Sigma_\OWSG[2]$ satisfy correctness because $\Comb.\Sigma_{\OWSG}$ does not satisfy correctness otherwise.
We observe that given an OWSG candidate $\Sigma_{\OWSG}$, we can construct $\Sigma_{\OWSG}^*$ with the following properties:
\begin{itemize}
    \item $\Sigma_{\OWSG}^*$ satisfies correctness regardless of $\Sigma_{\OWSG}$.
    \item $\Sigma_{\OWSG}^*$ satisfies security as long as $\Sigma_{\OWSG}$ satisfies correctness and security.
\end{itemize}
Once we have obtained such a transformation, we can construct a robust OWSG combiner $\RobComb.\cM_{\OWSG}$ as follows.
Given two OWSGs candidates $\Sigma_{\OWSG}[1]$ and $\Sigma_{\OWSG}[2]$, our robust OWSG combiner $\RobComb.\cM_{\OWSG}$ first transforms them into $\Sigma_{\OWSG}[1]^*$ and $\Sigma_{\OWSG}[2]^*$, respectively, and then outputs $\Comb.\Sigma_{\OWSG}$ which runs $\Sigma_{\OWSG}[1]^*$ and $\Sigma_{\OWSG}[2]^*$ in parallel.
$\Comb.\Sigma_{\OWSG}$ satisfies correctness because $\Sigma_{\OWSG}[1]^*$ and $\Sigma_{\OWSG}[2]^*$ satisfies correctness no matter what $\Sigma_{\OWSG}[1]$ and $\Sigma_{\OWSG}[2]$ are.
$\Comb.\Sigma_{\OWSG}$ satisfies security as long as either $\Sigma_{\OWSG}[1]$ or $\Sigma_{\OWSG}[2]$ satisfy correctness and security because either $\Sigma_{\OWSG}[1]^*$ or $\Sigma_{\OWSG}[2]^*$ satisfies security as long as either $\Sigma_{\OWSG}[1]$ or $\Sigma_{\OWSG}[2]$ satisfies correctness and security.

\paragraph{Transform Incorrect Candidate into Correct One.}
Now, we consider how to obtain such a transformation.
In the previous work \cite{EC:HKNRR05}, it was shown that we can transform PKE $\Sigma_{\PKE}$ into $\Sigma_\PKE^*$ that satisfies correctness regardless of $\Sigma_\PKE$ and satisfies security as long as $\Sigma_\PKE$ satisfies correctness and security.
In the same way as \cite{EC:HKNRR05}, we can obtain such transformation for OWSGs.
However, in this work, we take a different approach because the technique by~\cite{EC:HKNRR05} does not work for unclonable encryption
\footnote{The technique we introduce here cannot be applied to public-key quantum money. For public-key quantum money, we apply the technique introduced by \cite{EC:HKNRR05} in order to transform an incorrect candidate into a correct one.
The idea of transformation is first checking the correctness of a public-key quantum money candidate $\Sigma=(\Mint,\Vrfy)$.
If the candidate $\Sigma$ satisfies the correctness, then we amplify the correctness by parallel repetition.
Otherwise, we use the scheme $\Sigma^*=(\Mint^*,\Vrfy^*)$, where $\Vrfy^*$ algorithm always outputs $\top$. 
For details, please see \cref{sec:app_money}.}.

First, we observe that without loss of generality, $\Vrfy(k,\psi)$ can be considered working as follows: It appends $\ket{0 }\bra{0}$ to $\psi$, applies $U_{k}$ to $\psi\otimes \ket{0}\bra{0}$, measures the first qubit of $U_{k} (\psi\otimes \ket{0}\bra{0})U_{k}^{\dagger}$, and outputs the measurement outcome.
Now, we describe $\Sigma_{\OWSG}^*=(\keygen^*,\StateGen^*,\Vrfy^*)$.
$\keygen^*$ is the same as the original $\keygen$.
$\StateGen^*(k)$ first runs $\psi_k\la\StateGen(k)$, then measures the first qubit of $U_{k} (\psi_k\otimes \ket{0}\bra{0})U_{k}^{\dagger}$ in the computational basis, and obtains $b$.
If $b=1$, $\StateGen^*(k)$ rewinds its register and outputs the register as $\psi_k^*$.
Otherwise, output $\psi_k^*=\bot$, where $\bot$ is a special symbol.
$\Vrfy^*(k,\psi)$ first checks the form of $\psi$.
If $\psi=\bot$, $\Vrfy^*(k,\psi)$ outputs $1$.
Otherwise, $\Vrfy^*(k,\psi)$ applies $U_{k}$ to $\psi$, then measures the first qubit of $U_{k}\psi U_{k}^{\dagger}$, and finally outputs the measurement outcome.
We can see that $\Sigma^*$ satisfies correctness.
If $\StateGen^*(k)$ outputs $\psi_k^*=\bot$, then $\Vrfy^*$ always outputs $1$.
On the other hand, if $\psi^*\neq \bot$, then $\StateGen^*(k)$ outputs $\psi_k^*$ with the form $U_k^{\dagger}(\ket{1}\bra{1}\otimes \rho) U_{k}$ for some quantum state $\rho$.
Therefore, $\Vrfy^{*}(k,\psi_k^*)$ outputs $1$ since $U_{k}\psi_k^{*}U_{k}^{\dagger}=\ket{1}\bra{1}\otimes \rho$.
Moreover, we can see that $\Sigma^*$ satisfies security as long as $\Sigma$ satisfies correctness and security.
As long as $\Sigma$ satisfies correctness, if we measure the first qubits of $U_{k}(\psi_k\otimes \ket{0}\bra{0})U_{k}^{\dagger}$ in the computational basis, then the measurement result is $1$ with overwhelming probability, where $k\la\keygen(1^\secp)$ and $\psi_k\la\StateGen(k)$.
This indicates that the measurement does not disturb the quantum state $U_{k}(\psi_k\otimes \ket{0}\bra{0})U_{k}^{\dagger}$ from gentle measurement lemma.
Therefore, $\psi_k^*$ is statistically close to $\psi_k\otimes\ket{0}\bra{0}$ as long as $\Sigma$ satisfies correctness.
In particular, this implies that we can reduce the security of $\Sigma^*$ to that of $\Sigma$ as long as $\Sigma$ satisfies correctness.

\subsection{Robust Combiner for Unclonable Encryption}
In this section, we explain how to obtain a robust combiner for unclonable SKE.
As a corollary, we can obtain a robust combiner for unclonable PKE.
This is because we can construct unclonable PKE from unclonable SKE and PKE with quantum ciphertexts~\cite{Asia:HMNY21,TCC:AK21}, and a robust combiner for PKE with quantum ciphertexts can be constructed in the same way as the classical ciphertexts case~\cite{EC:HKNRR05}.

\paragraph{Definition of Unclonable SKE.}
First of all, we explain the definition of unclonable SKE.
Unclonable SKE $\Sigma_{\unc}$ is the same as standard SKE $\Sigma_{\SKE}$ except that the ciphertext of unclonable SKE is a quantum state and it satisfies unclonable IND-CPA security in addition to standard IND-CPA security.
In unclonable IND-CPA security, the cloning adversary $\cA$ with oracle $\Enc(\sk,\cdot)$ first sends the challenge plaintext $(m_0,m_1)$, then receives a ciphertext $\ct_b$, where $\ct_b\la\Enc(\sk,m_b)$, and finally generates a quantum state $\rho_{\cB,\cC}$ over the $\cB$ and $\cC$ registers.
The adversary $\cB$ (resp. $\cC$) receives the $\cB$ register (resp. the $\cC$ register) and the secret-key $\sk$, and outputs $b_\cB$ (resp. $b_\cC$) which is a guess of $b$.
The unclonable IND-CPA security guarantees that for any QPT adversaries $(\cA,\cB,\cC)$, we have
\begin{align}
    \Pr[b=b_\cB=b_\cC]\leq \frac{1}{2}+\negl(\secp).
\end{align}

\paragraph{Robust Combiner.}
First, we consider the simpler case, where given candidates $\Sigma_{\unc}[1]=(\keygen[1],\Enc[1],\Dec[1])$ and $\Sigma_{\unc}[2]=(\keygen[2],\Enc[2],\Dec[2])$ are promised to satisfy at least correctness.
In that case, a combined unclonable SKE scheme $\Comb.\Sigma_\unc=(\keygen,\Enc,\Dec)$ simply runs $\Sigma_{\unc}[1]$ and $\Sigma_\unc[2]$ by using X-OR secret sharing.
In other words, for encrypting bit $b$, $\Comb.\Sigma_\unc$ first samples $r[1]$ and $r[2]$ such that $r[1]+r[2]=b$, and encrypts $r[1]$ by using $\Sigma_{\unc}[1]$ and $r[2]$ by using $\Sigma_\unc[2]$.
Clearly, $\Comb.\Sigma_\unc$ satisfies correctness and security as long as both $\Sigma_\unc[1]$ and $\Sigma_{\unc}[2]$ satisfy correctness and either $\Sigma_\unc[1]$ or $\Sigma_{\unc}[2]$ satisfies security.

Does the same strategy work for the general setting, where original candidates are not promised to satisfy even correctness?
Unfortunately, the simple X-OR protocol above works only when both $\Sigma_\unc[1]$ and $\Sigma_\unc[2]$ satisfy correctness because $\Comb.\Sigma_\unc$ does not satisfy correctness otherwise.
Our key observation is that given a candidate of unclonable SKE $\Sigma_{\unc}$ we can construct a new candidate $\Sigma_{\unc}^*$ with the following properties:
\begin{itemize}
    \item $\Sigma_{\unc}^*$ satisfies correctness regardless of $\Sigma_{\unc}$.
    \item $\Sigma_{\unc}^*$ satisfies security as long as $\Sigma$ satisfies correctness and security.
\end{itemize}
Once we have obtained such a transformation, we can construct a robust combiner for unclonable SKE as follows.
Given two unclonable SKE candidates $\Sigma_{\unc}[1]$ and $\Sigma_{\unc}[2]$, a robust combiner for unclonable SKE first transforms $\Sigma_{\unc}[1]$ and $\Sigma_{\unc}[2]$ into $\Sigma_{\unc}[1]^*$ and $\Sigma_{\unc}[2]^*$, respectively, and then outputs $\Comb.\Sigma_{\unc}$ which runs $\Sigma_{\unc}[1]^*$ and $\Sigma_{\unc}[2]^*$ by using X-OR secret sharing.
$\Comb.\Sigma_{\unc}$ satisfies correctness because $\Sigma_{\unc}[1]^*$ and $\Sigma_{\unc}[2]^*$ satisfy correctness no matter what $\Sigma_{\unc}[1]$ and $\Sigma_{\unc}[2]$ are.
Moreover, $\Comb.\Sigma_{\unc}$ satisfies security as long as either $\Sigma_{\unc}[1]$ or $\Sigma_{\unc}[2]$ satisfies correctness and security.
This is because either $\Sigma_{\unc}^*[1]$ or $\Sigma_{\unc}[2]^*$ satisfies security as long as either $\Sigma_{\unc}[1]$ or $\Sigma_{\unc}[2]$ satisfies correctness and security.

\paragraph{Transform Incorrect Candidate into Correct One.}
Now, we consider how to obtain such a transformation.
It is known that we can obtain such a transformation for PKE~\cite{EC:HKNRR05}.
In their technique, they use parallel repetition to amplify correctness.
We emphasize that we cannot apply their technique for unclonable encryption because correctness amplification via parallel repetition does not work for unclonable encryption.
Therefore, we take a different approach, whose idea is the same as OWSGs.
Without loss of generality, we can assume that $\Dec(\sk,\ct)$ first appends $\ket{0 }\bra{0}$ to $\ct$, applies $U_{\sk}$ to $\ct\otimes \ket{0}\bra{0}$, measures the first $\abs{m}$-bit of $U_{\sk} (\ct\otimes \ket{0}\bra{0})U_{\sk}^{\dagger}$, and outputs the measurement outcome.
Now, we describe $\Sigma_{\unc}^*=(\keygen^*,\Enc^*,\Dec^*)$.
$\keygen^*$ is the same as the original $\keygen$.
$\Enc^*(\sk,m)$ first runs $\ct\la\Enc(\sk,m)$, then measures the first $\abs{m}$-bit of $U_{\sk} (\ct\otimes \ket{0}\bra{0})U_{\sk}^{\dagger}$ in the computational basis, obtains $m^*$, and checks whether $m=m^*$.
If $m=m^*$, $\Enc^*(\sk,\ct)$ rewinds its register and outputs the register as the quantum ciphertext $\ct^*$.
Otherwise, output $\ct^*=(\bot,m)$, where $\bot$ is a special symbol.
$\Dec^*(\sk,\ct^*)$ first checks the form of $\ct^*$, and outputs $m$ if $\ct^*$ is of the form $(\bot,m)$.
Otherwise, $\Dec^*(\sk,\ct^*)$ applies $U_{\sk}$ to $\ct^*$, and outputs the measurement outcome of first $\abs{m}$-qubits of $U_{\sk}\ct^*U_{\sk}^{\dagger}$.
Clearly, the new construction $\Sigma_\unc^*$ satisfies correctness in the same reason as OWSG.
Furthermore, $\Sigma_\unc^*$ satisfies security as long as $\Sigma_\unc$ satisfies correctness and security.
This is because $\ct^*$ is statistically close to $\ct\otimes\ket{0}\bra{0}$ as long as $\Sigma_\unc$ satisfies correctness, and thus we can reduce the security of $\Sigma_\unc^*$ to that of $\Sigma_\unc$.

\subsection{Robust Combiner for Quantum Bit Commitment}
\paragraph{Definition of Quantum Bit Commitment.}
In the following, we consider a robust combiner for quantum bit commitment.
In this work, we consider a canonical quantum bit commitment.
Any quantum bit commitment can be written in the following canonical form~\cite{Asia:Yan22}.
A canonical quantum bit commitment scheme is a pair of unitaries $(Q_0,Q_1)$ acting on the registers $\mathbf{C}$ called the commitment register and $\mathbf{R}$ called the reveal register, and works as follows.
\begin{itemize}
    \item[{\bf Commit Phase}:] A sender runs $Q_b\ket{0}_{\mathbf{C},\mathbf{R}}$ and sends the $\mathbf{C}$ to a receiver for committing a bit $b\in\bit$.
    \item[{\bf Reveal Phase}:]
    For revealing the committed bit $b$, the sender sends $b$ and the $\mathbf{R}$ register to the receiver. 
    The receiver applies $Q_b^{\dagger}$ to the $\mathbf{C}$ and $\mathbf{R}$ register and measures both registers in the computational basis.
    The receiver accepts if the measurement outcomes are all $0$, and rejects otherwise.
\end{itemize}
We require that a canonical quantum bit commitment satisfies hiding and binding. 
The computational (resp. statistical) hiding requires that no quantum polynomial-time (resp. unbounded) adversaries distinguish  $Q_0\ket{0}_{\mathbf{C},\mathbf{R}}$ from $Q_1\ket{0}_{\mathbf{C},\mathbf{R}}$ without touching the $\mathbf{R}$ register with non-negligible probability.

The binding requires that no adversaries can map an honestly generated quantum bit commitment of $0$ (i.e. $Q_0\ket{0}_{\mathbf{C,R}}$) to that of $1$ (i.e. $Q_1\ket{0}_{\mathbf{C,R}}$) without touching $\mathbf{C}$ registers.
More formally, computational (resp. statistical) binding requires that for any quantum polynomial-time (resp. unbounded) unitary $U_{\mathbf{R},\mathbf{Z}}$ acting on the $\mathbf{R}$ and $\mathbf{Z}$ register and any quantum state $\ket{\tau}_{\mathbf{Z}}$ on $\mathbf{Z}$ register, we have
\begin{align}
    \norm{(Q_1\ket{0}\bra{0}Q_1^{\dagger})_{\mathbf{C,R}}(I_{\mathbf{C}}\otimes U_{\mathbf{R,Z}})(Q_0\ket{0}_{\mathbf{C},\mathbf{R}}\ket{\tau}_{\mathbf{Z}})}\leq \negl(\secp).
\end{align}

It was shown that we can change the flavor of quantum bit commitment~\cite{Asia:Yan22,EC:HMY23}.
More formally, if we have a canonical quantum bit commitment $(Q_0,Q_1)$ that satisfies $X$-hiding and $Y$-binding, then we can construct a canonical quantum bit commitment $(\widetilde{Q_0},\widetilde{Q_1})$ that satisfies $X$-binding and $Y$-hiding for $X,Y\in$ $\{$statistical, computational$\}$.

\paragraph{Robust Combiner.}
First, let us clarify our final goal.
Given two candidates of canonical quantum bit commitments $(Q_0[1],Q_1[1])$ and $(Q_0[2],Q_1[2])$, our robust combiner $\RobComb.\cM_\Commit$ generates a new candidate $(\Comb.Q_0,\Comb.Q_1)$ that satisfies hiding and binding as long as either $(Q_0[1],Q_1[1])$ or $(Q_0[2],Q_1[2])$ satisfies hiding and binding.
More formally, our robust combiner $\RobComb.\cM_\Commit$ outputs $(\Comb.Q_0,\Comb.Q_1)$ with the following properties:
\begin{itemize}
    \item $(\Comb.Q_0,\Comb.Q_1)$ satisfies statistical binding regardless of $(Q_0[1],Q_1[1])$ and $(Q_0[2],Q_1[2])$.
    \item $(\Comb.Q_0,\Comb.Q_1)$ satisfies computational hiding as long as either $(Q_0[1],Q_1[1])$ or $(Q_{0}[2],Q_1[2])$ satisfies computational hiding and computational binding.
\end{itemize}

To achieve this final goal, let us consider the following simpler goal first, where both candidates $(Q_0[1], Q_1[1])$ and $(Q_0[2], Q_1[2])$ satisfy at least statistical binding.
More formally, given candidates $(Q_0[1],Q_1[1])$ and $(Q_0[2],Q_1[2])$, we consider constructing a new candidate $(\Comb.Q_0,\Comb.Q_1)$ with the following properties:
\begin{itemize}
    \item $(\Comb.Q_0,\Comb.Q_1)$ satisfies statistical binding as long as both $(Q_0[1],Q_1[1])$ and $(Q_0[2],Q_1[2])$ satisfies statistical binding.
    \item $(\Comb.Q_0,\Comb.Q_1)$ satisfies computational hiding as long as either $(Q_0[1],Q_1[1])$ or $(Q_0[2],Q_1[2])$ satisfies computational hiding.
\end{itemize}

We can construct such $(\Comb.Q_0,\Comb.Q_1)$ by simply using X-OR secret sharing.
More formally, for $b\in\bit$, $\Comb.Q_b$ first samples $r[1]$ and $r[2]$ conditioned on $r[1]+r[2]=b$, and then commits $r[1]$ by using $(Q_0[1],Q_1[1])$ and commits $r[2]$ by using $(Q_0[2],Q_1[2])$.
Our construction satisfies statistical binding as long as both $(Q_0[1],Q_1[1])$ and $(Q_0[2],Q_1[2])$ satisfy statistical binding.
The intuitive reason is that the adversary of $(\Comb.Q_0,\Comb.Q_1)$ needs to change $r[1]$ or $r[2]$ after sending the commitment register to break binding of $(\Comb.Q_0,\Comb.Q_1)$, but the adversary cannot do this because both $(Q_0[1],Q_1[1])$ and $(Q_0[2],Q_1[2])$ satisfy statistical binding.
Furthermore, $(\Comb.Q_0,\Comb.Q_1)$ satisfies computational hiding as long as either $(Q_0[1],Q_1[1])$ or $(Q_0[2],Q_1[2])$ satisfies computational hiding.
The intuitive reason is that the adversary of $(\Comb.Q_0,\Comb.Q_1)$ needs to obtain both $r[1]$ and $r[2]$ from the commitment register of $(Q_0[1],Q_1[1])$ and $(Q_0[2],Q_1[2])$, but the adversary cannot do this because either $(Q_0[1],Q_1[1])$ and $(Q_0[2],Q_1[2])$ satisfies computational hiding.

Does the same strategy work for a robust quantum bit commitment combiner $\RobComb.\cM_{\Commit}$?
Unfortunately, the simple X-OR protocol above works only when both $(Q_0[1],Q_1[1])$ and $(Q_0[2],Q_1[2])$ satisfy statistical binding because $(\Comb.Q_0,\Comb.Q_1)$ does not satisfy statistical binding otherwise.
Our key observation is that, given a candidate of canonical quantum bit commitment $(Q_0,Q_1)$, we can construct a new candidate $ (Q_0^*,Q_1^*)$ that satisfies at least statistical binding regardless of $(Q_0,Q_1)$.
More formally, we can construct $(Q_0^*,Q_1^*)$ with the following properties:
\begin{itemize}
    \item $(Q_0^*,Q_1^*)$ satisfies statistical binding regardless of $(Q_0,Q_1)$.
    \item $(Q_0^*,Q_1^*)$ satisfies computational hiding if $(Q_0,Q_1)$ satisfies computational hiding and computational binding.
\end{itemize}
Once we have obtained such a transformation, we can construct a robust quantum bit commitment combiner $\RobComb.\cM_\Commit$.
Given two candidates of canonical quantum bit commitment $(Q_0[1],Q_1[1])$ and $(Q_0[2],Q_1[2])$, $\RobComb.\cM_\Commit$ first transforms $(Q_0[1],Q_1[1])$ and $(Q_0[2],Q_1[2])$ into $(Q_0[1]^*,Q_1[1]^*)$ and $(Q_0[2]^*,Q_1[2]^*)$, respectively and then outputs $(\Comb.Q_0,\Comb.Q_1)$, which runs $(Q_0[1]^*,Q_1[1]^*)$ and $(Q_0[2]^*,Q_1[2]^*)$ by using X-OR secret sharing.
Clearly, $(\Comb.Q_0,\Comb.Q_1)$ satisfies statistical binding.
Moreover, $(\Comb.Q_0,\Comb.Q_1)$ satisfies computational hiding as long as either $(Q_0[1],Q_1[1])$ or $(Q_0[2],Q_1[2])$ satisfies computational hiding and computational binding.

\paragraph{Transform Candidate without Statistical Binding into One with Statistical Binding.}
Now, we consider how to obtain such a transformation.
Our first observation is that either $(Q_0,Q_1)$ or $(\widetilde{Q_0},\widetilde{Q_1})$, which is the flavor conversion of $(Q_0, Q_1)$ obtained by \cite{EC:HMY23}, satisfies statistical binding in a possibly weak sense.
To see this let us denote $\rho_b\seteq\Tr_{\mathbf{R}}(Q_b\ket{0}_{\mathbf{C,R}})$.
Then, there exists some constant $f$ such that
\begin{align}
    F(\rho_0,\rho_1)=f,
\end{align}
where $F(\rho_0,\rho_1)$ is the fidelity between $\rho_0$ and $\rho_1$.
If $f$ is small, then $(Q_0, Q_1)$ satisfies statistical binding in a possibly weak sense from Uhlmann's theorem.
On the other hand, if $f$ is large, then $(Q_0,Q_1)$ does not satisfy statistical binding, but $(\widetilde{Q_0},\widetilde{Q_1})$ satisfies statistical binding instead.
This is because if $f$ is large, then $(Q_0,Q_1)$ satisfies statistical hiding, and thus $(\widetilde{Q_0},\widetilde{Q_1})$ satisfies statistical binding.
Therefore, either $(Q_0,Q_1)$ or $(\widetilde{Q_0},\widetilde{Q_1})$ satisfies statistical binding in a possibly weak sense regardless of $(Q_0,Q_1)$.
Furthermore, we observe that such a possibly weak binding property can be amplified to a strong one by parallel repetition.

Based on these observations, we construct our transformation.
Given a candidate of canonical quantum bit commitment $(Q_0,Q_1)$, our transformation outputs a new candidate $(Q_0^*,Q_1^*)$ working as follows.
\begin{itemize}
    \item If we write $\mathbf{C}$ and $\mathbf{R}$ to mean the commitment register and the reveal register of $(Q_0,Q_1)$, and write
    $\mathbf{\widetilde{C}}$ and $\mathbf{\widetilde{R}}$ to mean the commitment and the reveal register of $(\widetilde{Q_0},\widetilde{Q_1})$, then
    the commitment register $\mathbf{C^*}$ of $(Q_0^*,Q_1^*)$ is $(\mathbf{C}^{\otimes \secp},\mathbf{\widetilde{C}}^{\otimes \secp})$, and the reveal register $\mathbf{R^*}$ of $(Q_0^*,Q_1^*)$ is $(\mathbf{R}^{\otimes \secp},\mathbf{\widetilde{R}}^{\otimes \secp})$.
    \item For $b\in\bit$, $Q_b^*$ works as follows:
    \begin{align}
        Q_b^*\seteq (Q_b\otimes \widetilde{Q_b})^{\otimes \secp}.
    \end{align}
    Note that we have 
    \begin{align}
        Q_b^*\ket{0}_{\mathbf{C^*,R^*}}= (Q_b\ket{0}_{\mathbf{C,R}})^{\otimes \secp}\otimes (\widetilde{Q_b}\ket{0}_{\mathbf{\widetilde{C},\widetilde{R}}})^{\otimes \secp} .
    \end{align}
\end{itemize}

We can see that $(Q_0^*,Q_1^*)$ satisfies statistical binding regardless of $(Q_0,Q_1)$.
If we write $\rho_b\seteq\Tr_{\mathbf{R}}(Q_b\ket{0}_{\mathbf{C,R}})$, there exists some constant $0\leq f\leq 1$ such that 
\begin{align}
    F(\rho_0,\rho_1)=f.
\end{align}
If we write $\widetilde{\rho_b}\seteq\Tr_{\mathbf{\widetilde{R}}}(\widetilde{Q_b}\ket{0}_{\mathbf{\widetilde{C},\widetilde{R}}})$, then we can show that
\begin{align}
    F(\widetilde{\rho_0},\widetilde{\rho_1})\leq (1-f)^{1/2}
\end{align}
by using the technique by \cite{EC:HMY23}.
Therefore, if we write $\rho_b^*\seteq \Tr_{\mathbf{R^*}}(Q_b^*\ket{0}_{\mathbf{C^*,R^*}})$, we have
\begin{align}
    F(\rho_0^*,\rho_1^*)= F( (\rho_0\otimes\widetilde{\rho_0})^{\otimes \secp},(\rho_1\otimes\widetilde{\rho_1})^{\otimes \secp} )\leq F(\rho_0,\rho_1)^\secp F(\widetilde{\rho_0},\widetilde{\rho_1})^{\secp}\leq f^\secp(1-f)^{\secp/2}\leq 2^{-\secp/2}.
\end{align}
This implies that 
$(Q_0^*,Q_1^*)$ satisfies statistical binding regardless of $(Q_0,Q_1)$ from Uhlmann's Theorem.

Moreover, we can see that $(Q_0^*,Q_1^*)$ satisfies computational hiding as long as $(Q_0,Q_1)$ satisfies computational hiding and computational binding.
The hiding QPT adversary of $(Q_0^*,Q_1^*)$ needs to obtain $b$ from $\rho_b^*=(\rho_b\otimes \widetilde{\rho_b})^{\otimes \secp}$.
For that, the adversary needs to obtain $b$ from $\rho_b$ or $\widetilde{\rho_b}$.
Because $(Q_0,Q_1)$ satisfies computational hiding, the QPT adversary cannot obtain $b$ from $\rho_b$.
Furthermore, $(\widetilde{Q_0},\widetilde{Q_1})$ also satisfies computational hiding because $(\widetilde{Q_0},\widetilde{Q_1})$ is a flavor conversion of $(Q_0,Q_1)$.
Therefore, the QPT adversary cannot obtain $b$ from $\widetilde{\rho_b}$.

\subsection{Universal Constructions}
Let us recall the definition of universal construction.
A universal construction for a primitive $P$ is an explicit construction of $P$, which satisfies correctness and security as long as $P$ exists.
In this section, we explain how to provide universal constructions via robust combiners.
In particular, we explain how to construct a universal construction for OWSGs by using a robust OWSG combiner $\RobComb.\cM_{\OWSG}$.
We can give universal constructions for other cryptographic primitives in the same way.

In a nutshell, the idea of universal construction via robust combiner~\cite{EC:HKNRR05} is to think of all descriptions of algorithms as OWSG candidates and combine them.
For a set of classical Turing machines $\cM=(x,y,z)$, we write $(\keygen[x],\StateGen[y],\Vrfy[z])$ to mean a OWSG candidate described by $(x,y,z)$.
For simplicity, we assume that $(\keygen[x],\StateGen[y],\Vrfy[z])$ are efficient for all $x,y,z\in\N$.
The universal construction $(\keygen_{\mathsf{Univ}}(1^\secp),\StateGen_{\mathsf{Univ}}(k),\Vrfy_{\mathsf{Univ}}(k,\psi_k))$ works as follows:

\begin{itemize}
    \item $\keygen_{\mathsf{Univ}}(1^\secp)$ first runs \begin{align}(\keygen_\secp,\StateGen_\secp,\Vrfy_\secp)\la\RobComb.\cM_{\OWSG}(\{\keygen[x],\StateGen[y],\Vrfy[z]\}_{x,y,z\in[\secp]}),
    \end{align}
    where $[\secp]=\{1,\cdots,\secp\}$.
    Then, $\keygen_{\mathsf{Univ}}(1^\secp)$ runs $k\la\keygen_\secp(1^\secp)$, and outputs $k$.
    \item $\StateGen_{\mathsf{Univ}}(k)$ runs $\psi_k\la\StateGen_{\secp}(1^\secp,k)$, and outputs $\psi_k$.
    \item $\Vrfy_{\mathsf{Univ}}(k,\psi_k)$ runs $\Vrfy_{\secp}(1^\secp,k,\psi_k)$, and outputs its output.
\end{itemize}
Assume that there exist OWSGs, then there also exists a set of classical Turing machine $\cM^*=(x^*,y^*,z^*)$ such that the OWSG scheme $(\keygen[x^*],\StateGen[y^*],\Vrfy[z^*])$ satisfies correctness and security.
For all sufficiently large $\secp\in\N$, one of $\{\keygen[x],\StateGen[y],\Vrfy[z]\}_{x,y,z\in[\secp]}$ includes a correct and secure OWSG scheme $(\keygen[x^*],\StateGen[y^*],\Vrfy[z^*])$ as long as OWSGs exist.
Therefore, $(\keygen_\secp,\StateGen_\secp,\Vrfy_\secp)$ satisfies correctness and security for all sufficiently large $\secp\in\N$ as long as OWSGs exist. Because $(\keygen_{\mathsf{Univ}},\StateGen_{\mathsf{Univ}},\Vrfy_{\mathsf{Univ}})$ emulates $(\keygen_\secp,\StateGen_\secp,\Vrfy_\secp)$, it also satisfies correctness and security

\subsection{Universal Plaintext Expansion for Unclonable Encryption}
We give another universal construction for one-time unclonable SKE assuming decomposable quantum randomized encoding whose construction is inspired by \cite{Eprint:WW23}.
Although we additionally use a decomposable quantum randomized encoding for this construction, we can expand the plaintext of one-time unclonable SKE.
Note that it was an open problem to expand the plaintext of unclonable encryption since a standard transformation via bit-wise encryption does not work as pointed out in \cite{C:AKLLZ23}.

First, let us recall the decomposable quantum randomized encoding $\Sigma_\RE=\RE.(\Enc,\Dec)$ given in \cite{Stoc:BY22}.
In their decomposable quantum randomized encoding, $\RE.\Enc$ takes as input a quantum circuit $F$,  $\secp$-length possibly quantum input $q$ and $\secp$-length classical input $x$, and outputs $\widehat{F(q,x)}$. 
Let $q[i]$ and $x[i]$ be the $i$-th qubit and bit of $q$ and $x$, respectively.
Decomposability guarantees that $\widehat{F(q,x)}$ can be separated into the offline encoding part $\widehat{F}_{\mathsf{off}}$ and online encoding parts $\left(\{\lab_i(q[i])\}_{i\in\{1,\cdots,\secp\}},\{\lab_{i+\secp}(x[i])\}_{i\in\{1,\cdots,\secp\}} \right)$ as follows:
\begin{align}
    \widehat{F(q,x)}\seteq
    \left(
    \widehat{F}_{\mathsf{off}},\lab_1(q[1]),\cdots, \lab_\secp(q[\secp]), \lab_{\secp+1}(x[1]),\cdots,\lab_{2\secp}(x[\secp])
    \right),
\end{align}
where $\widehat{F}_{\mathsf{off}}$ does not depend on $q$ and $x$, $\lab_i(q[i])$ depends on only $q[i]$ for $i\in[\secp]$ and $\lab_{i+\secp}(x[i])$ depends on only $x[i]$ for $i\in[\secp]$.
$\RE.\Dec$ takes as input $\widehat{F(q,x)}$ and outputs $F(q,x)$.
The security roughly guarantees that for any quantum circuits $F_1,F_2$ with the same size, and any quantum and classical inputs $(\{q_1,x_1\}, \{q_2,x_2\})$ such that $F_1(q_1,x_1)=F_2(q_2,x_2)$,
$
    \widehat{F_1(q_1,x_1)}
$
is computationally indistinguishable from
$
    \widehat{F_2(q_2,x_2)}.
$
 
Now, we describe our one-time unclonable SKE $\Sigma_{\Univ}=(\keygen_{\Univ},\Enc_{\Univ},\Dec_{\Univ})$:
\begin{itemize}
    \item[$\keygen_{\Univ}(1^\secp)$:]Our key generation algorithm $\keygen_{\Univ}(1^\secp)$ first samples $x\la\bit^{\secp}$.
    Then, it samples $R[i]\la\bit^{\ell(\secp)}$ for $i\in[\secp]$, and outputs $\sk\seteq (x,\{R[i]\}_{i\in[\secp]})$.
    Here, $\ell(\secp)$ is the size of online encoding of $\RE.\Enc$.
    \item[$\Enc_{\Univ}(\sk,m)$:]
    Our encryption algorithm $\Enc_{\Univ}(\sk,m)$ first generates a quantum circuit $C[m]$ that outputs $m$ for any inputs, where
    the quantum circuit is padded to an appropriate size, which we will specify later.
    Then, $\Enc_{\Univ}(\sk,m)$ computes $\widehat{C[m]}_{\mathsf{off}}$, which is the offline encoding of $C[m]$.
    Next, it computes $\lab_{i}(0)$ for $i\in[\secp]$
    and $\lab_{\secp+i}(b)$ for $i\in[\secp]$ and $b\in\bit$.
    Finally, it samples $S[i]\la\bit^{\secp}$, and computes $\Lab.\ct[i,x[i]]=R[i]+ \lab_{\secp+i}(x[i])$ and
    $\Lab.\ct[i,1-x[i]]=S[i]+ \lab_{\secp+i}(1-x[i])$ for all $i\in[\secp]$.
    The ciphertext of $\Enc_{\Univ}(\sk,m)$ is 
    \begin{align}
    \widehat{C[m]}_{\mathsf{off}},\{\lab_i(0)\}_{i\in[\secp]}, \{\Lab.\ct[i,b]\}_{i\in[\secp],b\in\bit}.
    \end{align}
    \item[$\Dec_{\Univ}(\sk,\ct)$:]
    Our decryption algorithm $\Dec_{\Univ}(\sk,\ct)$ works as follows.
    First, let $\sk=(x,\{R[i]\}_{i\in[\secp]})$ and $\ct=
    \left(\widehat{C[m]}_{\mathsf{off}},\{\lab_i(0)\}_{i\in[\secp]}, \{\Lab.\ct[i,b]\}_{i\in[\secp],b\in\bit}\right)$.
    $\Dec_{\Univ}(\sk,\ct)$ first computes $ $
    $\lab_{\secp+i}(x[i])=R[i]+\Lab.\ct[i,x[i]]$ for all $i\in[\secp]$, and runs
    $\RE.\Dec(\widehat{C[m]}_{\mathsf{off}}, \{\lab_i(0)\}_{i\in[\secp]}, \{\lab_{i+\secp}(x[i])\}_{i\in [\secp]})$. 
\end{itemize}
Clearly, our encryption algorithm can encrypt arbitrary-length plaintext.
We can see that our construction satisfies correctness.
More formally, $\Dec_{\Univ}(\sk,\ct_m)$ outputs $m$ with high probability if $\sk\la\keygen_{\Univ}(1^\secp)$ and $\ct_m\la\Enc_{\Univ}(\sk,m)$.
From our construction, $\Dec_{\Univ}(\sk,\ct_m)$ outputs the output of $\RE.\Dec(\widehat{C[m]}_{\mathsf{off}}, \{\lab_i(0)\}_{i\in[\secp]},\allowbreak \{\lab_{i+\secp}(x[i])\}_{i\in[\secp]})$, where $\Big(\widehat{C[m]}_{\mathsf{off}}, \{\lab_i(0)\}_{i\in[\secp]},\allowbreak \{\lab_{i+\secp}(x[i])\}_{i\in[\secp]}\Big)\la\RE.\Enc(C,0^{\secp},x)$.
From the correctness of decomposable quantum randomized encoding, $\RE.\Dec(\widehat{C[m]}_{\mathsf{off}}, \{\lab_i(0)\}_{i\in[\secp]}, \{\lab_{i+\secp}(x[i])\}_{i\in[\secp]})$ outputs $C[m](0^{\secp},x)$, which is equal to $m$.

Furthermore, our construction $\Sigma_{\Univ}$ satisfies unclonable IND-CPA security as long as the underlying decomposable quantum randomized encoding $\Sigma_{\RE}$ satisfies security and there exists a one-time unclonable SKE for single-bit plaintexts.
To see this, we introduce some notations and observations.
We write $\Sigma_{\unc}=\Unc.(\keygen,\Enc,\Dec)$ to mean a one-time unclonable SKE for single-bit plaintexts, which we assume to exist.
Without loss of generality, we can assume that the secret key $\sk$ generated by $\Unc.\keygen(1^\secp)$ is uniformly randomly sampled and $\abs{\sk}=\abs{\ct}=\secp$ for all security parameters $\secp$.
Moreover, we can assume that for a security parameter $\secp$, $\Unc.\Dec(\sk,\ct)$ is a quantum algorithm that runs some quantum circuit $\Unc.\Dec_{\secp}$ on $\ct$ and $\sk$, and outputs its output.
We introduce a quantum circuit $D_\secp[m_0,m_1]$ that takes as input $\ct$ and $\sk$, and runs the quantum circuit $\Unc.\Dec_{\secp}$ on $\ct$ and $\sk$, obtains $b$ and outputs $m_b$.
The size of $C[m]$ is padded so that its size is equal to $D_\secp[m_0,m_1]$.

Now, we can see that our construction $\Sigma_{\Univ}$ satisfies one-time unclonable IND-CPA security.
In the first step of the proof,
 we switch the following real ciphertext for message $m_b$
\begin{align}
    \ct_b=\left(\widehat{C[m_b]}_{\mathsf{off}},\{\lab_i(0)\}_{i\in[\secp]}, \{\Lab.\ct[i,\beta]\}_{i\in[\secp],\beta\in\bit}\right)
\end{align}
to the following modified ciphertext
\begin{align}
    \widetilde{\ct_b}= \left(\widehat{D[m_0,m_1]}_{\mathsf{off}}, \{\lab_i(\unc.\ct_b[i])\}_{i\in[\secp]}, \{\Lab.\ct[i,\beta]\}_{i\in[\secp],\beta\in\bit}\right),
\end{align}
where $\unc.\ct_b\la\Unc.\Enc(x,b)$ and $\unc.\ct_b[i]$ is the $i$-th qubit of $\unc.\ct_b$ and $x\la\bit^{\secp}$.
This change does not affect the output of the security experiment because $\Sigma_\RE $ satisfies security and we have
\begin{align}
    D[m_0,m_1](\unc.\ct_b,x)=C[m_b](0^{\secp},x)=m_b.
\end{align}
In the next step, we can reduce the security of our construction $\Sigma_{\Univ}$ to that of one-time unclonable SKE for single-bit plaintexts $\Sigma_{\unc}$.
This is because the adversary $(\cA,\cB,\cC)$ of $\Sigma_\unc$ can simulate the challenger of $\Sigma_\Univ$ sicne $(\cA,\cB,\cC)$ can simulate $\widetilde{\ct_b}$ by using $\unc.\ct_b$.

\section{Preliminaries}\label{sec:pre}
\subsection{Notations}

Here we introduce basic notations we will use in this paper.
$x\la X$ denotes selecting an element $x$ from a finite set $X$ uniformly at random, and $y \la \cA(x)$ denotes assigning to $y$ the output of a quantum or probabilistic or deterministic algorithm $\cA$ on an input $x$.
When we explicitly write that $\cA$  uses randomness $r$, we write $y \la \cA(x;r)$.
Let $[n]\seteq \{1,\cdots,n\}$.
For $x\in\bit^n$ and $i\in[n]$, $x_i$ and $x[i]$ are the $i$-th bit value of $x$.
For an $n$-qubit state $\rho$ and $i\in[n]$, we write $\rho_i$ and $\rho[i]$ to mean a quantum state that traces out all states other than the $i$-th qubit of $\rho$.
QPT stands for quantum polynomial time. 
A function $f : \N \ra \R$ is a negligible function if, for any constant $c$, there exists $\secp_0 \in \N$ such that for any $\secp>\secp_0$, $f(\secp) < 1/\secp^c$.
We write $f(\secp) \leq \negl(\secp)$ to denote $f(\secp)$ being a negligible function.

For simplicity, we often write $\ket{0}$ to mean $\ket{0\cdots 0}$.
For any two quantum states $\rho_1$ and $\rho_2$, $F(\rho_1,\rho_2)$ is the fidelity between them, and $\mathsf{TD}(\rho_1,\rho_2)$ is the trace distance between them.

For a quantum algorithm $\cA$, and quantum states $\rho$ and $\sigma$, we say that $\cA$ distinguishes $\rho$ from $\sigma$ with advantage $\Delta$ if
\begin{align}
    \abs{\Pr[1\la\cA(\rho)]-\Pr[1\la\cA(\sigma)]}=\Delta.
\end{align}
We say that $\rho$ is $c$-computationally indistinguishable (resp. $c$-statistically indistinguishable) from $\sigma$ if no QPT algorithms (resp. unbounded algorithms) can distinguish $\rho$ from $\sigma$ with advantage greater than $c$.

\paragraph{Quantum Circuits}
For convenience,
we assume that all quantum circuits use gates from the universal gate set $\{I,H, CNOT, T\}$. 
A unitary quantum circuit is one that consists only of gates from this gate set.
A general quantum circuit is a quantum circuit that can additionally have non-unitary gates that (a) introduce new qubits initialized in the zero state, (b) trace them out, or (c) measure them in the computational basis.
We say that a general quantum circuit has size $s$ if the total number of gates is at most $s$.

\begin{definition}[Uniform Quantum Polynomial Time Algorithm]
    We say that an algorithm $\cA$ is a uniform quantum polynomial time (QPT) algorithm if $\cA$ works as follows:
    For any pair of classical and quantum input $(x,\rho)$, $\cA$ runs some deterministic classical polynomial-time Turing machine $\cM$ on $(x,\abs{\rho})$, and obtains a general quantum circuit $C_{x,\abs{\rho}}$ within $\poly(\abs{x},\abs{\rho})$ steps, and outputs the output of $C_{x,\abs{\rho}}(\rho)$.

    We say that the sequence of unitaries $\{U_\secp\}_{\secp\in\N}$ is a uniform QPT unitary if $U_\secp$ is the output of $\cM(1^\secp)$ for all $\secp\in\N$, where $\cM$ is a classical Turing machine that halts within $\poly(\secp)$ steps for any input $\secp\in\N$.
\end{definition}

\begin{remark}
    We consider many algorithms as uniform QPT algorithms, and thus an algorithm $\mathsf{Alg}$ is represented as a classical Turing machine that generates general quantum circuits.
    If $x\in\bit^*$ is a classical Turing machine that represents $\mathsf{Alg}$, then we sometimes explicitly write $\mathsf{Alg}[x]$.
\end{remark}

\begin{definition}[Non-Uniform Quantum Polynomial Time Algorithm]
    We say that an algorithm $\cA$ is a non-uniform quantum polynomial time algorithm if $\cA$ works as follows:
    For any pair of classical and quantum input $(x,\rho)$, $\cA$ runs a general quantum circuit $C$ with size $\poly(\abs{x},\abs{\rho})$ on $(x,\rho)$ and a quantum advice $\psi$ with size $\poly(\abs{x},\abs{\rho})$, and outputs its output.
\end{definition}

\begin{remark}
    Throughout this work, we model adversaries as non-uniform QPT algorithms.
    Note that all results except for \cref{Sec:expanstion_UE} hold in the uniform adversary setting with appropriate modifications.
\end{remark}

\paragraph{Other Notions:}

\begin{lemma}[Gentle Measurement Lemma]\label{lem:gentle}
 Let $\rho$ be a mixed state, and let $E$ be a measurement operator.
 Suppose that $\Tr(E\rho)\geq 1-\epsilon$, where $0<\epsilon\leq 1$.
 Then, the post-measurement quantum state $\rho'\seteq \frac{\sqrt{E}\rho\sqrt{E}}{\Tr(E\rho)}$ satisfies:
 \begin{align}
     \norm{\rho-\rho'}_1\leq 2\sqrt{\epsilon}.
 \end{align}
\end{lemma}

\begin{theorem}[Uhlmann's Theorem]\label{lem:uhlmann}
    Let $\ket{\psi}_{\mathbf{C,R}}$ and $\ket{\phi}_{\mathbf{C,R}}$ be quantum states over the $\mathbf{C}$ and $\mathbf{R}$ registers.
    Then, for any unitary $U_{\mathbf{R}}$ acting over $\mathbf{R}$ register, we have
    \begin{align}
        F(\rho,\sigma)=\abs{\bra{\psi}_{\mathbf{C,R}}(I_{\mathbf{C}}\otimes  U_{\mathbf{R}})\ket{\phi}_{\mathbf{C,R}}}^2,
    \end{align}
    where $\rho=\Tr_{\mathbf{R}}(\ket{\psi}\bra{\psi}_{\mathbf{C,R}})$ and $\sigma=\Tr_{\mathbf{R}}(\ket{\phi}\bra{\phi}_{\mathbf{C,R}})$.
\end{theorem}

\subsection{Cryptographic Tools}
In this section, we introduce cryptographic tools which we will use.
\paragraph{One-Way State Generators.}
In this work, we consider the mixed-state output version of one-way state generators introduced in \cite{Eprint:MY22}.

\begin{definition}[One-way state generators(OWSGs)]\label{def:OWSG}
A one-way state generator (OWSG) candidate is a set of algorithms $\Sigma\seteq(\keygen,\StateGen,\Vrfy)$ such that:
\begin{itemize}
    \item[$\keygen(1^\secp)$:] It takes a security parameter $1^\secp$, and outputs a classical string $k$.
    \item[$\StateGen(1^\secp,k)$:] It takes a security parameter $1^\secp$ and $k$, and outputs a quantum state $\psi_k$.
    \item[$\Vrfy(1^\secp,k,\psi_k)$:]It takes a security parameter $1^\secp$, $k$ and $\psi_k$, and outputs $\top$ or $\bot$. 
\end{itemize}
We say that a candidate $\Sigma$ is a OWSG scheme if $\Sigma$ satisfies the following efficiency, correctness, and security properties.
\paragraph{Efficiency.}
The algorithms $(\keygen,\StateGen,\Vrfy)$ are uniform QPT algorithms.
\paragraph{Correctness.}
We have
\begin{align}
\Pr[\top\la\Vrfy(1^\secp,k,\psi_k):k\la\keygen(1^\secp),\psi_k\la\StateGen(1^\secp,\psi_k)]\geq 1-\negl(\secp).
\end{align}

\paragraph{Security.}
For any non-uniform QPT algorithm $\cA$ and any polynomial $t(\cdot)$,
\begin{align}
\Pr[\top\la\Vrfy(1^\secp,k^*,\psi_k):k\la\keygen(1^\secp),\psi_k\la\StateGen(1^\secp,k),k^*\la\cA(\psi_k^{\otimes t(\secp)})]\leq \negl(\secp).
\end{align}
\end{definition}
\begin{remark}
    If a OWSG scheme $(\keygen,\StateGen,\Vrfy)$ satisfies
\begin{align}
\Pr[\top\la\Vrfy(1^\secp,k,\psi_k):k\la\keygen(1^\secp),\psi_k\la\StateGen(1^\secp,\psi_k)]=1    
\end{align}
for all security parameters $\secp\in\N$, then we say that the OWSG scheme satisfies perfect correctness.
\end{remark}

\paragraph{Public-Key Quantum Money Mini-Scheme.}
In this work, we consider public-key quantum money mini-scheme.
\begin{definition}[Public-Key Quantum Money Mini-Scheme~\cite{STOC:AC12}]\label{def:money}
    A public-key quantum money mini-scheme candidate is a set of algorithms $\Sigma\seteq(\Mint,\Vrfy)$ such that:
\begin{itemize}
    \item[$\Mint(1^\secp)$:] It takes a security parameter $1^\secp$, and outputs a serial number $s$ and a quantum state $\rho_s$.
    \item[$\Vrfy(1^\secp,s,\rho_s)$:] It takes a security parameter $1^\secp$, $s$, and $\rho_s$, and outputs $\top$ or $\bot$.
\end{itemize}

We say that a candidate $\Sigma$ is a public-key quantum money mini-scheme if it satisfies the following efficiency, correctness, and security properties.

\paragraph{Efficiency.}
The algorithms $(\Mint,\Vrfy)$ are uniform QPT algorithms.

\paragraph{Correctness.}
We have
\begin{align}
    \Pr[\top\la\Vrfy(1^\secp,s,\rho_s):(s,\rho_s)\la\Mint(1^\secp)]\geq 1-\negl(\secp).
\end{align}

\paragraph{Security.}
Given a public-key quantum money mini-scheme $\Sigma$, we consider the security experiment $\mathsf{Exp}_{\Sigma,\cA}^{\mathsf{unc}}(\secp)$ against $\cA$.
\begin{enumerate}
    \item The challenger first runs $(s,\rho_s)\la\Mint(1^\secp)$, and sends $(s,\rho_s)$ to $\cA$.
    \item $\cA$ outputs $\sigma_{R[1],R[2]}$ over the $R[1]$ register and $R[2]$ register, and sends it to the challenger.
    \item For $i\in\{1,2\}$, the challenger runs $\Vrfy(s,\cdot)$ on the $R[i]$ register and obtains $b[i]$.
    \item The experiment outputs $1$ if $b[1]=b[2]=\top$.
\end{enumerate}
We say that $\Sigma$ satisfies security if for all non-uniform QPT adversaries $\cA$, we have
\begin{align}
    \Pr[\mathsf{Exp}_{\Sigma,\cA}^{\mathsf{unc}}(\secp)=1]\leq \negl(\secp).
\end{align}

\end{definition}

We note that a public-key quantum money mini-scheme can be upgraded into a full-fledged public-key quantum money additionally using standard digital signatures~\cite{STOC:AC12}.

\paragraph{Canonical Quantum Bit Commitment.} 
\begin{definition}[Canonical Quantum Bit Commitment~\cite{Asia:Yan22}]\label{def:quantum_commitment}
    A candidate for canonical quantum bit commitment is a set of uniform QPT unitaries $\{Q_0(\secp),Q_1(\secp) \}_{\secp\in\N}$ acting on the register $\mathbf{C}$ and $\mathbf{R}$.
    We consider the following two properties.
\paragraph{Hiding.}
    We say that a candidate for canonical quantum bit commitment $\{Q_0(\secp),Q_1(\secp)\}_{\secp\in\N}$ satisfies $c$-statistical hiding (resp. $c$-computational hiding) if $\Tr_{\mathbf{R}}(Q_0(\secp)\ket{0}_{\mathbf{CR}})$ is $c$-statistically indistinguishable (resp. $c$-computationally indistinguishable) from $\Tr_{\mathbf{R}} (Q_1(\secp)\ket{0}_{\mathbf{CR}})$ for all sufficiently large $\secp\in\N$.
    
    If a candidate for canonical quantum bit commitment satisfies $\negl(\secp)$-statistical hiding (resp. $\negl(\secp)$-computational hiding), then we say that the candidate satisfies statistical hiding (resp. computational hiding).

\paragraph{Binding.}
    We say that a candidate for canonical quantum bit commitment $\{Q_0(\secp),Q_1(\secp)\}_{\secp\in\N}$ satisfies $c$-statistical binding (resp. $c$-computational binding) if for all sufficiently large security parameters $\secp\in\N$, any unbounded-time (resp. QPT) unitary $U$ over $\mathbf{R}$ and an additional register $\mathbf{Z}$ and any polynomial-size $\ket{\tau}$, it holds that
    \begin{align}
        \norm{(\bra{0}Q_1^{\dagger}(\secp))_{\mathbf{C,R}}(I_{\mathbf{C}}\otimes U_{\mathbf{R},\mathbf{Z}})((Q_0(\secp)\ket{0}_{\mathbf{C,R}})\ket{\tau}_{\mathbf{Z}})}\leq c.
    \end{align} 

    If a candidate for canonical quantum bit commitment satisfies $\negl(\secp)$-statistical binding (resp. $\negl(\secp)$-computational binding), then we say that the candidate satisfies statistical binding (resp. computational binding).
\end{definition}

It was shown that we can convert the flavor of quantum bit commitment as follows.
\begin{lemma}[Converting Flavors:\cite{EC:HMY23}]\label{lem:converting_flavor}
Let $\{Q_0(\secp),Q_1(\secp)\}_{\secp\in\N}$ be a candidate of canonical quantum bit commitment.
Let $\{\widetilde{Q_{0}}(\secp),\widetilde{Q_{1}}(\secp)\}_{\secp\in\N}$ be a candidate of canonical quantum bit commitment described as follows:
\begin{itemize}
    \item The role of commitment and reveal registers are swapped from $(Q_{0}(\secp),Q_{1}(\secp))$ and the commitment register is augmented by an additional one-qubit register which we denote $\mathbf{D}$.
    In other words, if $\mathbf{C}$ and $\mathbf{R}$ are the commitment and reveal registers of $(Q_{0}(\secp),Q_{1}(\secp))$, then the commitment and reveal registers of $(\widetilde{Q_{0}}(\secp),\widetilde{Q_{1}}(\secp))$ are defined as $\widetilde{\mathbf{C}}\seteq (\mathbf{R},\mathbf{D})$ and $\widetilde{\mathbf{R}}\seteq \mathbf{C}$, where $\mathbf{D}$ is an additional one-qubit register.
    \item For $b\in\bit$, the unitary $\widetilde{Q_b}(\secp)$ is defined as follows:
    \begin{align}
        \widetilde{Q_{b}}(\secp)\seteq \left(Q_{0}(\secp)\otimes \ket{0}\bra{0}_{\mathbf{D}}+Q_{1}(\secp)\otimes \ket{1}\bra{1}_{\mathbf{D}}\right)\left(I_{\mathbf{RC}}\otimes Z^b_{\mathbf{D}}H_{\mathbf{D}}\right).
    \end{align}
\end{itemize}
The following holds for $X,Y\in\{\mbox{statistical, computational}\}$.
\begin{enumerate}
\item 
If $\{Q_{0}(\secp),Q_{1}(\secp)\}_{\secp\in\N}$ satisfies $c$-$X$ hiding, then $\{\widetilde{Q_{0}}(\secp),\widetilde{Q_{1}}(\secp)\}_{\secp\in\N}$ satisfies $\sqrt{c}$-$X$ binding.
\item
If $\{Q_{0}(\secp),Q_{1}(\secp)\}_{\secp\in\N}$ satisfies $\negl(\secp)$-$Y$ binding, then $\{\widetilde{Q_{0}}(\secp),\widetilde{Q_{1}}(\secp)\}_{\secp\in\N}$ satisfies $\negl(\secp)$-$Y$ hiding.
\end{enumerate}
\end{lemma}
\begin{remark}
The previous work \cite{EC:HMY23} considered the case where the original commitment $\{Q_0(\secp),Q_1(\secp)\}_{\secp\in\N}$ satisfies $\negl(\secp)$-$X$ hiding. 
However, for our purpose, we need to analyze the case where the original commitment $\{Q_0(\secp),Q_1(\secp)\}_{\secp\in\N}$ satisfies $c$-X hiding for some constant $c$ instead of $\negl(\secp)$-X hiding.
For the reader's convenience, we describe the proof in \cref{sec:convert}.
Remark that the proof is the same as the previous work.
\end{remark}

\paragraph{Unclonable Encryption.}
In this work, we consider unclonable encryption with unclonable IND-CPA security.

\begin{definition}[Unclonable Secret-Key Encryption~\cite{TQC:BL20}]\label{def:unc_one_ske}
A candidate for unclonable secret-key encryption for $n(\secp)$-bit plaintexts is a set of algorithms $\Sigma\seteq(\keygen,\Enc,\Dec)$ such that:
    \begin{itemize}
        \item[$\keygen(1^\secp)$:] It takes as input a security parameter $1^\secp$, and outputs a classical secret-key $\sk$.
        \item[$\Enc(1^\secp,\sk,m)$:] It takes as input a security parameter $1^\secp$, $\sk$ and $m\in\bit^{n(\secp)}$, and outputs a quantum ciphertext $\ct$.
        \item[$\Dec(1^\secp,\sk,\ct)$:] It takes as input a security parameter $1^\secp$, $\sk$ and $\ct$, and outputs $m$. 
    \end{itemize}
    We say that a candidate $\Sigma$ is an unclonable SKE scheme if it satisfies the following efficiency, correctness, IND-CPA security, and unclonable IND-CPA security.
    \paragraph{Efficiency.}
    The algorithms $(\keygen,\Enc,\Dec)$ are uniform QPT algorithms.
    \paragraph{Correctness.}
We have
    \begin{align}
        \Pr[m\la\Dec(1^\secp,\sk,\ct):\sk\la\keygen(1^\secp),\ct\la\Enc(1^\secp,\sk,m)]\geq 1-\negl(\secp).
    \end{align}
\paragraph{Unclonable IND-CPA Security.}
    We require that $\Sigma$ satisfies standard IND-CPA security.
    In addition to the standard IND-CPA security, we require that $\Sigma$ satisfies the unclonable IND-CPA security defined below.
    Given an unclonable encryption $\Sigma$, we consider the unclonable IND-CPA security experiment $\mathsf{Exp_{\Sigma,(\cA,\cB,\cC)}^{unclone}}(\secp)$ against $(\cA,\cB,\cC)$.
    \begin{enumerate}
        \item The challenger runs $\sk\la\keygen(1^\secp)$.
        \item $\cA$ can query $\Enc(1^\secp,\sk,\cdot)$ polynomially many times.
        \item $\cA$ sends $(m_0,m_1)$ to the challenger.
        \item The challenger samples $b\la\bit$, runs $\ct_b\la\Enc(1^\secp,\sk,m_b)$, and sends $\ct_b$ to $\cA$.
        \item $\cA$ produces $\rho_{\cB,\cC}$ and sends the corresponding registers to $\cB$ and $\cC$.
        \item $\cB$ and $\cC$ receive $\sk$ and output $b_\cB$ and $b_\cC$.
        \item The experiment outputs $1$ indicating win if $b_{\cB}=b_{\cC}=b$, and otherwise $0$.
    \end{enumerate}
We say that $\Sigma$ is unclonable IND-CPA secure if for all sufficiently large security parameters $\secp\in\N$, for all non-uniform QPT adversaries $(\cA,\cB,\cC)$,
\begin{align}
    \Pr[\mathsf{Exp_{\Sigma,(\cA,\cB,\cC)}^{unclone}}(\secp)=1]\leq \frac{1}{2}+\negl(\secp).
\end{align}
\end{definition}
\begin{remark}
We also consider one-time unclonable secret-key encryption.
It is the same as unclonable secret-key encryption except that it satisfies one-time IND-CPA security and one-time unclonable IND-CPA security instead of IND-CPA security and unclonable IND-CPA security.
The one-time unclonable IND-CPA security is the same as unclonable IND-CPA security except that the adversary is not allowed to query the encryption oracle.

\end{remark}

\begin{remark}
    If an unclonable SKE scheme $(\keygen,\Enc,\Dec)$ satisfies
\begin{align}
\Pr[m\la\Dec(1^\secp,\sk,\ct):\sk\la\keygen(1^\secp),\ct\la\Enc(1^\secp,\sk,m)]=1    
\end{align}
for all security parameters $\secp\in\N$ and all $m\in\cM_\secp$, then we say that the unclonable SKE scheme satisfies perfect correctness.
\end{remark}

We also consider unclonable PKE.
For clarity, we describe unclonable PKE with unclonable IND-CPA security.
\begin{definition}[Unclonable Public-Key Encryption\cite{TCC:AK21}]
    A candidate for unclonable public-key encryption for $n(\secp)$-bit plaintexts is a set of algorithms $\Sigma\seteq(\keygen,\Enc,\Dec)$ such that:
    \begin{itemize}
        \item[$\keygen(1^\secp)$:] It takes as input a security parameter $1^\secp$, and outputs a classical secret-key $\sk$ and a classical public-key $\pk$.
        \item[$\Enc(1^\secp,\pk,m)$:] It takes as input a security parameter $1^\secp$, $\pk$ and $m\in\bit^{n(\secp)}$, and outputs a quantum ciphertext $\ct$.
        \item[$\Dec(1^\secp,\sk,\ct)$:] It takes as input a security parameter $1^\secp$, $\sk$ and $\ct$, and outputs $m$. 
    \end{itemize}
    We say that a candidate $\Sigma$ satisfies efficiency, correctness, IND-CPA security, and unclonable IND-CPA security, respectively if $\Sigma$ satisfies the following efficiency, correctness, IND-CPA security, and unclonable IND-CPA security property, respectively.
    \paragraph{Efficiency.}
    The algorithms $(\keygen,\Enc,\Dec)$ are uniform QPT algorithms.
    \paragraph{Correctness.}
We have
    \begin{align}
        \Pr[m\la\Dec(1^\secp,\sk,\ct):(\sk,\pk)\la\keygen(1^\secp),\ct\la\Enc(1^\secp,\pk,m)]\geq 1-\negl(\secp).
    \end{align}
\paragraph{Unclonable IND-CPA Security.}
    We require that $\Sigma$ satisfies standard IND-CPA security.
    In addition to the standard IND-CPA security, we require that $\Sigma$ satisfies the unclonable IND-CPA security defined below.
    Given an unclonable encryption $\Sigma$, we consider the unclonable IND-CPA security experiment $\mathsf{Exp_{\Sigma,(\cA,\cB,\cC)}^{unclone}}(\secp)$ against $(\cA,\cB,\cC)$.
    \begin{enumerate}
        \item The challenger runs $(\sk,\pk)\la\keygen(1^\secp)$, and sends $\pk$ to $\cA$.
        \item $\cA$ sends $(m_0,m_1)$ to the challenger.
        \item The challenger samples $b\la\bit$, runs $\ct_b\la\Enc(1^\secp,\pk,m_b)$, and sends $\ct_b$ to $\cA$.
        \item $\cA$ produces $\rho_{\cB,\cC}$ and sends the corresponding registers to $\cB$ and $\cC$.
        \item $\cB$ and $\cC$ receive $\sk$ and output $b_\cB$ and $b_\cC$.
        \item The experiment outputs $1$ indicating win if $b_{\cB}=b_{\cC}=b$, and otherwise $0$.
    \end{enumerate}
We say that $\Sigma$ is unclonable IND-CPA secure if for all sufficiently large security parameters $\secp\in\N$, for all non-uniform QPT adversaries $(\cA,\cB,\cC)$,
\begin{align}
    \Pr[\mathsf{Exp_{\Sigma,(\cA,\cB,\cC)}^{unclone}}(\secp)=1]\leq \frac{1}{2}+\negl(\secp).
\end{align}
\end{definition}

\begin{remark}
    We say that (one-time) unclonable SKE (resp. PKE) $\Sigma$ is unclonable SKE (resp. SKE) for single-bit plaintexts if a plaintext space $\Ms_\secp$ is $\Ms_\secp\seteq\bit$ for all security parameters $\secp\in\N$.
    Note that we cannot expand the plaintext space by bit-wise encryption. 
\end{remark}

\paragraph{Decomposable Quantum Randomized Encoding.}
\begin{definition}[Decomposable Quantum Randomized Encoding(DQRE)~\cite{Stoc:BY22}]
    A DQRE scheme is a tuple of algorithms $(\Enc,\Dec)$ such that:    
    \begin{itemize}
        \item[$\Enc(1^\secp,F,x)$:]
        It takes $1^\secp$ with $\secp\in\N$, a general quantum circuit $F$ and a possibly quantum input $x$ as inputs, and outputs $\widehat{F(x)}$.
        \item[$\Dec(1^\secp,\widehat{F(x)})$:]
        It takes as input $1^\secp$, and $\widehat{F(x)}$, and outputs $F(x)$.
    \end{itemize}
    We require the following four properties:
    \paragraph{Efficiency.}
    $(\Enc,\Dec)$ are uniform QPT algorithms.
    \paragraph{Correctness.}
    For all quantum states $(x,q)$ and randomness $r$, it holds that $(F(x),q)=(\Dec(1^\secp,\widehat{F}(x;r)),q)$, where $\widehat{F}(x;r)$ is an output of $\mathsf{Enc}(1^{\secp},F,x;r)$.
    \paragraph{Security.}
    There exists a uniform QPT algorithm $\Sim$ such that
    for all quantum states $(x,q)$ and non-uniform QPT adversary $\cA$, there exists some negligible function $\negl$ that satisfies,
    \begin{align}
        \abs{\Pr[1\la\cA(\widehat{F}(x;r),q)] -\Pr[1\la\cA(\Sim(1^\secp,\abs{F},F(x)),q)]}\leq \negl(\secp),
    \end{align}
    where the state on the left-hand side is averaged over $r$ and $\abs{F}$ is the size of the general quantum circuit $F$.
    \begin{remark}
    In the security of the original paper~\cite{Stoc:BY22}, the simulator $\Sim$ takes the topology of $F$ as input.
    Without loss of generality, we can replace the topology of $F$ with the size of $F$ because we can hide the topology of $F$ by using a universal quantum circuit.
    \end{remark}
    \paragraph{Decomposability.}
    There exists a quantum state $e$ (called the resource state of the encoding), and operation $\widehat{F}_{\mathsf{off}}$ (called the offline part of the encoding) and a collection of input encoding operations $\widehat{F_1},\cdots,\widehat{F_n}$ such that for all inputs $x=(x_1,\cdots,x_n)$,
    \begin{align}
        \widehat{F}(x;r)=\left(\widehat{F}_{\mathsf{off}},
        \widehat{F}_1,\widehat{F_2},\cdots,\widehat{F_n}\right)(x,r,e)
    \end{align}
    where the functions $\widehat{F}_{\mathsf{off}},\widehat{F_1},\cdots,\widehat{F_n}$ act on disjoint subsets of qubits from $e,x$ (but can depend on all bits of $r$), each $\widehat{F_i}$ acts on a single qubit $x_i$ and $\widehat{F}_{\mathsf{off}}$ does not act on any of the qubits of $x$.
    \paragraph{Classical Labels.}
    If $x_i$ is a classical bit, then $\widehat{F_i}(x_i,r)$ is a classical string as well.
\end{definition}

\begin{theorem}[\cite{Stoc:BY22}]\label{thm:quantum_garble}
    Decomposable quantum randomized encoding exists if OWFs exist.
\end{theorem}

\begin{proposition}\label{prop:ind_RE}
    Let $\Sigma\seteq(\Enc,\Dec)$ be a decomposable quantum randomized encoding. Then, for any quantum circuits $F_0,F_1$ with the same size, for any possibly quantum input $x_0$ and $x_1$ such that $F_0(x_0)=F_1(x_1)$, 
    $\widehat{F_0}(x_0;r_0)$ is computationally indistinguishable from $\widehat{F_1}(x_1;r_1)$,
    where both quantum states are averaged over the randomness $r_0$ and $r_1$. 
\end{proposition}
This can be shown by a standard hybrid argument, and thus we omit the proof.

\section{Robust OWSGs Combiner}\label{sec:owsg}
\begin{definition}[Robust OWSGs Combiner]\label{def:robust_OWSG_comb}
    A robust OWSGs combiner is a deterministic classical polynomial-time Turing machine $\cM$ with the following properties:
    \begin{itemize}
        \item $\cM$ takes as input $1^n$ with $n\in\N$ and $n$-candidates OWSGs $\{\Sigma_i\seteq(\keygen_i,\StateGen_i,\Vrfy_i)\}_{i\in[n]}$ promised that all candidates satisfy efficiency, and outputs a single set of algorithms $\Sigma\seteq(\keygen,\StateGen,\Vrfy)$.
        \item If all of $\{\Sigma_i\}_{i\in[n]}$ satisfy efficiency and at least one of $\{\Sigma_i\}_{i\in[n]}$ satisfies both correctness and security, then $\Sigma$ is an OWSG scheme that satisfies efficiency, correctness, and security.
    \end{itemize}
\end{definition}
\begin{remark}
    In the previous work~\cite{EC:HKNRR05}, they define robust combiners in a similar way where $n$ is treated as an arbitrary function in the security parameter.
    However, it is unclear what is meant by the definition where $n$ is a super-constant.
    This is because the security parameter for the scheme $\Sigma$ obtained by a robust combiner is an arbitrary non-negative integer after combining $n$ candidates $\{\Sigma_i\}_{i\in[n]}$.
    Therefore, in the definition above, we consider $n$ as a constant in $\secp$.
    On the other hand, \cref{def:robust_OWSG_comb} is not sufficient to construct universal construction since $n$ is constant in $\secp$.
    Therefore, we also introduce another definition (\cref{def:robust_owsg_comb_univ}) of a robust combiner, where $n$ can be dependent on $\secp$.
    Although our construction actually satisfies \cref{def:robust_owsg_comb_univ}, here we consider \cref{def:robust_OWSG_comb} for simplicity.
\end{remark}

\begin{theorem}\label{thm:OWSG_combiner}
   A robust OWSGs combiner exists.
\end{theorem}

For proving \cref{thm:OWSG_combiner}, we introduce the following \cref{lem:amp_OWSG_cor}.
\begin{lemma}\label{lem:amp_OWSG_cor}
    Let $\Sigma=(\keygen,\StateGen,\Vrfy)$ be a candidate of OWSG.
    From $\Sigma$, we can construct a OWSG scheme $\Sigma^*$ with the following properties:
    \begin{enumerate}
        \item If $\Sigma$ is uniform QPT algorithm, $\Sigma^*$ is uniform QPT algorithm.
        \item $\Sigma^*$ satisfies perfect correctness.
        \item If $\Sigma$ is a uniform QPT algorithm and satisfies correctness and security, then $\Sigma^*$ satisfies security.
    \end{enumerate}
\end{lemma}
\begin{proof}[Proof of \cref{lem:amp_OWSG_cor}]
    Without loss of generality, $\Vrfy(1^\secp,k,\psi)$ can be considered as the algorithm working in the following way:

For input $(1^\secp,k,\psi)$, run a classical Turing machine $\cM$ on $(1^\secp,k,\abs{\psi})$, obtain $U_{\Vrfy,k}$, append auxiliary state $\ket{0\cdots0}\bra{0\cdots 0}$ to $\psi$, apply a unitary $U_{\Vrfy,k}$ on $\psi\otimes\ket{0\cdots 0}\bra{0\cdots 0} $, and measure the first qubit of $U_{\Vrfy,k}(\psi\otimes\ket{0\cdots 0}\bra{0\cdots 0})U_{\Vrfy,k}^{\dagger}$ with the computational basis and output $\top$ if the measurement result is $1$ and $\bot$ otherwise.

We describe the $\Sigma^*\seteq(\keygen^*,\StateGen^*,\Vrfy^*)$.
\begin{description}
    \item[$\keygen^*(1^\secp):$]$ $
    \begin{itemize}
        \item Run $k\la\keygen(1^\secp)$.
        \item Output $k^*\seteq k$.
    \end{itemize}
    \item[$\StateGen^*(1^\secp,k^*)$:]$ $
    \begin{itemize}
        \item Parse $k^*=k$. 
        \item Run $ \psi_k\la \StateGen(1^\secp,k)$.
        \item Run $U_{\Vrfy,k}$ on $\psi_k\otimes\ket{0\cdots 0}\bra{0\cdots 0}$, and measures the first qubit of $ U_{\Vrfy,k}\left(\psi_{k} \otimes \ket{0\cdots 0}\bra{0\cdots 0}\right)U_{\Vrfy,k}^{\dagger}$,  in the computational basis, and obtains the measurement result $b$ and post-measurement quantum state $\rho_{b,k}$.
        \begin{itemize}
            \item If the measurement result is $1$, then output $\psi_k^*\seteq U_{\Vrfy,k}^{\dagger}(\ket{1}\bra{1}\otimes \rho_{1,k})U_{\Vrfy,k}\otimes \ket{1}\bra{1}$.
            \item If the measurement result is $0$, then output $\psi_k^*\seteq U_{\Vrfy,k}^{\dagger}(\ket{0}\bra{0}\otimes \rho_{0,k})U_{\Vrfy,k}\otimes \ket{0}\bra{0}$.
        \end{itemize}
    \end{itemize}
    \item[$\Vrfy^*(1^\secp,k^*,\psi^*)$:]$ $
    \begin{itemize}
        \item Parse $k^*=k$ and $\psi^*\seteq \rho\otimes \ket{b}\bra{b}$.
        \item Measure the last bit of $\psi^*$ in the computational basis.
        \begin{itemize}
            \item If $1$ is obtained, then measure the first qubit of $U_{\Vrfy,k}\rho U_{\Vrfy,k}^{\dagger} $ in the computational basis, and output $\top$ if the measurement outcome is $1$ and $\bot$ otherwise.
            \item If $0$ is obtained, then output $\top$.
        \end{itemize}
    \end{itemize}
\end{description}
The first item and the second item straightforwardly follow, and thus we skip the proof.

\paragraph{Proof of the third item.}
Assume that $\Sigma^*$ is not secure for contradiction.
More formally, assume that there exists a QPT adversary $\cA$ such that the following probability is non-negligible
\begin{align}
    \Pr[\top\la\Vrfy^*(1^\secp,k',\psi_k^*):
    \begin{array}{ll}
         k\la\keygen(1^\secp)  \\
         \psi_k^*\la\StateGen^*(1^\secp,k)\\
         k'\la\cA(\psi_k^{*\otimes t(\secp)})
    \end{array}
    ].
\end{align}
Then, construct $\cB$ that breaks the security of $\Sigma$ as follows.
\begin{enumerate}
    \item $\cB$ receives $\psi_k^{\otimes t(\secp)}$ from $\cC$ which is the challenger of $\Sigma$.
    \item $\cB$ sends $\left(\psi_k\otimes \ket{0\cdots 0}\bra{0\cdots 0}\otimes \ket{1}\bra{1}\right)^{\otimes t}$ to $\cA$.
    \item $\cB$ receives $k'$ from $\cA$.
    \item $\cB$ sends $k'$ to $\cC$.
\end{enumerate}
From the construction of $\cB$, $\cB$ simulates the security experiment of $\Sigma^*$ except that it uses $\psi_k\otimes \ket{0\cdots 0}\bra{0\cdots 0}\otimes \ket{1}\bra{1}$ instead of $\psi_{k}^*$.
Because we assume that $\Sigma$ satisfies correctness, we have
\begin{align}
    U_{\Vrfy,k}(\psi_k\otimes \ket{0\cdots 0}\bra{0\cdots 0}) U_{\Vrfy,k}^{\dagger} =\negl(\secp)\ket{0}\bra{0}\otimes\rho_{0,k}+(1-\negl(\secp))\ket{1}\bra{1}\otimes \rho_{1,k}
\end{align}
where $k\la\keygen(1^\secp)$, $\psi_k\la\StateGen(1^\secp,k)$, and $\rho_{0,\sk}$ and $\rho_{1,\sk}$ are some appropriate quantum state.

From the gentle measurement lemma (\cref{lem:gentle}), we have
\begin{align}
\norm{U_{\Vrfy,k}(\psi_k\otimes \ket{0\cdots 0}\bra{0\cdots 0}) U_{\Vrfy,k}^{\dagger}-\ket{1}\bra{1}\otimes\rho_{1,k}}_1\leq  \negl(\secp).
\end{align}  
In particular, this implies that
\begin{align}
    \norm{\psi_{k}\otimes \ket{0\cdots 0}\bra{0\cdots 0}\otimes \ket{1}\bra{1}-\psi_{k}^*}_1\leq\negl(\secp).
\end{align}
Therefore, we have
\begin{align}
    &\Pr[\top\la\Vrfy(1^\secp,k',\psi_k):
    \begin{array}{ll}
         k\la\keygen(1^\secp)  \\
         \psi_k\la\StateGen(k)\\
         k'\la\cB\left(\psi_k^{\otimes t(\secp)}\right)
    \end{array}
    ]\\
    &=
    \Pr[\top\la\Vrfy^*(1^\secp,k',\psi_k\otimes\ket{0\cdots0}\bra{0\cdots0}\otimes \ket{1}\bra{1}):
    \begin{array}{ll}
         k\la\keygen(1^\secp)  \\
         \psi_k\la\StateGen(k)\\
         k'\la\cA\left((\psi_k\otimes \ket{0\cdots0}\bra{0\cdots0}\otimes \ket{1}\bra{1})^{\otimes t(\secp)}\right)
    \end{array}
    ]\\
    &\geq\Pr[\top\la\Vrfy^*(1^\secp,k',\psi_k^*):
    \begin{array}{ll}
         k\la\keygen(1^\secp)  \\
         \psi_{k}^*\la\StateGen^*(k)\\
         k'\la\cA(\psi_k^{*\otimes t(\secp)})
    \end{array}
    ]-\negl(\secp)\\
    &\geq 1/\secp^c-\negl(\secp),
\end{align} 
where in the first equation we have used
\begin{align}
    \Pr[\top\la\Vrfy(1^\secp,k',\psi)]=\Pr[\top\la\Vrfy^*(1^\secp,k',\psi\otimes \ket{0\cdots0}\bra{0\cdots0}\otimes \ket{1}\bra{1})]
\end{align}
for any $\secp\in\N$, $k^*$, and $\psi$,
and in the second inequality, we have used that $\norm{\psi_{k}\otimes \ket{0\cdots 0}\bra{0\cdots 0}\otimes \ket{1}\bra{1}-\psi_{k}^*}_1\leq\negl(\secp)$.
This contradicts that $\Sigma$ satisfies security, and thus $\Sigma^*$ satisfies security.

\end{proof}

\begin{proof}[Proof of \cref{thm:OWSG_combiner}]
Below, we consider a fixed constant $n$. 
Let us introduce some notations.

\paragraph{Notations:}
\begin{itemize}
    \item Let $\Sigma_i\seteq(\keygen_i,\StateGen_i,\Vrfy_i)$ be a candidate of OWSG for $i\in[n]$.
    \item For a candidate of OWSG $\Sigma_i$, let $\Sigma_i^*\seteq(\keygen_i^*,\StateGen_i^*,\Vrfy_i^*)$ be a candidate of OWSG derived from \cref{lem:amp_OWSG_cor} with the following properties:
    \begin{itemize}
        \item If $\Sigma_i$ satisfies efficiency, then $\Sigma_i^*$ satisfies efficiency.
        \item $\Sigma_i^*$ satisfies perfect correctness.
        \item If $\Sigma_i$ satisfies efficiency, correctness and security, then $\Sigma_i^*$ satisfies security.
    \end{itemize}
\end{itemize}

\paragraph{Construction of Robust OWSG Combiner:}
A robust combiner $\cM$ is a classical Turing machine that takes as input $1^n$ and $\{\Sigma_i\}_{i\in[n]}$, and outputs $\Sigma=(\keygen,\StateGen,\Vrfy)$ working in the following way.
\begin{description}
    \item[$\keygen(1^\secp)$:]$ $ 
    \begin{itemize}
        \item For all $i\in[n]$, run $k_i^*\la\keygen_i^*(1^\secp)$.
        \item Output $k\seteq\{k_i^*\}_{i\in[n]}$.
    \end{itemize}
    \item[$\StateGen(1^\secp,k)$:]$ $
    \begin{itemize}
        \item Parse $k=k_1^*||\cdots ||k_{n}^*$.
        \item For all $i\in[n]$, run $\psi_{k_i^*}\la\StateGen_i^*(k_i^*)$.
        \item Output $\psi_{k}\seteq \bigotimes_{i\in[n]}\psi_{k_i^*}$.
    \end{itemize}
    \item[$\Vrfy(1^\secp,k,\psi_k)$:]$ $
    \begin{itemize}
        \item Parse $k=k_1||\cdots ||k_{n}$ and $\psi_k=\bigotimes_{i\in[n]}\psi_{k_i}$.
        \item For all $i\in[n]$, run $\Vrfy_i^*(k_i,\psi_{k_i})$.
        If $\top\la\Vrfy_i^*(k_i^*,\psi_{k_i^*})$ for all $i\in[n]$, output $\top$. Otherwise, output $\bot$.
    \end{itemize}
\end{description}
\cref{thm:OWSG_combiner} follows from the following \cref{lem:OWSG_eff,lem:OWSG_cor,lem:OWSG_sec}.

\begin{lemma}\label{lem:OWSG_eff}
If all of $\{\Sigma_i\}_{i\in[n]}$ satisfies efficiency, then $\Sigma$ satisfies efficiency.
\end{lemma}

\begin{lemma}\label{lem:OWSG_cor}
    $\Sigma$ satisfies perfect correctness. 
\end{lemma}

\begin{lemma}\label{lem:OWSG_sec}
    If all of $\{\Sigma_i\}_{i\in[n]}$ satisfies efficiency and at least one of $\{\Sigma_i\}_{i\in[n]}$ satisfies correctness and security, then $\Sigma$ satisfies security.
\end{lemma}
\cref{lem:OWSG_eff} trivially follows.
\cref{lem:OWSG_cor} follows because $\Sigma_i^*$ satisfies correctness for all $i\in[n]$.
The proof of \cref{lem:OWSG_sec} is a standard hybrid argument, and thus we skip the proof.

\end{proof}

\subsection{Universal Construction}\label{sec:univ_owsg}
\begin{definition}
    We say that a set of uniform QPT algorithms $\Sigma_{\mathsf{Univ}}=(\keygen,\StateGen,\Vrfy)$ is a universal construction of OWSG if $\Sigma_{\mathsf{Univ}}$ is an OWSG scheme as long as there exists an OWSG.
\end{definition}

\begin{theorem}\label{thm:univ_owsg}
    There exists a universal construction of OWSG.
\end{theorem}

For showing \cref{thm:univ_owsg}, the robust OWSGs combiner of \cref{def:robust_OWSG_comb} is not adequate to construct universal construction for OWSGs. 
Therefore, we reintroduce a definition of robust OWSGs combiner, which we call robust OWSGs combiner for universal construction.

\begin{definition}[Robust OWSGs Combiner for universal construction]\label{def:robust_owsg_comb_univ}
    A $(1,n)$-robust OWSGs combiner for universal construction $\Comb.\Sigma$ consists of three algorithms $(\Comb.\keygen,\Comb.\StateGen,\Comb.\Vrfy)$, where $n$ is some polynomial. 
    A $(1,n)$-robust OWSG combiner $(\Comb.\keygen,\Comb.\StateGen,\Comb.\Vrfy)$ has the following syntax:
    \begin{itemize}
    \item[$\Comb.\keygen(1^\secp,\{\Sigma_i\}_{i\in[n(\secp)]})$:]$ $
    It takes as input a security parameter $\secp$ and $n(\secp)$ candidates of OWSGs $\{\Sigma_i\}_{i\in[n(\secp)]}$ and outputs a classical key $k$.
    \item[$\Comb.\StateGen(1^\secp,k,\{\Sigma_i\}_{i\in[n(\secp)]})$:]$ $
    It takes as input a security parameter $1^\secp$, $k$ and $\{\Sigma_i\}_{i\in[n(\secp)]}$, and outputs a quantum state $\psi_k$.
    \item[$\Comb.\Vrfy(1^\secp,k,\psi_k,\{\Sigma_i\}_{i\in[n(\secp)]})$:]$ $
    It takes as input a security parameter $1^\secp$, $k$, $\psi_k$, and $\{\Sigma_i\}_{i\in[n(\secp)]}$, and outputs $\top$ or $\bot$.
    \end{itemize}
\paragraph{Efficiency.}
The algorithms $(\Comb.\keygen,\Comb.\StateGen,\Comb.\Vrfy)$ are uniform QPT algorithms.

\paragraph{Correctness.}
For all $n$ candidates $\{\Sigma_i\}_{i\in[n(\secp)]}$,  
\begin{align}
\Pr\left[
\top\la\Comb.\Vrfy(1^\secp,k,\psi_k,\{\Sigma_i\}_{i\in[n(\secp)]})
\  :
\begin{array}{ll}
k\la\Comb.\keygen(1^\secp,\{\Sigma_i\}_{i\in[n(\secp)]})\\
\psi_k\la\Comb.\StateGen(1^\secp,k,\{\Sigma_i\}_{i\in[n(\secp)]})
\end{array}
\right] 
\geq
1-\negl(\secp).
\end{align}

\paragraph{Security.}
Let $\{\Sigma_i\}_{i\in\N}$ be a sequence of candidates of OWSGs promised that $\Sigma_i$ satisfies efficiency for all $i\in\N$.
If there exists $i^*\in\N$ such that $\Sigma_{i^*}$ satisfies correctness and security and $i^*<n(\secp)$ for all sufficiently large security parameters $\secp\in\N$, then for all non-uniform QPT adversaries $\cA$ and all polynomials $t$, we have
\begin{align}
\Pr\left[\top\la\Comb.\Vrfy(1^\secp,k^*,\psi_k,\{\Sigma_i\}_{i\in[n(\secp)]})
\ :
\begin{array}{ll}
k\la\Comb.\keygen(1^\secp,\{\Sigma_i\}_{i\in[n(\secp)]})\\
\psi_k\la\Comb.\StateGen(1^\secp,k,\{\Sigma_i\}_{i\in[n(\secp)]})\\
k^*\la\cA(\psi_k^{\otimes t(\secp)})
\end{array}
\right]
\leq \negl(\secp).
\end{align}
\end{definition}

\begin{theorem}\label{thm:owsg_comb_univ}
    There exists a $(1,n)$-robust OWSG combiner for universal construction for all polynomial $n$.
\end{theorem}

We can show \cref{thm:owsg_comb_univ} in the same way as \cref{thm:OWSG_combiner}, and thus we skip the proof.
For proving \cref{thm:univ_owsg}, let us introduce the following \cref{prop:owsg_time}.

\begin{proposition}\label{prop:owsg_time}
Assume that there exist OWSGs.
Then, there exists a set of classical polynomial-time Turing machine $\cM^*\seteq(x^*,y^*,z^*)$ such that 
\begin{itemize}
    \item $\Sigma[\cM^*]\seteq (\keygen[x^*],\StateGen[y^*],\Vrfy[z^*])$ is a OWSG scheme that satisfies correctness and security.
    \item $x^*(1^\secp)$ halts within $\secp^3$ steps for all sufficiently large $\secp\in\N$.
    \item $y^*(1^\secp,k)$ halts within $\secp^3$ steps for all sufficiently large $\secp\in\N$, where $k\la\keygen[x^*](1^\secp)$.
    \item $z^*(1^\secp,k,\abs{\psi_k})$ halts within $\secp^3$ steps for all sufficiently large $\secp\in\N$, where $k\la\keygen[x^*](1^\secp)$ and $\psi_k\la\StateGen[y^*](1^\secp,k,\psi_k)$.
\end{itemize}
\end{proposition}

This can be shown by a standard padding trick.
For the reader's convenience, we describe the proof in \cref{sec:padding}.

\begin{proof}[Proof of \cref{thm:univ_owsg}]
First, let us describe some notations:
\paragraph{Notations.}
\begin{itemize}
    \item For a set of classical Turing machines $\cM\seteq x||y||z$, we write $\Sigma[\cM]\seteq(\keygen[x],\StateGen[y],\Vrfy[z])$ to mean the candidate of OWSG that works as follows:
    \begin{itemize}
        \item $\keygen[x](1^\secp)$ runs $x(1^\secp)$, obtains a general quantum circuit $C_\secp[x]$, runs $C_\secp[x]$, and outputs its output.
        \item $\StateGen[y](1^\secp,k)$ runs $y(1^\secp,k)$, obtains a general quantum circuit $C_{\secp,k}[y]$, runs $C_{\secp,k}[y]$, and outputs its output.
        \item $\Vrfy[z](1^\secp,k,\psi_k)$ runs $z(1^\secp,k,\abs{\psi_k})$, obtains a general quantum circuit $C_{\secp,k,\abs{\psi_k}}[z]$, runs $C_{\secp,k,\abs{\psi_k}}[z]$ on input $\psi_k$, and outputs its output.
    \end{itemize}
    \item For a set of classical Turing machines $\cM\seteq x||y||z$, we write $\widetilde{\Sigma}[\cM]\seteq(\widetilde{\keygen}[x],\widetilde{\StateGen}[y],\widetilde{\Vrfy}[z])$ to mean the candidate of OWSGs that works as follows:
        \begin{itemize}
        \item $\widetilde{\keygen}[x](1^\secp)$ runs $x(1^\secp)$.
        If $x$ does not halt within $\secp^3$ steps, $\widetilde{\keygen}[x]$ outputs $\top$. 
        Otherwise, obtains a general quantum circuit $C_\secp[x]$, runs $C_\secp[x]$, and outputs its output.
        \item $\widetilde{\StateGen}[y](1^\secp,k)$ outputs $\top$ if $k=\top$.
        Otherwise, $\widetilde{\StateGen}[y](1^\secp,k)$ runs $y(1^\secp,k)$.
        If $y$ does not halt within $\secp^3$ steps, $\widetilde{\StateGen}[y]$ outputs $\top$.
        Otherwise, obtains a general quantum circuit $C_{\secp,k}[y]$, runs $C_{\secp,k}[y]$, and outputs its output.
    \item $\widetilde{\Vrfy}[z](1^\secp,k,\psi_k)$ outputs $\top$ if $k=\top$ or $\psi_k=\top$. 
    Otherwise, $\widetilde{\Vrfy}[z]$ runs $z(1^\secp,k,\abs{\psi_k})$.
    If it does not halt within $\secp^3$ steps, $\widetilde{\Vrfy}[z]$ outputs $\top$.
    Otherwise, obtains a general quantum circuit $C_{\secp,k,\abs{\psi_k}}$, runs $C_{\secp,k,\abs{\psi_k}}(\psi_k)$ on input $\psi_k$, and outputs its output.
    \end{itemize}
    \item For any $\secp\in\N$,
    we write $\{\widetilde{\Sigma}[\cM]\}_{x,y,z\in[\secp]}$ to mean
    \begin{align}
        \{\widetilde{\keygen[x]},\widetilde{\StateGen[y]},\widetilde{\Vrfy[z]}\}_{x,y,z\in[\secp]}.
    \end{align}
    \item 
    We consider a polynomial $n$ such that $n(\secp)=\secp^3$ for all $\secp\in\N$ since we combine $\secp^3$-OWSG candidates.
    We write $\Comb.\Sigma\seteq \Comb.(\keygen,\StateGen,\Vrfy)$ to mean a $(1,n)$-robust OWSGs combiner for universal construction.
\end{itemize}
\paragraph{Construction.}
We give a description of $\Sigma_{\mathsf{Univ}}\seteq (\keygen_{\mathsf{Univ}},\StateGen_{\mathsf{Univ}},\Vrfy_{\mathsf{Univ}})$.
\begin{description}
    \item[$\keygen_{\mathsf{Univ}}(1^\secp)$:]$ $
    \begin{itemize}
        \item Output $k\la \Comb.\keygen(1^\secp, \{(\widetilde{\keygen}[x],\widetilde{\StateGen}[y],\widetilde{\Vrfy}[z])\}_{x,y,z\in[\secp]} )$.
    \end{itemize}
    \item[$\StateGen_{\mathsf{Univ}}(1^\secp,k)$:]$ $
    \begin{itemize}
        \item Output $\psi_k\la\Comb.\StateGen(1^\secp,k,\{(\widetilde{\keygen}[x],\widetilde{\StateGen}[y],\widetilde{\Vrfy}[z])\}_{x,y,z\in[\secp]})$, and output its output.
    \end{itemize}
    \item[$\Vrfy_{\Univ}(1^\secp,k,\psi_k)$:]$ $
    \begin{itemize}
        \item Output $\top/\bot\la\Comb.\Vrfy(1^\secp,k,\psi_k,\{(\widetilde{\keygen}[x],\widetilde{\StateGen}[y],\widetilde{\Vrfy}[z])\}_{x,y,z\in[\secp]})$.
    \end{itemize}
\end{description}

\cref{thm:univ_owsg} follows from the following \cref{lem:univ_owsg_eff,lem:univ_owsg_cor,lem:univ_owsg_sec}.
\begin{lemma}\label{lem:univ_owsg_eff}
    $\Sigma_{\mathsf{Univ}}$ satisfies efficiency.
\end{lemma}

\begin{lemma}\label{lem:univ_owsg_cor}
    $\Sigma_{\mathsf{Univ}}$ satisfies correctness.
\end{lemma}

\begin{lemma}\label{lem:univ_owsg_sec}
   If there exist OWSGs, then $\Sigma_{\mathsf{Univ}}$ satisfies security.
\end{lemma}

\begin{proof}[Proof of \cref{lem:univ_owsg_eff}]
    \cref{lem:univ_owsg_eff} follows because $(\widetilde{\keygen}[x],\widetilde{\StateGen}[y],\widetilde{\Vrfy}[z])$ is a set of uniform QPT algorithms for any $x,y,z\in[\secp]$, and $\Comb.\Sigma$ is also a set of uniform QPT algorithms.
\end{proof}

\begin{proof}[Proof of \cref{lem:univ_owsg_cor}]
From the construction, we have
\begin{align}
    &\Pr[\top\la\Vrfy_{\mathsf{Univ}}(1^\secp,k,\psi_k):k\la\keygen_{\mathsf{Univ}}(1^\secp),\psi_k\la\StateGen_{\mathsf{Univ}}(1^\secp,k)]\\
    &=\Pr\left[
    \top\la\Comb.\Vrfy(1^\secp,\psi_k,\{\widetilde{\Sigma}[\cM]\}_{x,y,z\in[\secp]})
    \ :
    \begin{array}{ll}
    k\la\Comb.\keygen(1^\secp,\{\widetilde{\Sigma}[\cM]\}_{x,y,z\in [\secp]})\\
    \psi_k\la\Comb.\StateGen(1^\secp,k,\{\widetilde{\Sigma}[\cM]\}_{x,y,z\in[\secp] })
    \end{array}
    \right].
\end{align}
    Because $\widetilde{\Sigma}[\cM]$ is a set of uniform QPT algorithms and the robust combiner $\Comb.\Sigma$ satisfies correctness, we have
\begin{align}
    \Pr\left[
    \top\la\Comb.\Vrfy(1^\secp,\psi_k,\{\widetilde{\Sigma}[\cM]\}_{x,y,z\in[\secp]})
    \ :
    \begin{array}{ll}
    k\la\Comb.\keygen(1^\secp,\{\widetilde{\Sigma}[\cM]\}_{x,y,z\in[\secp]})\\
    \psi_k\la\Comb.\StateGen(1^\secp,k,\{\widetilde{\Sigma}[\cM]\}_{x,y,z\in[\secp] })
    \end{array}
    \right] 
    \geq
    1-\negl(\secp).
\end{align}
This implies that $\Sigma_{\mathsf{Univ}}$ satisfies correctness.

\end{proof}

\begin{proof}[Proof of \cref{lem:univ_owsg_sec}]

Assume that there exists an OWSG.
Then, from \cref{prop:owsg_time}, there exists a set of classical Turing machines $(x^*,y^*,z^*)$ such that $x^*,y^*,z^*\in[n]$ for some $n\in\N$, and $x^*(1^\secp)$, $y^*(1^\secp)$, and $z^*(1^\secp)$ halt within $\secp^3$ steps for all sufficiently large security parameters $\secp\in\N$, and moreover $(\keygen[x^*],\StateGen[y^*],\Vrfy[z^*])$ is an OWSG scheme that satisfies correctness and security.
Furthermore, $(\widetilde{\keygen}[x^*],\widetilde{\StateGen}[y^*],\widetilde{\Vrfy}[z^*])$ also satisfies correctness and security
because for all sufficiently large security parameters, $(\widetilde{\keygen}[x^*],\widetilde{\StateGen}[y^*],\widetilde{\Vrfy}[z^*])$ emulates a correct-and-secure OWSG scheme $(\keygen[x^*],\StateGen[y^*],\Vrfy[z^*])$.
Therefore, for any polynomial $t$ and QPT adversary $\cA$, we have
\begin{align}
\Pr\left[\top\la\Comb.\Vrfy(1^\secp,k^*,\psi_k,\{\widetilde{\Sigma}[\cM]\}_{x,y,z\in[\secp]})
\ :
\begin{array}{ll}
k\la\Comb.\keygen(1^\secp,\{\widetilde{\Sigma}[\cM]\}_{x,y,z\in[\secp]})\\
\psi_k\la\Comb.\StateGen(1^\secp,k,\{\widetilde{\Sigma}[\cM]\}_{x,y,z\in[\secp]})\\
k^*\la\cA(\psi_k^{\otimes t(\secp)})
\end{array}
\right]
\leq \negl(\secp).
\end{align}
This is because $\Comb.\Sigma$ satisfies security and $\{(\widetilde{\keygen}[x],\widetilde{\StateGen}[y],\widetilde{\Vrfy}[z])\}_{x,y,z\in[\secp]}$ includes $(\widetilde{\keygen}[x^*],\widetilde{\StateGen}[y^*],\allowbreak\widetilde{\Vrfy}[z^*])$ for all sufficiently large $\secp\in\N$.

Furthermore, from the construction of $(\keygen_{\mathsf{Univ}},\StateGen_{\mathsf{Univ}},\Vrfy_{\mathsf{Univ}})$,
for all polynomial $t$, QPT adversary $\cA$, and security parameters $\secp\in\N$, we have
\begin{align}
    &\Pr\left[\top\la\Comb.\Vrfy(1^\secp,k^*,\psi_k,\{\widetilde{\Sigma}[\cM]\}_{x,y,z\in[\secp]})
\ :
\begin{array}{ll}
k\la\Comb.\keygen(1^\secp,\{\widetilde{\Sigma}[\cM]\}_{x,y,z\in[\secp]})\\
\psi_k\la\Comb.\StateGen(1^\secp,k,\{\widetilde{\Sigma}[\cM]\}_{x,y,z\in[\secp]})\\
k^*\la\cA(\psi_k^{\otimes t(\secp)})
\end{array}
\right]\\
&=\Pr\left[\top\la\Vrfy_{\mathsf{Univ}}(1^\secp,k^*,\psi_k)
\ :
\begin{array}{ll}
k\la\keygen_{\mathsf{Univ}}(1^\secp)\\
\psi_k\la\StateGen_{\mathsf{Univ}}(1^\secp,k)\\
k^*\la\cA(\psi_k^{\otimes t(\secp)})
\end{array}
\right].
\end{align}
Therefore, our universal construction $\Sigma_{\Univ}$ satisfies security.

\end{proof}
\end{proof}

\color{black}


\section{Robust Combiner for Public-Key Quantum Money Mini-Scheme}\label{sec:money}
\begin{definition}[Robust Combiner for Public-Key Quantum Money Mini-Scheme]
A robust combiner for public-key quantum money mini-scheme is a deterministic classical polynomial-time Turing machine $\cM$ with the following properties:
\begin{itemize}
    \item $\cM$ takes as input $1^n$ with $n\in\N$ and $n$-candidates for public-key quantum money mini-schemes $\{\Sigma_i\seteq(\Mint_i,\Vrfy_i)\}_{i\in[n]}$ promised that all candidates satisfy efficiency, and outputs a single set of algorithms $\Sigma\seteq(\Mint,\Vrfy)$.
    \item If all of $\{\Sigma_i\}_{i\in[n]}$ satisfy efficiency and at least one of $\{\Sigma_i\}_{i\in[n]}$ satisfies both correctness and security, then $\Sigma$ is a public-key quantum money mini-scheme that satisfies efficiency, correctness, and security.
\end{itemize}
\end{definition}

\begin{theorem}\label{thm:combiner_money}
    A robust combiner for public-key quantum money mini-scheme exists.
\end{theorem}
For proving \cref{thm:combiner_money}, we introduce the following \cref{lem:amp_cor_money}.
\begin{lemma}\label{lem:amp_cor_money}
    Let $\Sigma=(\mathsf{Mint},\Vrfy)$ be a candidate for public-key quantum money mini-scheme.
    From $\Sigma$, we can construct a public-key quantum money mini-scheme $\Sigma^*=(\mathsf{Mint}^*,\Vrfy^*)$ with the following properties:
    \begin{enumerate}
        \item If $\Sigma$ is a uniform QPT algorithm, then $\Sigma^*$ is a uniform QPT algorithm.
        \item $\Sigma^*$ satisfies correctness.
        \item If $\Sigma$ is a uniform QPT algorithm and satisfies both correctness and security, then $\Sigma^*$ satisfies security.
    \end{enumerate}
\end{lemma}
We describe the proof of \cref{lem:amp_cor_money} in \cref{sec:app_money}.

\begin{proof}[Proof of \cref{thm:combiner_money}]
    Below, we consider a fixed constant $n$.
    Let us introduce some notations.
\paragraph{Notations.}
\begin{itemize}
    \item Let $\Sigma_i$ be a candidate of public-key quantum money mini-scheme for $i\in[n]$.
    \item For a candidate of public-key quantum money mini-scheme $\Sigma_i$, let $\Sigma_i^*\seteq(\Mint_i^*,\Vrfy_i^*)$ be a candidate of public-key quantum money mini-scheme derived from \cref{lem:amp_cor_money}, which satisfies:
    \begin{itemize}
        \item $\Sigma_i^*$ is a uniform QPT algorithm if $\Sigma_i$ is a uniform QPT algorithm.
        \item $\Sigma_i^*$ satisfies correctness.
        \item $\Sigma_i^*$ satisfies security if $\Sigma_i$ is a uniform QPT algorithm and satisfies both correctness and security.
    \end{itemize}
\end{itemize}

    \paragraph{Construction of Robust Combiner for Public-Key Quantum Money Mini-Scheme:}
    A robust combiner for public-key quantum money mini-scheme is a deterministic classical polynomial-time Turing machine $\cM$ that takes as input $1^n$ and $\{\Sigma_i\}_{i\in[n]}$, and outputs the following set of algorithms $\Sigma=(\Mint,\Vrfy)$:
\begin{description}
\item[$\mathsf{Mint}(1^\secp)$:]$ $ 
    \begin{itemize}
        \item For all $i\in[n]$, run $(s_i^*,\rho_{s_i}^*)\la \Mint_i^*(1^\secp)$.
        \item Output $s\seteq \{s_i^*\}_{i\in[n]}$ and $\rho_s\seteq\bigotimes_{i\in[n]}\rho_{s_i}^*$.
    \end{itemize}
    \item[$\Vrfy(1^\secp,s,\rho)$:]$ $
    \begin{itemize}
        \item Parse $s= \{s_i\}_{i\in[n]}$. Let $\rho$ be a quantum state on $n$ registers, $\{R[i]\}_{i\in[n]}$, each of which is of $\abs{\rho_{s_i}}$ qubits.
        \item For all $i\in[n]$, run $\Vrfy_{i}^*(1^\secp,s_i,\cdot)$ on the $R[i]$ register and obtain $b[i]$.
        If $b[i]=\top$ for all $i\in[n]$, output $\top$. Otherwise, output $\bot$.
    \end{itemize}
\end{description}
\cref{thm:combiner_money} follows from the following \cref{lem:money_eff,lem:money_cor,lem:money_sec}.
\begin{lemma}\label{lem:money_eff}
    If all of $\{\Sigma_i\}_{i\in[n]}$ satisfies efficiency, then $\Sigma$ satisfies efficiency.
\end{lemma}
\begin{lemma}\label{lem:money_cor}
    $\Sigma$ satisfies correctness.
\end{lemma}
\begin{lemma}\label{lem:money_sec}
    If all of $\{\Sigma_i\}_{i\in[n]}$ satisfies efficiency and at least one of $\{\Sigma_i\}_{i\in[n]}$ satisfies both correctness and security, then $\Sigma$ satisfies security.
\end{lemma}
\cref{lem:money_eff} trivially follows.
\cref{lem:money_cor} follows because $\Sigma_i^*$ satisfies correctness for all $i\in[n]$.
\begin{proof}[Proof of \cref{lem:money_sec}]
We can prove \cref{lem:money_sec} via a standard hybrid argument.
For a reader's convenience, we describe the proof.
Let $\Sigma_x$ be a candidate of public-key quantum money mini-scheme that satisfies both correctness and security.
Then, $\Sigma_x^*$ satisfies security from \cref{lem:amp_cor_money}.
Assume that there exists a QPT adversary $\cA$ that breaks the security of $\Sigma$, and then construct an adversary $\cB_x$ that breaks the security of $\Sigma_x^*$.
We describe $\cB_x$:
\begin{enumerate}
    \item $\cB_x$ receives $\rho_{s_x}^*$ from the challenger of $\Sigma_x^*$.
    \item $\cB_x$ runs $(s_i^*,\rho_{s_i}^*)\la\Mint_i^*(1^\secp)$ for all $i\in[n]\setminus \{x\}$, and sends $\left(\{s_i^*\}_{i\in[n]},\{\rho_{s_i}^*\}_{i\in[n]}\right)$ to $\cA$.
    \item $\cB_x$ receives $\sigma$ from $\cA$. Here, $\sigma$ is a quantum state on $2n$ registers, $\{R[i]\}_{i\in[2n]}$, where $R[i]$ and $R[i+n]$ are registers over $\abs{\rho_{s_i}^*}$-length qubits.
    \item For $i\in[n]\setminus \{x\}$, $\cB_x$ runs $\Vrfy_i^*(1^\secp,s_i^*,\cdot)$ on the $R[i]$ and $R[i+n]$ registers, and obtains $b[i]$ and $b[i+n]$, respectively.
    \item $\cB_x$ sends the $R[x]$ and $R[x+n]$ registers to the challenger of $\Sigma_x^*$.
    \item The challenger runs $\Vrfy_x^*(1^\secp,s_x^*,\cdot)$ on the $R[x]$ and $R[x+n]$ registers, and obtains $b[x]$ and $b[x+n]$, respectively.
    If $b[x]=b[x+n]=\top$, then the challenger outputs $\top$.
\end{enumerate}
Clearly, $\cB_x$ perfectly simulates the challenge of $\Sigma$.
Because $\cA$ breaks the security of $\Sigma$, $\cA$ outputs $\sigma$ such that $b[i]=b[i+n]=\top$ for all $i\in[n]$ with non-negligible probability. 
Therefore, the challenger outputs $\top$ with non-negligible probability, which implies that $\cB_x$ breaks the security of $\Sigma_x^*$.
This completes the proof.
\end{proof}
\end{proof}

\subsection{Universal Construction}
\begin{definition}
    We say that a set of uniform algorithms $\Sigma_{\mathsf{Univ}}=(\Mint,\Vrfy)$ is a universal construction of public-key quantum money mini-scheme if $\Sigma_{\mathsf{Univ}}$ is a public-key quantum money mini-scheme as long as there exists a public-key quantum money mini-scheme.
\end{definition}

\begin{theorem}\label{thm:univ_money}
    There exists a universal construction of public-key quantum money mini-scheme.
\end{theorem}
The proof is almost the same as \cref{thm:univ_owsg}, and thus we skip the proof.

\section{Robust Canonical Quantum Bit Commitment Combiner}\label{sec:commit}
\begin{definition}[Robust Canonical Quantum Bit Commitment Combiner]
A robust canonical quantum bit commitment combiner is a deterministic classical polynomial-time Turing machine $\cM$ with the following properties:
\begin{itemize}
    \item $\cM$ takes as input $1^n$ and $n$-deterministic classical polynomial-time Turing machine $\{\cT_i\}_{i\in[n]}$ that produces unitary, and outputs a deterministic classical polynomial-time Turing machine $\cT$ that produces unitary.
    \item Let $\left(Q_{i,0}(\secp),Q_{i,1}(\secp)\right)$ be the unitary obtained by $\cT_i(\secp)$ and let $\left(Q_{0}(\secp),Q_1(\secp)\right)$ be the unitary obtained by $\cT(\secp)$.
    If one of $\{ \{Q_{i,0} (\secp),Q_{i,1}(\secp)\}_{\secp\in\N} \}_{i\in[n]}$ satisfies  
    computational binding and computational hiding, then $\{Q_{0}(\secp),Q_1(\secp)\}_{\secp\in\N}$ is a quantum bit commitment that satisfies statistical binding and computational hiding.
\end{itemize}
\end{definition}

In this section, we show the \cref{thm:combiner_com}.
\begin{theorem}\label{thm:combiner_com}
    There exists a robust canonical quantum bit commitment combiner.
\end{theorem}
First, let us introduce the following \cref{lem:amp_binding_com}.
\begin{proposition}
\label{lem:amp_binding_com}
    Let $\Sigma=\{Q_0(\secp),Q_1(\secp)\}_{\secp\in\N}$ be a candidate of a canonical quantum bit commitment. From $\Sigma$, we can construct a canonical quantum bit commitment $\Sigma^*\seteq \{Q_0^*(\secp),Q_1^*(\secp)\}_{\secp\in\N}$ such that:
    \begin{enumerate}
        \item $\Sigma^*$ satisfies statistical binding. 
        \item If $\Sigma$ satisfies computational binding and computational hiding,  then $\Sigma^*$ satisfies computational hiding.
    \end{enumerate}
\end{proposition}

\cref{lem:amp_binding_com} directly follows from the following \cref{lem:amp}.
\begin{lemma}[Amplifying Binding]\label{lem:amp}
    Let $\{Q_{0}(\secp),Q_{1}(\secp)\}_{\secp\in\N}$ be a candidate of canonical quantum bit commitment.
    Let $\{Q_{0}^*(\secp),Q_{1}^*(\secp)\}_{\secp\in\N}$ be a candidate of canonical quantum bit commitment described as follows:
    \begin{itemize}
        \item If $\mathbf{C}$ and $\mathbf{R}$ are the commitment and reveal registers of $(Q_{0}(\secp),Q_{1}(\secp))$, and $\mathbf{\widetilde{C}}$ and $\mathbf{\widetilde{R}}$ are the commitment and reveal registers of $(\widetilde{Q_{0}}(\secp),\widetilde{Q_{1}}(\secp))$, which is the flavor conversion of $(Q_{0}(\secp),Q_{1}(\secp))$ introduced in \cref{lem:converting_flavor}, then the commitment and reveal registers of $(Q^*_{0}(\secp),Q_{1}^*(\secp))$ are defined as $\mathbf{C}^*\seteq \left(\mathbf{C}^{\otimes \secp},\mathbf{\widetilde{C}}^{\otimes \secp} \right)$, and $\mathbf{R^*}\seteq \left(\mathbf{R}^{\otimes \secp},\widetilde{\mathbf{R}}^{\otimes \secp}\right)$.
        \item For $b\in\bit$, the unitary $Q_b^*(\secp)$ is defined as follows:
        \begin{align}
            Q_{b}^*(\secp)\seteq (Q_{b}(\secp)\otimes \widetilde{Q_{b}}(\secp))^{\otimes \secp}.
        \end{align}
    Then, the following is satisfied:
    \begin{enumerate}
        \item $\{Q_{0}^*(\secp),Q_{1}^*(\secp)\}_{\secp\in\N}$ satisfies statistical binding.
        \item If $\{Q_{0}(\secp),Q_{1}(\secp)\}_{\secp\in\N}$ satisfies  computational hiding and computational binding, then $\{Q_{0}^*(\secp),Q_{1}^*(\secp)\}_{\secp\in\N}$ satisfies computational hiding.
    \end{enumerate}
    \end{itemize}
\end{lemma}

\begin{proof}[Proof of \cref{lem:amp}]
Below, we fix the security parameter $\secp$, and write $(Q_0,Q_1)$, $(\widetilde{Q_0},\widetilde{Q_1})$ and $(Q_0^*,Q_1^*)$ to mean $(Q_0(\secp),Q_1(\secp))$, $(\widetilde{Q_0}(\secp),\widetilde{Q_1}(\secp))$ and $(Q_0^*(\secp),Q_1^*(\secp))$, respectively. 
\paragraph{Proof of the first item}
    We define 
\begin{align}
    \rho_{b}\seteq\Tr_{\mathbf{R}}(Q_{b}\ket{0}_{\mathbf{C,R}})\mbox{\,\,\,and\,\,\,}
    \widetilde{\rho_{b}}\seteq \Tr_{\mathbf{\widetilde{R}}}(\widetilde{Q_{b}}\ket{0}_{\mathbf{\widetilde{C},\widetilde{R}}})\mbox{\,\,\,and\,\,\,}\rho_b^*\seteq \Tr_{\mathbf{R^*}}(Q_b^*\ket{0}_{\mathbf{C^*,R^*}}).
\end{align}
From the construction of $Q_0^*$ and $Q_1^*$, we have
\begin{align}
    (\rho_b\otimes \widetilde{\rho_b})^{\otimes \secp}\seteq \Tr_{\mathbf{R^*}}(Q_b^*\ket{0}_{\mathbf{C^*,R^*}}).
\end{align}
Let $0\leq f\leq 1$ be some value such that 
\begin{align}
    F(\rho_{0},\rho_{1})=f.
\end{align}

We have
\begin{align}
    \mathsf{TD}(\rho_{0},\rho_{1})\leq \sqrt{1-F(\rho_{0},\rho_{1})}\leq \sqrt{1-f}.
\end{align}
In particular, this implies that $(Q_{0},Q_{1})$ satisfies $\sqrt{1-f}$-statistical hiding.
From \cref{lem:converting_flavor}, this implies that $(\widetilde{Q_{0}},\widetilde{Q_{1}})$ satisfies $(1-f)^{1/4}$-statistical binding.
Furthermore, from Uhlmann's theorem (\cref{lem:uhlmann}), this implies that
\begin{align}
    F(\widetilde{\rho_{0}},\widetilde{\rho_{1}})\leq (1-f)^{1/2},
\end{align}
which we prove later.

Therefore, we have
\begin{align}
    &F(\rho_0^*,\rho_1^*)=F\left((\rho_{0}\otimes\widetilde{\rho_{0}})^{\otimes \secp},(\rho_{1}\otimes\widetilde{\rho_{1}})^{\otimes \secp}\right)
    \leq F(\rho_{0},\rho_{1})^{\secp}F(\widetilde{\rho_{0}},\widetilde{\rho_{1}})^{\secp}\leq f^\secp(1-f)^{\secp/2}\leq 2^{-\secp/2},
\end{align}
which implies that $(Q_{0}^*,Q_{1}^*)$ satisfies statistical binding.

Now, we show that if $(\widetilde{Q_{0}},\widetilde{Q_{1}})$ satisfies $(1-f)^{1/4}$-statistical binding, then
$
    F(\widetilde{\rho_{0}},\widetilde{\rho_{1}})\leq (1-f)^{1/2}.
$
For contradiction, assume that 
$
    F(\widetilde{\rho_{0}},\widetilde{\rho_{1}})>(1-f)^{1/2},
$
and then show that $(\widetilde{Q_{0}},\widetilde{Q_{1}})$ does not satisfies $(1-f)^{1/4}$-statistical binding.
From Uhlmann's Theorem (\cref{lem:uhlmann}), there exists some unitary $U_{\mathbf{\widetilde{R}}}$ acting on the register $\mathbf{\widetilde{R}}$ such that
\begin{align}
F(\widetilde{\rho_{0}},\widetilde{\rho_{1}})=
\abs{\bra{0}_{\mathbf{\widetilde{C},\widetilde{R}}}\widetilde{Q_0}^{\dagger}(I_{\mathbf{\widetilde{C}}}\otimes U_{\mathbf{\widetilde{R}}})\widetilde{Q_1}\ket{0}_{\mathbf{\widetilde{C},\widetilde{R}}}}^2>(1-f)^{1/2}.
\end{align}
Now, we have
\begin{align}
    &\norm{\left(\bra{0}_{\mathbf{\widetilde{C},\widetilde{R}}}\widetilde{Q_0}^{\dagger}\otimes I_{\mathbf{\widetilde{C},\widetilde{Z}}}\right)(I_{\mathbf{\widetilde{C},\widetilde{Z}}}\otimes U_{\mathbf{\widetilde{R}}})\widetilde{Q_1}\ket{0}_{\mathbf{\widetilde{C},\widetilde{R}}}\ket{\tau}_{\mathbf{\widetilde{Z}}}}_1\\
    &=\norm{\left(\bra{0}_{\mathbf{\widetilde{C},\widetilde{R}}}\widetilde{Q_0}^{\dagger}(I_{\mathbf{\widetilde{C}}}\otimes U_{\mathbf{\widetilde{R}}})\widetilde{Q_1}\ket{0}_{\mathbf{\widetilde{C},\widetilde{R}}}\right) \ket{\tau}_{\mathbf{\widetilde{Z}}}}_1
    =\abs{\bra{0}_{\mathbf{\widetilde{C},\widetilde{R}}}\widetilde{Q_0}^{\dagger}(I_{\mathbf{\widetilde{C}}}\otimes U_{\mathbf{\widetilde{R}}})\widetilde{Q_1}\ket{0}_{\mathbf{\widetilde{C},\widetilde{R}}}}>(1-f)^{\frac{1}{4}},
\end{align}
which contradicts that$(\widetilde{Q_{0}},\widetilde{Q_{1}})$ satisfies $(1-f)^{1/4}$-statistical binding.

\paragraph{Proof of the second item.}
We prove that $(Q_{0}^*,Q_{1}^*)$ satisfies computational hiding if $(Q_0,Q_1)$ satisfies computational hiding and computational binding.
Because $(\widetilde{Q_0},\widetilde{Q_1})$ is the flavor conversion of $(Q_0,Q_1)$, $(\widetilde{Q}_0,\widetilde{Q}_1)$ also satisfies computational hiding.
Therefore, we can reduce the computational hiding of $(Q_{0}^*,Q_{1}^*)$ to those of $(Q_{0},Q_{1})$ and $(\widetilde{Q}_0,\widetilde{Q}_1)$ by a standard hybrid argument.
\end{proof}

\begin{proof}[Proof of \cref{thm:combiner_com}]
Below, we consider some fixed constant $n$.
For $i\in[n]$, let $\cT_i$ be a deterministic classical Turing machine that takes as input $1^\secp$, and outputs $(Q_{i,0}(\secp),Q_{i,1}(\secp))$.
Let $\Sigma_i\seteq \{Q_{i,0}(\secp),Q_{i,1}(\secp)\}_{\secp\in\N}$ be a candidate of canonical quantum bit commitment.
Let $\Sigma_i^*\seteq \{Q_{i,0}^*(\secp),Q_{i,1}^*(\secp)\}_{\secp\in\N}$ be a candidate of canonical quantum bit commitment such that: 
\begin{enumerate}
    \item $\Sigma_i^*$ satisfies statistical binding.
    \item $\Sigma_i^*$ satisfies computational hiding if $\Sigma_i$ satisfies computational hiding and computational binding.
\end{enumerate}
Note that such a canonical quantum bit commitment is obtained from \cref{lem:amp_binding_com}.

A robust canonical quantum bit commitment combiner is a deterministic classical polynomial-time Turing machine $\cM$ that takes as input $1^n$ and $\{\cT_i\}_{i\in[n]}$, and outputs a deterministic classical polynomial-time Turing machine $\cT$ that works as follows.
$\cT$ takes as input $1^\secp$ and outputs the following QPT unitary $\left(\Comb.Q_0(\secp), \Comb.Q_1(\secp)\right)$:
\begin{itemize}
    \item If $\mathbf{C}^*_i$ and $\mathbf{R}^*_i$ are the commitment register and the reveal register of $(Q_{i,0}^*(\secp),Q_{i,1}^*(\secp))$, then the commitment and reveal register of $(\Comb.Q_0(\secp),\Comb.Q_1(\secp))$ are defined as $\mathbf{C}\seteq \{\mathbf{C}^*_i\}_{i\in[n]}$ and $\mathbf{R}=(\{\mathbf{R}^*_i\}_{i\in[n]},\{\mathbf{D}^*_i\}_{i\in[n]})$, where $\mathbf{D}^*_i$ is an additional one-qubit register for $i\in[n]$.
    \item For $b\in\bit$, the unitary $\Comb.Q_b$ is defined as follows:
    \begin{align}
        \Comb.Q_b(\secp)\seteq\left(\sum_{r\in\bit^n}\bigotimes_{i\in[n]} (Q^*_{i,r_i}(\secp)\otimes \ket{r_i}\bra{r_i}_{\mathbf{D}^*_i})\right)\left(\bigotimes_{i\in[n]}I_{\mathbf{C}^*_i,\mathbf{R}^*_i}\otimes X^b_{\mathbf{D}_1^*}\bigotimes_{i\in \{2,\cdots,n\}}  \mathsf{CNOT}_{\mathbf{D}_1^*,\mathbf{D}^*_i}\bigotimes_{i\in\{2,\cdots,n\}} H_{\mathbf{D}^*_i}\right).
    \end{align}
    Here, $r_i$ is the $i$-th bit of $r$ and $\mathsf{CNOT}_{\mathbf{D}^*_1,\mathbf{D}^*_i}$ is a CNOT gate, where $\mathbf{D}^*_1$ is a target register and $\mathbf{D}^*_i$ is a control register.
Note that we have
\begin{align}
    \Comb.Q_b(\secp)\ket{0}_{\mathbf{C,R}}=\frac{1}{2^{(n-1)/2}}\sum_{\{r: \sum_{i\in[n]}r_i=b\}}\bigotimes_{i\in[n]}(Q_{i,r_i}^*(\secp)\ket{0}_{\mathbf{C}^*_i,\mathbf{R}^*_i}\otimes \ket{r_i}_{\mathbf{D}^*_i}).
\end{align}
\end{itemize}

We have the following \cref{thm:binding,thm:preserve}, which we prove later.
Therefore, \cref{thm:combiner_com} holds.
\begin{lemma}\label{thm:binding}
    $\{\Comb.Q_0(\secp),\Comb.Q_1(\secp)\}_{\secp\in\N}$ satisfies statistical binding.
\end{lemma}
\begin{lemma}\label{thm:preserve}
    If one of $\{\Sigma_i\}_{i\in[n]}$ satisfies computational hiding and computational binding, then $\{\Comb.Q_0(\secp),\Comb.Q_1(\secp)\}_{\secp\in\N}$ satisfies computational hiding.
\end{lemma}
\end{proof}

\begin{proof}[Proof of \cref{thm:binding}]
    Below, we fix security parameter $\secp$, and write $\{Q_{i,0}^*,Q_{i,1}^*\}_{i\in[n]}$ and $(\Comb.Q_0,\Comb.Q_1)$ to mean $\{Q_{i,0}^*(\secp),Q_{i,1}^*(\secp)\}_{i\in[n]}$ and $(\Comb.Q_0(\secp),\Comb.Q_1(\secp) )$, respectively.
    Let us denote ${\rho_{i,b}^*}\seteq \Tr_{\mathbf{R}_i}(Q^*_{i,b}\ket{0}_{\mathbf{C}_i,\mathbf{R}_i})$.
    We write $\Comb.\rho_b\seteq \Tr_{\mathbf{R}}(\Comb.Q_b\ket{0}_{\mathbf{C,R}})$ and write $R$ to mean $\sum_{i\in[n]}r_i$.
    Note that we have 
    \begin{align}
        \Comb.\rho_b= \frac{1}{2^{n-1}}\sum_{\{r:R=b \}}\bigotimes_{i\in[n]}\rho_{i,r_i}^*.
    \end{align}
    Now, we show that 
    \begin{align}
    \mathsf{TD}(\Comb.\rho_0,\Comb.\rho_1)\geq 1-\negl(\secp).
    \end{align}
    For that, it is sufficient to show that there exists a POVM measurement $\{\Comb.\Pi_0,\Comb.\Pi_1\}$ that distinguishes $\Comb.\rho_0$ from $\Comb.\rho_1$.
    From \cref{lem:amp}, all $\Sigma_i^*$ satisfies statistical binding. This implies that we have $\mathsf{TD}(\rho_{i,0}^*,\rho_{i,1}^*)\geq 1-\negl(\secp)$.
    Moreover, this implies that there exists a two-outcome POVM measurement $\{\Pi_{i,0}^*,\Pi_{i,1}^*\}$ such that
    \begin{align}
        \Tr\left(\Pi_{i,0}^*(\rho_{i,0}^*-\rho_{i,1}^*)\right)=\Tr\left(\Pi_{i,1}^*(\rho_{i,1}^*-\rho_{i,0}^*)\right)\geq 1-\negl(\secp).
    \end{align}
    We introduce the two-outcome POVM measurement $\{\Comb.\Pi_0\seteq\sum_{\{r:R=0\}}\bigotimes_{i\in[n]}\Pi_{i,r_i}^*,\Comb.\Pi_1\seteq\sum_{\{r:R=1\}}\bigotimes_{i\in[n]}\Pi_{i,r_i}^*\}$.
    Then, we have
    \begin{align}
        \mathsf{TD}(\Comb.\rho_0,\Comb.\rho_1)&\geq\Tr(\Comb.\Pi_0(\Comb.\rho_0-\Comb.\rho_1))\\
        &=
        \frac{1}{2^{n-1}}\Tr(\Comb.\Pi_0\left(\sum_{\{r:R=0\}}\bigotimes_{i\in[n]}\rho_{i,r_i}^*-\sum_{\{r:R=1\}}\bigotimes_{i\in[n]}\rho_{i,r_i}^*\right))\\
        &=\frac{1}{2^{n-1}}\Tr(\Comb.\Pi_0\left(\sum_{r}\bigotimes_{i\in[n]}(-1)^{r_i}\rho_{i,r_i}^*\right))\\
        &=\frac{1}{2^{n-1}}\Tr(\Comb.\Pi_0\bigotimes_{i\in[n]}\left(\rho^*_{i,0}-\rho_{i,1}^*\right))\\
        &=\frac{1}{2^{n-1}}\Tr(\left(\sum_{\{r:R=0\}}\bigotimes_{i\in[n]}\Pi^*_{i,r_i}\right)\left(\bigotimes_{i\in[n]}\rho^*_{i,0}-\rho_{i,1}^*\right))\\
        &=\frac{1}{2^{n-1}}\Tr(\sum_{\{r:R=0\}}\left(\bigotimes_{i\in[n]}\Pi^*_{i,r_i}\left(\rho^*_{i,0}-\rho_{i,1}^*\right)\right))\\
        &=\frac{1}{2^{n-1}}\sum_{\{r:R=0\}}\Tr(\bigotimes_{i\in[n]}\Pi^*_{i,r_i}\left(\rho^*_{i,0}-\rho_{i,1}^*\right))\\
        &=\frac{1}{2^{n-1}}\sum_{\{r:R=0\}}\prod_{i\in[n]}\Tr(\Pi^*_{i,r_i}\left(\rho^*_{i,0}-\rho_{i,1}^*\right))\\
       &=\prod_{i\in[n]}\Tr(\Pi^*_{i,0}\left(\rho^*_{i,0}-\rho_{i,1}^*\right))\geq (1-\negl(\secp))^{n}\geq 1-n\cdot\negl(\secp). 
    \end{align}
    Here, we have used that $\Tr(A+B)=\Tr(A)+\Tr(B)$ in the sixth equation, and we have used $\Tr(A\otimes B)=\Tr(A)\Tr(B)$ in the seventh equation, and we have used that
    \begin{align}
        \prod_{i\in[n]}\Tr(\Pi_{i,r_i}^*\left(\rho_{i,0}^*-\rho_{i,1}^* \right))=\prod_{i\in[n]}\Tr(\Pi_{i,0}^*\left(\rho_{i,0}^*-\rho_{i,1}^* \right))
    \end{align}
    for all $r\in\bit^n$ with $\sum_{i\in[n]}r_i=0$ in the final equation.
    
    Furthermore, we have
    $F(\Comb.\rho_0,\Comb.\rho_1)\leq 1-\mathsf{TD}(\Comb.\rho_0,\Comb.\rho_1)^2\leq 2n\cdot\negl(\secp)$.
    From Uhlmann's theorem (\cref{lem:uhlmann}), this implies that $\{\Comb.Q_0(\secp),\Comb.Q_1(\secp)\}_{\secp\in\N}$ satisfies statistical binding.
\end{proof}

\begin{proof}[Proof of \cref{thm:preserve}]
    Below, we fix security parameter $\secp$, and write $\{Q_{i,0}^*,Q_{i,1}^*\}_{i\in[n]}$ and $(\Comb.Q_0,\Comb.Q_1)$ to mean $\{Q_{i,0}^*(\secp),Q_{i,1}^*(\secp)\}_{i\in[n]}$ and $(\Comb.Q_0(\secp),\Comb.Q_1(\secp) )$, respectively.
    Let $\rho_{i,b}^*\seteq \Tr_{\mathbf{R}_i}(Q_{i,b}^*\ket{0}_{\mathbf{C}_i,\mathbf{R}_i}) $.
    We show that $\{\Comb.Q_0(\secp),\Comb.Q_1(\secp)\}_{\secp\in\N}$ satisfies computational hiding
    as long as one of $\{\Sigma_i\}_{i\in[n]}$ satisfies computational hiding and computational binding.
    Let $\Sigma_x$ be the canonical quantum bit commitment that satisfies computational hiding and computational binding.
    Then, from \cref{lem:amp}, $\Sigma_{x}^*$ satisfies computational hiding.
    Now, we introduce the following sequence of hybrid experiments against QPT adversary $\cA$.
    \begin{description}
        \item[$\mathsf{Hyb}_0(b)$:]$ $
        \begin{enumerate}
            \item The challenger sends $\Comb.\rho_b$ to $\cA$.
            \item $\cA$ outputs $b^*$.
        \end{enumerate}
        \item[$\mathsf{Hyb}_1(b)$:]$ $
        \begin{enumerate}
            \item The challenger randomly samples $r_{i}\la\bit$ for all $i\in[n]\backslash x$. We
            write $R$ to mean $\sum_{i\in[n]\backslash x}r_i$.
            \item The challenger sends $\rho_{1,r_1}^*\otimes \cdots \otimes\rho_{x-1,r_{x-1}}^*\otimes \rho_{x,R+b}^*\otimes \rho_{x+1,r_{x+1}}^*\otimes\cdots \otimes\rho_{n,r_{n}}^*$ to $\cA$.
            \item $\cA$ outputs $b^*$.
        \end{enumerate}
        \item[$\mathsf{Hyb}_2(b)$:] $ $
        \begin{enumerate}
            \item The challenger randomly samples $r_i\la\bit$ for all $i\in[n]\backslash x$.
            \item The challenger sends $\rho_{1,r_1}^*\otimes\cdots\otimes \rho_{x-1,r_{x-1}}^*\otimes\rho_{x,0}^* \otimes\rho_{x+1,r_{x+1}}^* \otimes\cdots\otimes\rho_{n,r_{n}}^*$.
            \item $\cA$ outputs $b^*$.
        \end{enumerate}
    \end{description}
    We have the following \cref{prop:hyb_0_hyb_1_com,prop:hyb_1_hyb_2_com,prop:hyb_2_com}. Therefore, we have
    \begin{align}
        \abs{\Pr[\mathsf{Hyb}_0(0)=1]-\Pr[\mathsf{Hyb}_0(1)=1]}\leq \negl(\secp),
    \end{align}
    which implies that $\{\Comb.\rho_0(\secp),\Comb.\rho_1(\secp)\}_{\secp\in\N}$ satisfies computational hiding.
    \begin{proposition}\label{prop:hyb_0_hyb_1_com}
        $\Pr[\mathsf{Hyb}_0(b)=1]=\Pr[\mathsf{Hyb}_1(b)=1]$ for $b\in\bit$.
    \end{proposition}
    \begin{proposition}\label{prop:hyb_1_hyb_2_com}
        If $\Sigma_x^*$ satisfies computational hiding, then 
        \begin{align}
            \abs{\Pr[\mathsf{Hyb}_1(b)=1]-\Pr[\mathsf{Hyb}_2(b)=1]}\leq \negl(\secp)
        \end{align}
        for each $b\in\bit$.
    \end{proposition}
    \begin{proposition}\label{prop:hyb_2_com}
        $\Pr[\mathsf{Hyb}_2(0)=1]=\Pr[\mathsf{Hyb}_2(1)=1]$.
    \end{proposition}
    
    \cref{prop:hyb_0_hyb_1_com,prop:hyb_2_com} trivially follows, and thus we omit the proof.

\begin{proof}[Proof of \cref{prop:hyb_1_hyb_2_com}]
    For simplicity, we write $\sum_r$ to mean that $\sum_{\{r:r_i\in\bit\mbox{\,\,for\,\,} i\in[n]\backslash x\}} $, and recall that $R\seteq \sum_{i\in[n]\backslash x}r_i$.

    Then, we have
    \begin{align}
        &\Pr[\mathsf{Hyb}_1(0)=1]-\Pr[\mathsf{Hyb}_2(0)=1]\\
        &=\frac{1}{2^{n-1}}\sum_{r}\left(\Pr[1\la\cA(\rho_{1,r_1}^*\otimes\cdots\otimes \rho_{x-1,r_{x-1}}^*\otimes\rho_{x,R}^*\otimes\rho_{x+1,r_{x+1}}^*\otimes \cdots\otimes \rho_{n,r_{n}}^*) ]\right)\\
        &\,\,\,\,\,\,\,\,\,\,\,\,\,\,\,\,\,\,\,\,-\frac{1}{2^{n-1}}\sum_{r}\left(\Pr[1\la\cA(\rho_{1,r_1}^*\otimes\cdots\otimes \rho_{x-1,r_{x-1}}^*\otimes\rho_{x,0}^*\otimes\rho_{x+1,r_{x+1}}^*\otimes \cdots\otimes \rho_{n,r_{n}}^*) ]\right)\\
        &=\frac{1}{2^{n-1}}\sum_{r}\bigg(\Pr[1\la\cA(\rho_{1,r_1}^*\otimes\cdots\otimes \rho_{x-1,r_{x-1}}^*\otimes\rho_{x,R}^*\otimes\rho_{x+1,r_{x+1}}^*\otimes \cdots\otimes \rho_{n,r_{n}}^*) ]\\
        &\,\,\,\,\,\,\,\,\,\,\,\,\,\,\,\,\,\,\,\,\,\,\,-\Pr[1\la\cA(\rho_{1,r_1}^*\otimes\cdots\otimes \rho_{x-1,r_{x-1}}^*\otimes\rho_{x,0}^*\otimes\rho_{x+1,r_{x+1}}^*\otimes \cdots\otimes \rho_{n,r_{n}}^*) ]\bigg)\\
        &=\frac{1}{2^{n-1}}\sum_{\{r:R=1\}}\bigg(\Pr[1\la\cA(\rho_{1,r_1}^*\otimes\cdots\otimes \rho_{x-1,r_{x-1}}^*\otimes\rho_{x,1}^*\otimes\rho_{x+1,r_{x+1}}^*\otimes \cdots\otimes \rho_{n,r_{n}}^*) ]\\
        &\,\,\,\,\,\,\,\,\,\,\,\,\,\,\,\,\,\,\,\,\,\,\,-\Pr[1\la\cA(\rho_{1,r_1}^*\otimes\cdots\otimes \rho_{x-1,r_{x-1}}^*\otimes\rho_{x,0}^*\otimes\rho_{x+1,r_{x+1}}^*\otimes \cdots\otimes \rho_{n,r_{n}}^*) ]\bigg)
        \\
        &=\frac{1}{2^{n-1}}\sum_{\{r:R=1\}}\bigg(\Pr[0\la\cA(\rho_{1,r_1}^*\otimes\cdots\otimes \rho_{x-1,r_{x-1}}^*\otimes\rho_{x,0}^*\otimes\rho_{x+1,r_{x+1}}^*\otimes \cdots\otimes \rho_{n,r_{n}}^*) ]\\
        &\,\,\,\,\,\,\,\,\,\,\,\,\,\,\,\,\,\,\,\,\,\,\,-\Pr[0\la\cA(\rho_{1,r_1}^*\otimes\cdots\otimes \rho_{x-1,r_{x-1}}^*\otimes\rho_{x,1}^*\otimes\rho_{x+1,r_{x+1}}^*\otimes \cdots\otimes \rho_{n,r_{n}}^*) ]\bigg).
    \end{align}
    For showing a contradiction, assume that there exists some constant $c$ and a QPT adversary $\cA$ such that
    \begin{align}
        \abs{\Pr[\mathsf{Hyb}_1(0)=1]-\Pr[\mathsf{Hyb}_2(0)=1]}\geq  1/\secp^c
    \end{align}
    for all sufficiently large security parameters $\secp\in\N$
    and then construct a QPT algorithm $\cB_x$ that breaks the computational hiding of $\Sigma_x^*$.
    \begin{enumerate}
        \item $\cB_x$ receives $\rho_{x,b}^*$ from the challenger of $\Sigma_x^*$, where $b$ is randomly sampled from $\bit$.
        \item $\cB_x$ randomly samples $r_i\la\bit$ for all $i\in[n]\backslash x$.
        \item $\cB_x$ sends $\rho_{1,r_1}^*\otimes\cdots\otimes\rho_{x-1,r_{x-1}}^*\otimes\rho_{x,b}^* \otimes\rho_{x+1,r_{x+1}}^*\otimes\cdots \otimes\rho_{n,r_{n}}^* $ to $\cA$.
        \item $\cB_x$ receives $b^*$ from $\cA$.
        \item $\cB_x$ outputs $b^*+1$ if $R=1$, and outputs $0$ otherwise, where $R= \sum_{i\in[n]\backslash x}r_i$.
    \end{enumerate}
    We compute $\abs{\Pr[1\la\cB_x:b=0]-\Pr[1\la\cB_x:b=1]}$.
    It holds that
    \begin{align}
        &\abs{\Pr[1\la\cB_x:b=0]-\Pr[1\la\cB_x:b=1]}\\
        &=\frac{1}{2}\abs{\Pr[0\la\cA: b=0,R=1]-\Pr[0\la\cA:b=1,R=1]}\\
        &=\frac{1}{2^{n}}\Bigg|\sum_{\{r:R=1\}}\left(\Pr[0\la\cA(\rho_{1,r_1}^*\otimes\cdots\otimes\rho_{x-1,r_{x-1}}^*\otimes\rho_{x,0}^* \otimes\rho_{x+1,r_{x+1}}^*\otimes\cdots \otimes\rho_{n,r_{n}}^*)]\right)\\
        &-\sum_{\{r:R=1\}}\left(\Pr[0\la\cA(\rho_{1,r_1}^*\otimes\cdots\otimes\rho_{x-1,r_{x-1}}^*\otimes\rho_{x,1}^* \otimes\rho_{x+1,r_{x+1}}^*\otimes\cdots \otimes\rho_{n,r_{n}}^*)]\right) \Bigg|\\
        &=\frac{1}{2^{n}}\Bigg|\sum_{\{r:R=1\}}\bigg(\Pr[0\la\cA(\rho_{1,r_1}^*\otimes\cdots\otimes\rho_{x-1,r_{x-1}}^*\otimes\rho_{x,0}^* \otimes\rho_{x+1,r_{x+1}}^*\otimes\cdots \otimes\rho_{n,r_{n}}^*)]\\
        &-\Pr[0\la\cA(\rho_{1,r_1}^*\otimes\cdots\otimes\rho_{x-1,r_{x-1}}^*\otimes\rho_{x,1}^* \otimes\rho_{x+1,r_{x+1}}^*\otimes\cdots \otimes\rho_{n,r_{n}}^*)]\bigg)\Bigg|\\
        &=\frac{1}{2}\abs{\Pr[\mathsf{Hyb}_1(0)=1]-\Pr[\mathsf{Hyb}_2(0)=1]}.
    \end{align}
    This implies that if there exists a QPT adversary such that $\abs{\Pr[\mathsf{Hyb}_1(0)=1]-\Pr[\mathsf{Hyb}_2(0)=1]}$ is non-negligible, then $\cB_x$ breaks the computational hiding of $\Sigma_x^*$.
    Therefore, we have
        \begin{align}
        \abs{\Pr[\mathsf{Hyb}_1(0)=1]-\Pr[\mathsf{Hyb}_2(0)=1]}\leq \negl(\secp).
    \end{align}

In a similar way, we can prove that
    \begin{align}
        \abs{\Pr[\mathsf{Hyb}_1(1)=1]-\Pr[\mathsf{Hyb}_2(1)=1]}\leq \negl(\secp).
    \end{align}
\end{proof}

\end{proof}

\subsection{Universal Construction}
\begin{definition}
    We say that a sequence of uniform QPT unitaries $\Sigma_{\mathsf{Univ}}=\{Q_0(\secp),Q_1(\secp)\}_{\secp\in\N}$ is a universal construction of canonical quantum bit commitment if $\Sigma_{\mathsf{Univ}}$ is canonical quantum bit commitment as long as there exists canonical quantum bit commitment.
\end{definition}

\begin{theorem}\label{thm:univ_com}
    There exists a universal construction of canonical quantum bit commitment.
\end{theorem}
The proof is almost the same as \cref{thm:univ_owsg}, and thus we skip the proof.
\section{Robust Combiner for Unclonable Encryption}\label{sec:unc}
\begin{definition}[Robust Combiner for Unclonable Secret-Key Encryption]
    A robust combiner for (one-time) unclonable secret-key encryption with $\ell(\secp)$-bit plaintexts is a deterministic classical polynomial-time Turing machine $\cM$ with the following properties:
    \begin{itemize}
        \item $\cM$ takes as input $1^n$ with $n\in\N$ and $n$-candidates (one-time) unclonable secret-key encryption with $\ell(\secp)$-bit plaintexts $\{\Sigma_i\seteq(\keygen_i,\Enc_i,\Dec_i)\}_{i\in[n]}$ promised that all candidates satisfies efficiency, and outputs a set of algorithms $\Sigma\seteq(\keygen,\Enc,\Dec)$.
        \item If all of $\{\Sigma_i\}_{i\in[n]}$ satisfies efficiency and at least one of $\{\Sigma_i\}_{i\in[n]}$ satisfies correctness, (one-time) IND-CPA security and (one-time) unclonable IND-CPA security, then $\Sigma$ is (one-time) unclonable secret-key encryption for $\ell(\secp)$-bit plaintexts that satisfies efficiency, correctness, (one-time) IND-CPA security and (one-time) unclonable IND-CPA security.
    \end{itemize}
\end{definition}

In this section, we prove the following \cref{thm:unclone_ske_comb}.
\begin{theorem}\label{thm:unclone_ske_comb}
   There exists a robust combiner for (one-time) unclonable secret-key encryption with $\ell(\secp)$-bit plaintexts for all polynomial $\ell$. 
\end{theorem}

As a corollary, we obtain the following \cref{cor:unclone_pke}.
\begin{corollary}\label{cor:unclone_pke}
    There exists a robust combiner for unclonable public-key encryption with $\ell(\secp)$-bit plaintexts for all polynomial $\ell$. 
\end{corollary}
\begin{proof}[Proof of \cref{cor:unclone_pke}]
We give a rough sketch of the proof.

\cref{cor:unclone_pke} follows from the following observations. 
We can trivially obtain one-time unclonable SKE from unclonable PKE.
From \cref{thm:unclone_ske_comb}, we have a robust combiner for one-time unclonable SKE. 
Furthermore, we can trivially construct PKE with quantum ciphertexts from unclonable PKE. It is known that there exists a robust PKE combiner~\cite{EC:HKNRR05}, and we observe that we can also construct a robust combiner for PKE with quantum ciphertexts in the same way.
Moreover, we can construct unclonable PKE from one-time unclonable SKE, and PKE with quantum ciphertexts.
This is because we can construct unclonable PKE from one-time SKE and receiver non-committing encryption with quantum ciphertexts 
~\footnote{
\cite{TCC:AK21} shows that unclonable PKE can be constructed from one-time unclonable SKE and PKE with classical ciphertexts.
Note that it is unclear whether we can construct unclonable PKE from one-time SKE and PKE with ``quantum'' ciphertexts in the same way as \cite{TCC:AK21}.
This is because they use the existence of OWFs in their proof although it is unclear whether PKE with quantum ciphertexts implies OWFs.
Therefore, we use the technique of \cite{Asia:HMNY21} instead.
(For the detail, see \cref{sec:app_unc_pke})
} (For the detail, see \cref{sec:app_unc_pke}),
and
receiver non-committing encryption with quantum ciphertexts can be constructed from PKE with quantum ciphertexts in the same way as the classical ciphertext case~\cite{TCC:CanHalKat05,C:KNTY19}.

By combining these observations, we can construct a robust combiner for unclonable PKE as follows.
Given candidates of unclonable PKE $\{\Sigma_i\}_{i\in[n]}$, we first use a robust combiner for one-time unclonable SKE, and obtain a new candidate of one-time unclonable SKE $\Sigma_{\SKE}$ regarding each candidate $\Sigma_i$ as a one-time unclonable SKE scheme.
Next, we use a robust combiner for PKE with quantum ciphertexts and obtain a new candidate of PKE with quantum ciphertexts $\Sigma_{\PKE}$ regarding each candidate $\Sigma_i$ as a (not necessarily unclonable) PKE scheme.
Then, we construct a receiver non-committing encryption with quantum ciphertexts $\Sigma_{\NCE}$ from $\Sigma_{\PKE}$.
Finally, we construct unclonable PKE $\Sigma_{\unc}$ from one-time unclonable SKE $\Sigma_{\SKE}$ and receiver non-committing encryption with quantum ciphertexts $\Sigma_{\NCE}$.
\end{proof}

For proving \cref{thm:unclone_ske_comb}, we introduce the following \cref{lem:unclone_ske_cor}.

\begin{lemma}\label{lem:unclone_ske_cor}
    Let $\Sigma$ be a candidate for (one-time) unclonable secret-key encryption with $\ell(\secp)$-bit plaintexts.
    From $\Sigma$, we can construct a (one-time) unclonable secret-key encryption with $\ell(\secp)$-bit plaintexts  $\Sigma^*\seteq(\keygen^*,\Enc^*,\Dec^*)$ such that:
    \begin{enumerate}
        \item $\Sigma^*$ is a uniform QPT algorithm, if $\Sigma$ is a uniform QPT algorithm.
        \item $\Sigma^*$ satisfies perfect correctness.
        \item $\Sigma^*$ satisfies (one-time) IND-CPA security and (one-time) unclonable IND-CPA security if $\Sigma$ is a uniform QPT algorithm and satisfies correctness, (one-time) IND-CPA security and (one-time) unclonable IND-CPA security.
    \end{enumerate}
\end{lemma}
The proof is almost the same as \cref{lem:amp_OWSG_cor}.
For the reader's convenience, we describe the construction of $\Sigma^*$ in \cref{sec:app_unclone_comb}.

\begin{proof}[Proof of \cref{thm:unclone_ske_comb}]
    Below, we consider a fixed constant $n$ and a fixed polynomial $\ell$.
    Let us describe some notations:
    \paragraph{Notations.}
    \begin{itemize}
        \item Let $\Sigma_i$ be a candidate of (one-time) unclonable secret-key encryption with $\ell(\secp)$-length for $i\in[n]$.
        \item For a candidate of (one-time) unclonable secret-key encryption with $\ell(\secp)$-bit plaintexts $\Sigma_i$, let $\Sigma_i^*\seteq(\keygen^*_i,\Enc^*_i,\Dec^*_i)$ be a candidate of (one-time) unclonable secret-key encryption with $\ell(\secp)$-bit plaintexts derived from \cref{lem:unclone_ske_cor}, which satisfies:
        \begin{itemize}
            \item $\Sigma_i^*$ is a uniform QPT algorithm, if $\Sigma_i$ is a uniform QPT algorithm.
            \item $\Sigma_i^*$ satisfies correctness.
            \item $\Sigma_i^*$ satisfies (one-time) IND-CPA security and
            (one-time) unclonable IND-CPA security if $\Sigma_i$ is uniform QPT algorithm and satisfies correctness, (one-time) IND-CPA security, and
            (one-time) unclonable IND-CPA security.
        \end{itemize}
    \end{itemize} 
    \paragraph{Construction of Robust (One-Time) Unclonable Secret-Key Encryption.}
    A robust combiner for (one-time) unclonable secret-key encryption with $\ell(\secp)$-bit plaintexts is a deterministic classical polynomial-time Turing machine that takes as input $1^n$ and $\{\Sigma_i\}_{i\in[n]}$, and outputs the following set of algorithms  
    $\Sigma=(\keygen,\Enc,\Dec)$:
    \begin{description}
        \item[$\keygen(1^\secp)$:]$ $
        \begin{itemize}
            \item For all $i\in[n]$, run $\sk_i^*\la \keygen_i^*(1^{\secp})$.
            \item Output $\sk\seteq \{\sk_i^*\}_{i\in[n]}$.
        \end{itemize}
        \item[$\Enc(1^\secp,\sk,m)$:]$ $
        \begin{itemize}
        \item For all $i\in[n]$, sample $r_i\la\bit^{\ell(\secp)}$ promised that $\sum_{i\in[n]}r_i=m$, where the $\ell(\secp)$ is the length of plaintext $m$.
        \item For all $i\in[n]$, run $\ct_i^*\la\Enc_i^*(1^\secp,\sk_i^*,r_i)$ for all $i\in[n]$.
        \item Output $\ct\seteq \{\ct_i^*\}_{i\in[n]}$.
        \end{itemize}
        \item[$\Dec(1^\secp,\sk,\ct) $:]$ $
        \begin{itemize}
            \item Run $r_i^*\la\Dec_i^*(1^\secp,\sk_i^*,\ct_i^*)$ for all $i\in[n]$.
            \item Output $\sum_{i\in[n]}r_i^*$.
        \end{itemize}
    \end{description}
    
    \cref{thm:unclone_ske_comb} follows from the following \cref{lem:eff_unc_comb,lem:cor_unc_comb,lem:ind_sec_unc_comb,lem:sec_unc_comb}.
    \begin{lemma}\label{lem:eff_unc_comb}
        If all of $\{\Sigma_i\}_{i\in[n]}$ satisfies efficiency, $\Sigma$ satisfies efficiency.
    \end{lemma}
    \begin{lemma}\label{lem:cor_unc_comb}
        $\Sigma$ satisfies correctness.
    \end{lemma}
    \begin{lemma}\label{lem:ind_sec_unc_comb}
        If all of $\{\Sigma_i\}_{i\in[n]}$ satisfies efficiency and one of $\{\Sigma_i\}_{i\in[n]}$, satisfies both correctness and (one-time) IND-CPA security, then $\Sigma$ satisfies (one-time) IND-CPA security.
    \end{lemma}
    \begin{lemma}\label{lem:sec_unc_comb}
        If all of $\{\Sigma_i\}_{i\in[n]}$ satisfies efficiency and one of $\{\Sigma_i\}_{i\in[n]}$, satisfies both correctness and (one-time) unclonable IND-CPA security, then $\Sigma$ satisfies (one-time) unclonable IND-CPA security.
    \end{lemma}
\cref{lem:eff_unc_comb,lem:cor_unc_comb} trivially follows, and thus we skip the proof.
The proof of \cref{lem:ind_sec_unc_comb} is the same as that of \cref{lem:sec_unc_comb}, and thus we skip the proof.
\begin{proof}[Proof of \cref{lem:sec_unc_comb}]
    We prove the \cref{lem:sec_unc_comb} via a standard hybrid argument.
    For the reader's convenience, we describe the proof.
    For simplicity, we consider the one-time case where $\Sigma_i$ is a candidate of one-time unclonable secret-key encryption for each $i\in[n]$.
    We show that $\Sigma$ satisfies unclonable IND-CPA security as long as all of $\{\Sigma_i\}_{i\in[n]}$ satisfy efficiency and one of $\{\Sigma_i\}_{i\in[n]}$ satisfies one-time unclonable IND-CPA security.
    Let $\Sigma_x$ be the candidate for one-time unclonable secret-key encryption that satisfies both correctness and one-time unclonable IND-CPA security.
    Then, $\Sigma_x^*$ satisfies unclonable IND-CPA security from \cref{lem:unclone_ske_cor}.
    Assume that there exists a QPT adversary $(\cA,\cB,\cC)$ that breaks the one-time unclonable IND-CPA security of $\Sigma$, and then construct a set of QPT adversaries $(\widetilde{\cA_x},\widetilde{\cB_x},\widetilde{\cC_x})$ that breaks the one-time unclonable security of $\Sigma_x^*$. 
    \begin{enumerate}
        \item $\widetilde{\cA_x}$ receives $(m_0,m_1)$ from $\cA$.
        \item $\widetilde{\cA_x}$ samples $r_i\la\bit^{\ell(\secp)}$ for all $i\in[n]\backslash x$, and sends $(M_0\seteq m_0+\sum_{i\in[n]\backslash x}r_i,M_1\seteq m_1+\sum_{i\in [n]\backslash x}r_i)$ to the challenger of $\Sigma_x^*$.
        \item The challenger of $\Sigma_x^*$ samples $b\la\bit$, and runs $\ct_x[M_b]^*\la\Enc_x^*(1^\secp,\sk_x^*,M_b)$.
        \item $\widetilde{\cA_x}$ receives from $\ct_x[M_b]^*$, runs $\sk_i^*\la\keygen_i^*(1^\secp)$ for all $i\in[n]\backslash x$, samples $r_i$ for all $i\in[n]\backslash x$, runs $\ct_i[r_i]^*\la\Enc_i^*(1^\secp,\sk_i^*,r_i)$, and sends $(\ct_1[r_1]^*,\cdots, \ct_{x-1}[r_{x-1}]^*,\ct_x[M_b]^*,\ct_{x+1}[r_{x+1}]^*,\cdots, \ct_{n}[r_{n}]^*)$ to $\cA$.
        \item When $\cA$ outputs $\rho_{\cB,\cC}$, $\widetilde{\cA_x}$ sends $\{\sk_i\}_{i\in[n]\backslash x}$ and the $\cB$ register (resp. the $\cC$ register) to $\widetilde{\cB_x}$ (resp. $\widetilde{\cC_x}$).
        \item $\widetilde{\cB_x}$ and $\widetilde{\cC_x}$ receive $\sk_x^*$ from the challenger of $\Sigma_x^*$.
        \item $\widetilde{\cB_{x}}$ (resp. $\widetilde{\cC_x}$) sends $\{\sk_i^*\}_{i\in[n]}$ and the $\cB$ register to $\cB$ (resp. the $\cC$ register to $\cC$).
        \item The experiment outputs $1$ if $b=b_\cB=b_\cC$, where $b_\cB$ (resp. $b_\cC$) is the output of $\cB$ (resp. $\cC$).
    \end{enumerate}
From the construction of $(\widetilde{\cA_x},\widetilde{\cB_x},\widetilde{\cC_x})$, it perfectly simulates the challenger of $\Sigma$. Therefore, if $(\cA,\cB,\cC)$ breaks the one-time unclonable IND-CPA security of $\Sigma$, then $(\widetilde{\cA_x},\widetilde{\cB_x},\widetilde{\cC_x})$ breaks the one-time unclonable IND-CPA security of $\Sigma_x^*$.
\end{proof}    
\end{proof}

\subsection{Universal Constructions}
\begin{definition}
    We say that a set of uniform QPT algorithms $\Sigma_{\mathsf{Univ}}=(\keygen,\Enc,\Dec)$ is a universal construction of (one-time) unclonable SKE (resp. PKE) if $\Sigma_{\mathsf{Univ}}$ is (one-time) unclonable SKE (resp. PKE) as long as there exists (one-time) unclonable SKE (resp. PKE).
\end{definition}
We give a universal construction of unclonable encryption via robust combiners.

\paragraph{Universal Construction via Robust Combiner}
\begin{theorem}\label{thm:univ_unc}
    There exists a universal construction of (one-time) unclonable SKE and unclonable PKE.
\end{theorem}
The proof is almost the same as \cref{thm:univ_owsg}, and thus we skip the proof.

\section{Universal Plaintext Extension for Unclonable Encryption}\label{Sec:expanstion_UE}
In this section, we prove the following \cref{thm:unclone_multi_space}.

\begin{theorem}\label{thm:unclone_multi_space}
    Assume that there exists a decomposable quantum randomized encoding and one-time unclonable SKE $\Sigma_{\unc}= \Unc.(\keygen,\Enc,\Dec)$ where the size of the quantum circuit $\Unc.\Dec(1^\secp,\cdot,\cdot)$ is $\ell(\secp)$.
    Then, for all polynomial $n$, there exists a polynomial $p$ which depends on the polynomial $n$ and $\ell$ and a set of uniform QPT algorithms $\Sigma=(\keygen,\Enc,\Dec)$ which depends on the polynomial $p$ such that $\Sigma$ is a one-time unclonable secret-key encryption for $n(\secp)$-bit plaintexts.
\end{theorem}

\begin{remark}
Our construction is universal construction for one-time unclonable SKE in the sense that our construction does not depend on the single-bit scheme $\Sigma_{\unc}$ that is assumed to exist except for the size of the decryption circuit of $\Sigma_{\unc}$.
\end{remark}

\color{black}

As corollaries,
we obtain \cref{cor:unclone_expnasion,cor:pub_unclone_expansion}.
\begin{corollary}\label{cor:unclone_expnasion}
    For all polynomial $n$, there exists a set of uniform QPT algorithms $\Sigma=(\keygen,\Enc,\Dec)$ such that $\Sigma$ is unclonable secret-key encryption for $n(\secp)$-bit plaintexts if there exists unclonable secret-key encryption for single-bit plaintexts.
\end{corollary}
\begin{corollary}\label{cor:pub_unclone_expansion}
    For all polynomial $n$, there exists a set of uniform QPT algorithms $\Sigma=(\keygen,\Enc,\Dec)$ such that $\Sigma$ is unclonable public-key encryption for $n(\secp)$-bit plaintexts if there exists unclonable public-key encryption for single-bit plaintexts.
\end{corollary}
\begin{proof}[Proof of \cref{cor:pub_unclone_expansion}]
    We give a rough sketch of the proof of \cref{cor:pub_unclone_expansion}. Note that, in the same way, we can prove \cref{cor:unclone_expnasion}.
    
    We can construct PKE with quantum ciphertexts and one-time unclonable SKE with single-bit plaintexts from unclonable PKE for single-bit plaintexts.
    We can construct decomposable quantum randomized encoding from PKE with quantum ciphertexts.
    Furthermore, from \cref{thm:unclone_multi_space}, we can construct one-time unclonable SKE with $n(\secp)$-bit plaintexts from decomposable quantum randomized encoding and one-time unclonable SKE with single-bit plaintexts.

    On the other hand, we can construct receiver non-committing encryption with quantum ciphertexts from PKE with quantum ciphertexts.
    By combining the receiver non-committing encryption with quantum ciphertexts and one-time unclonable SKE with $n(\secp)$-bit plaintexts, we obtain unclonable PKE with $n(\secp)$-bit plaintexts (For the detail, see \cref{sec:app_unc_pke}).
\end{proof}

\begin{proof}[Proof of \cref{thm:unclone_multi_space}]
First, let us describe notations and observations.

\paragraph{Notations and observations.}
\begin{itemize}
    \item 
    Let $C_{\secp,p}[m]$ be a quantum circuit of size $p(\secp)$ with $\secp$-qubit quantum inputs and $\secp$-bit classical inputs such that it outputs $m$ for any inputs, where $p$ is a polynomial which we specify later. 
    \item Let $\Sigma_{\RE}\seteq \RE.(\Enc,\Dec)$ be a decomposable quantum randomized encoding.
    Given quantum circuit $C$ and $n_1$-length quantum input and $n_2$-length classical input $\mathbf{q}$ and $x$,
    the encoding $\widehat{C}(\mathbf{q},x)$ can be separated as follows:
    \begin{align}
        \widehat{C}(\mathbf{q},x,r,e)=(\widehat{C}_{\mathsf{off}},
        \widehat{C}_{1},\cdots,\widehat{C}_{n_1+n_2})(\mathbf{q},x,r,e),
    \end{align}
    where $r$ is uniformly ransom string and $e$ is some quantum state.
    From decomposability, $\widehat{C}_{\mathsf{off}}$ acts only on $r$ and $e$, and $\widehat{C_i}$ acts only on $\mathbf{q_i}, r$ and $e$ for $i\in[n_1]$, and $\widehat{C_i}$ acts only on $x_i$ and $r$ for $i\in\{n_1+1,\cdots,n_1+n_2\}$.
    For any quantum circuit $C$, we write $\lab[i,x_i]= \widehat{C_i}(x_i,r_i)$ and $\lab[i,\mathbf{q}_i]=\widehat{C_i}(\mathbf{q}_i,r,e)$. 
\end{itemize}
\paragraph{Construction.}
We give a construction of one-time unclonable secret-key encryption  $\Sigma\seteq(\keygen,\Enc,\Dec)$ with $n(\secp)$-bit plaintexts by using decomposable quantum randomized encoding.
In the construction, we only use decomposable quantum randomized encoding.
The construction is secure as long as the underlying decomposable quantum randomized encoding is secure and there exists one-time unclonable secret-key encryption for single-bit plaintexts.
\begin{description}
    \item[$\keygen(1^\secp)$:]$ $
    \begin{itemize}
        \item Sample $x\la\bit^{\secp}$.
        \item Sample $R[i]\la\bit^{\ell(\secp)}$ for all $i\in[\secp]$.
        \item Output $\sk\seteq \left(x,\{R[i]\}_{i\in[\secp]}\right)$.
    \end{itemize}
    \item[$\Enc(1^\secp,\sk,m)$:]$ $
    \begin{itemize}
        \item Parse $\sk=\left(x,\{R[i]\}_{i\in[\secp]}\right)$.
        \item Prepare the quantum circuit $C_{\secp,p}[m]$ that outputs $m$ for any inputs.
        \item Compute $\widehat{C_{\secp,p}}[m]_{\mathsf{off}}$.
        \item Compute 
        $  \{\lab[i,0]\}_{i\in[\secp]},\,\,\,\mbox{and}\,\,\,
        \{\lab[i,b]\}_{i\in \{\secp+1,\cdots,2\secp\},b\in\bit}. 
        $
        \item Sample $S[i]\la\bit^{\ell(\secp)}$ for all $i\in[\secp]$.
        \item Compute $\Lab.\ct[i+\secp,x[i]]\seteq R[i]+\lab[i+\secp,x[i]]$ and $\Lab.\ct[i+\secp,1-x[i]]\seteq S[i]+\lab[i+\secp,1-x[i]]$ for all $i\in[\secp]$.
        \item Output 
        \begin{align}    
        \ct\seteq\left(\widehat{C_{\secp,p}}[m]_{\mathsf{off}}, 
        \     \{\lab[i,0]\}_{i\in[\secp]},
        \{\Lab.\ct[i,b]\}_{i\in \{\secp+1,\cdots,2\secp\},b\in\bit}
        \right).
        \end{align}
    \end{itemize}
    \item[$\Dec(1^\secp,\sk,\ct)$:]$ $
    \begin{itemize}
        \item Parse $\sk=\left(x,\{R[i]\}_{i\in[\secp]}\right)$
        and
        \begin{align}    
        \ct=\left(\widehat{C_{\secp,p}}[m]_{\mathsf{off}}, 
               \{\lab[i,0]\}_{i\in[\secp]},
        \{\Lab.\ct[i,b]\}_{i\in \{\secp+1,\cdots,2\secp\},b\in\bit} 
        \right).
        \end{align}
        \item Compute $ \lab[i+\secp,x[i]]\seteq\Lab.\ct[i+\secp,x[i]]+R[i]$ for all $i\in[\secp]$.
        \item Compute 
        \begin{align}
            \RE.\Dec\left(\widehat{C_{\secp,p}}[m]_{\mathsf{off}},
                  \{\lab[i,0]\}_{i\in[\secp]},
        \{\lab[i,x[i]]\}_{i\in \{\secp+1,\cdots,2\secp\}}
            \right)
        \end{align}
        and outputs its output.
    \end{itemize}
\end{description}
\begin{lemma}\label{lem:extend_eff}
    $\Sigma$ satisfies efficiency if $\Sigma_{\mathsf{RE}}$ is decomposable quantum randomized encoding.
\end{lemma}
\begin{lemma}\label{lem:extend_cor}
    $\Sigma$ satisfies correctness if $\Sigma_{\mathsf{RE}}$ is decomposable quantum randomized encoding.
\end{lemma}
\begin{lemma}\label{lem:expand_sec}
    If $\Sigma_{\mathsf{RE}}$ is decomposable quantum randomized encoding and there exists one-time unclonable secret-key encryption with single-bit plaintexts,
    $\Sigma$ satisfies one-time IND-CPA security for some polynomial $p$.
\end{lemma}

\begin{lemma}\label{thm:expnad_unc}
    If $\Sigma_{\mathsf{RE}}$ is decomposable quantum randomized encoding and there exists one-time unclonable secret-key encryption with single-bit plaintexts,
    $\Sigma$ satisfies one-time unclonable IND-CPA security for some polynomial $p$.  
\end{lemma}
\cref{lem:extend_eff} straightforwardly follows.
We can see that \cref{lem:extend_cor} holds as follows.
First, if $\sk\la\keygen(1^\secp)$ and $\ct\la\Enc(\sk,m)$, $\Dec(\sk,\ct)$ outputs the output of $C_{\secp,p}[m](0^{\secp},x)$.
From the definition of $C_{\secp,p}[m]$, $C_{\secp,p}[m](0^{\secp},x)$ outputs $m$ for all $x$.

The proof of \cref{lem:expand_sec} is the same as \cref{thm:expnad_unc}, and thus we skip the proof.
\begin{proof}[Proof of \cref{thm:expnad_unc}]
By a standard argument, we can show the following \cref{prop:perfect_unc}.
\begin{proposition}\label{prop:perfect_unc}
    If there exists one-time unclonable secret-key encryption for single-bit plaintexts, then there exists a one-time unclonable secret-key encryption for single-bit plaintexts scheme $\Sigma_{\unc}=\Unc.(\keygen,\Enc,\Dec)$ such that the following properties are satisfied:
    \begin{enumerate}
        \item $\Sigma_{\unc}$ satisfies perfect correctness.
        \item For all security parameters $\secp\in\N$ and $b\in\bit$, we have $\abs{\sk_\secp}=\abs{\ct_{\secp,b}}=\secp$, where $\sk_\secp\la\Unc.\keygen(1^\secp)$ and $\ct_{\secp,b}\la\Unc.\Enc(1^\secp,\sk_\secp,b)$.
        \item For all security parameters $\secp$, $\Unc.\keygen(1^\secp)$ uniformly randomly samples $\sk_\secp$.
    \end{enumerate}
\end{proposition}
We give the proof of \cref{prop:perfect_unc} in \cref{sec:app_perfect_unc}.
We define$D_{\secp}[m_0,m_1]$ as a quantum circuit that takes as input $\secp$-qubit quantum inputs $\rho$ and $\secp$-bit classical bits $x$, runs the quantum circuit $b\la\Unc.\Dec(1^\secp,x,\rho)$, and outputs $m_b$.
Now, we define $p$ as a polynomial large enough to run the circuit $D_{\secp}[m_0,m_1]$.

We describe the sequence of hybrids against the adversary $(\cA, \cB, \cC)$.
\begin{description}
        \item[$\mathsf{Hyb_0}$:] This is the original one-time unclonable IND-CPA security experiment.
    \begin{enumerate}
        \item The challenger samples $b\la\bit$.
        \item The challenger samples $x\la\bit^{\secp}$ and $R[i]\la\bit^{\ell(\secp)}$ for all $i\in[\secp]$.
        \item $\cA$ sends $(m_0,m_1)$ to the challenger.
        \item The challenger computes $\widehat{C_{\secp,p}}[m_b]_{\mathsf{off}}$,
        $ \{\lab[i,0]\}_{i\in[\secp]}$, and $ \{\lab[i,\beta]\}_{i\in\{\secp+1,\cdots,2\secp\},\beta\in\bit}$.
        \item The challenger samples $S[i]\la\bit^{\ell(\secp)}$ for all $i\in[\secp]$, and computes 
        \begin{align}
        &\Lab.\ct[i+\secp,x[i]]\seteq R[i]+\lab[i+\secp,x[i]]\\
        &\Lab.\ct[i+\secp,1-x[i]]\seteq S[i]+\lab[i+\secp,1-x[i]]
        \end{align}
        for all $i\in[\secp]$.
        \item The challenger sends 
        \begin{align}    
        \ct\seteq\left(\widehat{C_{\secp,p}}[m]_{\mathsf{off}}, 
          \{\lab[i,0]\}_{i\in[\secp]},
        \{\Lab.\ct[i,\beta]\}_{i\in \{\secp+1,\cdots,2\secp\},\beta\in\bit} 
        \right).
        \end{align}
        to $\cA$.
        \item $\cA$ produces $\rho_{\cB,\cC}$ and sends the corresponding registers to $\cB$ and $\cC$.
        \item $\cB$ and $\cC$ receives $\left(x,\{R[i]\}_{i\in[\secp]}\right)$, and outputs $b_\cB$ and $b_\cC$.
        \item The experiment outputs $1$ if $b_\cB=b_\cC=b$, and otherwise $0$.
    \end{enumerate}
    \item[$\mathsf{Hyb_1}$:]$ $
    \begin{enumerate}        
        \item The challenger samples $b\la\bit$.
        \item The challenger samples $x\la\bit^{\secp}$ and $R[i]\la\bit^{\ell(\secp)}$ for all $i\in[\secp]$. 
        \item The adversary $\cA$ sends $(m_0,m_1)$ to the challenger.
        \item The challenger computes $\unc.\ct_b\la\Unc.\Enc(1^\secp,x,b)$, where $\unc.\ct_b$ is the $\secp$-length quantum states.
        \item The challenger computes $\widehat{D_{\secp}}[m_0,m_1]_{\mathsf{off}}$,
        $ \{\lab[i,\unc.\ct_b[i]]\}_{i\in[\secp]}$, and 
        $
        \{\lab[i,\beta]\}_{i\in \{\secp+1,\cdots,2\secp\},\beta\in\bit}.
        $
        \item The challenger samples $S[i]\la\bit^{\ell(\secp)}$ for all $i\in[\secp]$, and computes 
        \begin{align}
        &\Lab.\ct[i+\secp,x[i]]\seteq R[i]+\lab[i+\secp,x[i]]\\ &\Lab.\ct[i+\secp,1-x[i]]\seteq S[i]+\lab[i+\secp,1-x[i]]
        \end{align}
        for all $i\in[\secp]$.
        \item The challenger sends 
        \begin{align}    
        \ct\seteq\left(\widehat{D_{\secp}}[m_0,m_1]_{\mathsf{off}}, 
              \{\lab[i,\unc.\ct_b[i]]\}_{i\in[\secp]},
        \{\Lab.\ct[i,\beta]\}_{i\in \{\secp+1,\cdots,2\secp\},\beta\in\bit} 
        \right)
        \end{align}
        to $\cA$.
        \item $\cA$ produces $\rho_{\cB,\cC}$ and sends the corresponding registers to $\cB$ and $\cC$.
        \item $\cB$ and $\cC$ receives $\left(x,\{R[i]\}_{i\in[\secp]}\right)$, and outputs $b_\cB$ and $b_\cC$.
        \item The experiment outputs $1$ if $b_\cB=b_\cC=b$, and otherwise $0$.
    \end{enumerate}
\end{description}
\cref{thm:expnad_unc} follows from the following \cref{prop:hyb_0_hyb_1_unc_expand,prop:hyb_1_unc_expand}.
\begin{proposition}\label{prop:hyb_0_hyb_1_unc_expand}
If $\Sigma_{\mathsf{RE}}$ is decomposable quantum randomized encoding, then
\begin{align}
    \abs{\Pr[\mathsf{Hyb_0}=1]-\Pr[\mathsf{Hyb_1}=1]}\leq\negl(\secp).
\end{align}
\end{proposition}

\begin{proposition}\label{prop:hyb_1_unc_expand}
    If there exists a one-time unclonable secret-key encryption $\Sigma_{\Unc}$ with single-bit plaintexts, then
    \begin{align}
        \abs{\Pr[\mathsf{Hyb}_1=1]}\leq \frac{1}{2}+\negl(\secp).
    \end{align}
\end{proposition}
\end{proof}

\begin{proof}[Proof of \cref{prop:hyb_0_hyb_1_unc_expand}]
Assume that there exists a QPT adversary $(\cA,\cB,\cC)$ such that
\begin{align}
    \abs{\Pr[\mathsf{Hyb_0}=1]-\Pr[\mathsf{Hyb_1}=1]}
\end{align}
is non-negligible.
Then, construct a QPT adversary $\widetilde{\cA}$ that breaks the security of $\Sigma_{\RE}$ as follows.
\begin{enumerate}
    \item $\widetilde{\cA}$ samples $b\la\bit$.
    \item $\widetilde{\cA}$ samples $x\la\bit^{\secp}$ and $R[i]\la\bit^{\ell(\secp)}$ for all $i\in[\secp]$.
    \item $\widetilde{\cA}$ receives $(m_0,m_1)$ from the $\cA$.
    \item $\widetilde{\cA}$ computes $\unc.\ct_b\la\Unc.\Enc(1^\secp,x,b)$.
    \item $\widetilde{\cA}$ sends $\left(\{C_{\secp,p}[m_b],0^{\secp},x\},  \{ D_{\secp}[m_0,m_1],\unc.\ct_b,x\} \right)$ to the challenger of $\Sigma_{\mathsf{RE}}$ in \cref{prop:ind_RE}.
    \item The challenger samples $b^*\la\bit$, and does the following.
    \begin{itemize}
        \item If $b^*=0$, then the challenger computes 
        \begin{align}
         \left(\widehat{C}_{\mathsf{off}}, \{\lab[i]\}_{i\in[2\secp]}\right)\la
         \RE.\Enc\left(1^\secp,C_{\secp,p}[m_b],\left(0^{\secp},x\right)\right),  
        \end{align}
         and sends $\left(\widehat{C}_{\mathsf{off}}, \{\lab[i]\}_{i\in[2\secp]}\right)$ to $\widetilde{\cA}$.
        \item If $b^*=1$, then the challenger computes 
        \begin{align}
         \left(\widehat{C}_{\mathsf{off}}, \{\lab[i]\}_{i\in[2\secp]}\right)\la
         \RE.\Enc\left(1^\secp,D_{\secp,p}[m_0,m_1],\left(\unc.\ct_b,x\right)\right),  
        \end{align}
        and sends $\left(\widehat{C}_{\mathsf{off}}, \{\lab[i]\}_{i\in[2\secp]}\right)$ to $\widetilde{\cA}$.
    \end{itemize}
    \item $\widetilde{\cA}$ samples $S[i]\la\bit^{\ell(\secp)}$ for all $i\in[\secp]$, computes 
    \begin{align}
    &\Lab.\ct[i+\secp,x[i]]\seteq R[i]+\lab[i+\secp]\\ &\Lab.\ct[i+\secp,1-x[i]]\seteq S[i]
    \end{align}
    for all $i\in[\secp]$,
    and runs $\cA$ on 
    \begin{align}    
        \ct\seteq\left(\widehat{C}_{\mathsf{off}}, 
  \{\lab[i]\}_{i\in[\secp]},
        \{\Lab.\ct[i,\beta]\}_{i\in \{\secp+1,\cdots,2\secp\},\beta\in\bit}
        \right),
    \end{align}
    and generates $\rho_{\cB,\cC}$.
    \item $\widetilde{\cA}$  sends the corresponding register to $\cB$ and $\cC$, respectively.
    \item $\widetilde{\cA}$ sends $x$ and $\{R[i]\}_{i\in[\secp]}$ to $\cB$ and $\cC$.
    \item $\cB$ and $\cC$ outputs $b_\cB$ and $b_\cC$, respectively.
    \item $\widetilde{\cA}$ outputs $1$ if $b=b_\cB=b_\cC$, and outputs $0$ otherwise.
\end{enumerate}
From the construction of $\widetilde{\cA}$, if $b^*=0$, $\widetilde{\cA}$ perfectly simulates the challenger of $\mathsf{Hyb_0}$.
Otherwise, $\widetilde{\cA}$ perfectly simulates the challenger of $\mathsf{Hyb_1}$.
Furthermore, we have 
\begin{align}
C_{\secp,p}[m_b](0^{\secp},x)=D_{\secp}[m_0,m_1](\unc.\ct_b,x)=m_b,
\end{align}
and the size of $C_{\secp,p}$ is equal to $D_{\secp}[m_0,m_1]$ for an appropriate polynomial $p$.
Therefore, if there exists a QPT adversary $(\cA,\cB,\cC)$ such that
\begin{align}
    \abs{\Pr[\mathsf{Hyb_0}=1]-\Pr[\mathsf{Hyb_1}=1]}
\end{align}
is non-negligible,
then it contradicts that $\Sigma_{\RE}$ satisfies security from \cref{prop:ind_RE}.    
\end{proof}

\begin{proof}[Proof of \cref{prop:hyb_1_unc_expand}]
    Assume that there exists a QPT adversary $(\cA,\cB,\cC)$ such that
    $
        \Pr[\mathsf{Hyb_1}=1]
    $
    is non-negligible.
    Then, construct a QPT adversary $(\widetilde{\cA},\widetilde{\cB},\widetilde{\cC})$ that breaks the unclonable IND-CPA security of $\Sigma_{\unc}$ as follows.
    \begin{enumerate}
        \item The challenger of $\Sigma_{\unc}$ samples $b\la\bit$.
        \item $\widetilde{\cA}$ samples $R[i,\beta]\la\bit^{\ell(\secp)}$ for all $i\in[\secp]$ and $\beta\in\bit$.
        \item $\widetilde{\cA}$ receives $(m_0,m_1)$ from the $\cA$.
        \item $\widetilde{\cA}$ sends $(0,1)$ to the challenger, and receives $\unc.\ct_b$, where $\unc.\ct_b\la \Unc.\Enc(1^\secp,x,b)$ and $x\la\bit^{\secp}$.
        \item $\widetilde{\cA}$ computes $\widehat{D_{\secp}}[m_0,m_1]_{\mathsf{off}}$, $\{\lab[i,\unc.\ct_b[i]]\}_{i\in[\secp]}$, and  $\{\lab[i,\beta]\}_{i\in \{\secp+1,\cdots,2\secp\},\beta\in\bit}$.
        \item $\widetilde{\cA}$ computes $\Lab.\ct[i+\secp,\beta]\seteq R[i,\beta]+\lab[i+\secp,\beta]$ for all $i\in[\secp]$ and $\beta\in\bit$.
        \item $\widetilde{\cA}$ runs $\cA$ on 
        \begin{align}
        \left( \widehat{D_{\secp}}[m_0,m_1]_{\mathsf{off}},
             \{\lab[i,\unc.\ct_b[i]]\}_{i\in[\secp]},
        \{\Lab.\ct[i,\beta]\}_{i\in \{\secp+1,\cdots,2\secp\},\beta\in\bit}
        \right),
        \end{align}
         obtains $\rho_{\cB,\cC}$, and sends the $\cB$ register and $\{R[i,\beta]\}_{i\in[\secp],\beta\in\bit}$ to $\cB$ and the $\cC$ register and $\{R[i,\beta]\}_{i\in[\secp],\beta\in\bit}$ to $\cC$.
        \item $\widetilde{\cB}$ (resp. $\widetilde{\cC}$) receives the secret-key $x$ from the challenger of $\Sigma_{\unc}$ and sends $(x,\{R[i,x[i]]\}_{i\in[\secp]})$ and the $\cB$ register (resp. $\cC$ register) to $\cB$ (resp. $\cC$). 
        \item The experiment outputs $1$ if $b=b_\cB=b_\cC$ where $b_\cB$ and $b_\cC$ are the outputs of $\cB$ and $\cC$, respectively.
    \end{enumerate}
From the construction of $(\widetilde{\cA},\widetilde{\cB},\widetilde{\cC})$, it perfectly simulates the challenger of $\mathsf{Hyb_1}$.
Therefore, if there exists some QPT adversaries $(\cA,\cB,\cC)$ such that
$
        \Pr[\mathsf{Hyb_1}=1]
$
is non-negligible, it contradicts that $\Sigma_{\unc}$ satisfies unclonable IND-CPA security.
\end{proof}
\end{proof}

{\bf Acknowledgements.}
TH is supported by 
JSPS research fellowship and by JSPS KAKENHI No. JP22J21864.

\ifnum\llncs=1
\bibliographystyle{alpha} 
\bibliography{abbrev3,crypto,reference}
\else
\bibliographystyle{alpha} 
\bibliography{abbrev3,crypto,reference}
\fi

\ifnum\cameraready=1
\else

\appendix
\section{Proof of \cref{prop:owsg_time}}\label{sec:padding}
Assume that there exists an OWSG.
Then, there exists a set of classical Turing machines $\cM\seteq(x,y,z)$ such that $\Sigma[\cM]\seteq (\keygen[x],\StateGen[y],\Vrfy[z])$ satisfies correctness and security because OWSG is a set of uniform QPT algorithms.
Let $c_x$, and $c_y$, and $c_z$ be a constant such that $x$, $y$, and $z$ halts within $\secp^{c_x}$, $\secp^{c_y}$, and $\secp^{c_z}$ steps for all sufficiently large $\secp\in\N$, respectively.
For simplicity, let us assume that $c_x\geq c_y\geq c_z$.
Note that the same argument goes through in the other cases.

For the set of uniform QPT algorithms $\Sigma[\cM]= (\keygen[x],\StateGen[y],\Vrfy[z])$,   $\Sigma[\cM^*]\seteq(\keygen[x^*],\StateGen[y^*],\allowbreak\Vrfy[z^*])$ is the set of uniform algorithms working as follows:
\begin{description}
    \item[$\keygen\lbrack x^*\rbrack(1^\secp)$:]$ $
    \begin{itemize}
        \item It runs a classical Turing machine $x$ on $1^{\kappa}$ and obtain a general quantum circuit $C[x]_{\kappa}$, where the $\kappa\in\N$ is the largest integer such that $\kappa+\kappa^{c_x}\leq \secp$.
        \item Output $k$, which is the output of $C[x]_{\kappa}$.
    \end{itemize}
    \item[$\StateGen\lbrack y^*\rbrack(1^\secp,k)$:]$ $
    \begin{itemize}
        \item It runs a classical Turing machine $y$ on $(1^{\kappa},k)$ and obtain a general quantum circuit $C[y]_{\kappa,k}$, where the $\kappa\in\N$ is the largest integer such that $\kappa+\kappa^{c_x}\leq \secp$.
        \item Output $\psi_k$, which is the output of $C[y]_{\kappa,k}$.
    \end{itemize}
    \item[$\Vrfy\lbrack z^*\rbrack(1^\secp, k,\psi_k)$:]$ $
    \begin{itemize}
        \item It runs a classical Turing machine $z$ on $(1^{\kappa},k, \abs{\psi_k})$ and obtain a general quantum circuit $C[z]_{\kappa,k,\abs{\psi_k}}$, where the $\kappa\in\N$ is the largest integer such that $\kappa+\kappa^{c_x}\leq \secp$.
        \item Output $\top$ if $1\la C[z]_{\kappa,k,\abs{\psi_k}}(\psi_k)$, and output $\bot$ if $0\la C[z]_{\kappa,k,\abs{\psi_k}}(\psi_k)$.
    \end{itemize} 
\end{description}
We can see that $x^*$, $y^*$, and $z^*$ halts within $\secp^3$ steps for all sufficiently large $\secp\in\N$.
Given $1^\secp$, $x^*$ first computes $\kappa$ within $O(\secp^2)$ steps.
Furthermore, $x(1^\kappa)$ halts within $\kappa^{c_x}\leq \secp$ steps.
Overall, $x^*$ halts within $O(\secp^2)$ steps.
For the same reason, $y^*$ and $z^*$ also halts within $O(\secp^2)$ steps.
Therefore, for all sufficiently large security parameters $\secp\in\N$, $x^*(1^\secp)$, $y^*(1^\secp)$, and $z^*(1^\secp)$ halt within $\secp^3$ steps.
Apparently, $\Sigma[\cM^*]$ satisfies correctness if $\Sigma[\cM]$ satisfies correctness.

Furthermore, by a standard hybrid argument, we can show that the construction satisfies security as follows.
Suppose that $\Sigma[\cM^*]$ does not satisfy security and show that $\Sigma[\cM]$ does not satisfy security.
Since we assume that $\Sigma[\cM^*]$ does not satisfy security, there exists a polynomial $t$, a constant $C$ and a QPT adversary $\cA$ such that
\begin{align}
    \Pr[\top\la\Vrfy[z^*](1^{\secp+\secp^{c_x}}, k^*,\psi_k):
    \begin{array}{ll}
        k\la \keygen[x^*](1^{\secp+\secp^{c_x}})  \\
        \psi_k\la \StateGen[y^*](1^{\secp+\secp^{c_x}})\\
        k^*\la \cA(\psi_k^{\otimes t(\secp+\secp^{c_x})})
    \end{array}
    ]\geq 1/(\secp+\secp^{c_x})^C
\end{align}
for infinitely many security parameters $\secp$.
Let $t'$ be a polynomial such that $t'(\secp)\geq t(\secp+\secp^{c_x})$ for all $\secp\in\N$.
Now, we construct a QPT adversary $\cB$ that breaks $\Sigma[\cM]$ as follows.
\begin{enumerate}
    \item $\cB$ first receives $\psi_k^{\otimes t'(\secp)}$, where $k\la\keygen[x](1^{\secp})$ and $\psi_k\la\StateGen[y](1^\secp,k)$.
    \item $\cB$ runs $k^*\la\cA(\psi_k^{\otimes t(\secp+\secp^{c_x})})$.
    \item $\cB$ outputs $k^*$. 
\end{enumerate}
From the construction of $(\keygen[x^*],\StateGen[y^*],\Vrfy[z^*])$, $(\keygen[x^*](1^{\secp+\secp^{c_x}}), \StateGen[y^*](1^{\secp+\secp^{c_x}},k),\allowbreak\Vrfy[z^*](1^{\secp+\secp^{c_x}},k,\psi_k))$ works in the same way as $(\keygen[x](1^\secp),\StateGen[y](1^\secp,k),\Vrfy[z](1^{\secp},k,\psi_k))$.
Therefore, there exists some constant $D$ such that
\begin{align}
    &\Pr[\top\la\Vrfy[z](1^{\secp}, k^*,\psi_k):
    \begin{array}{ll}
        k\la \keygen[x](1^{\secp})  \\
        \psi_k\la \StateGen[y](1^{\secp},k)\\
        k^*\la \cB(\psi_k^{\otimes t'(\secp)})
    \end{array}
    ]\\
   &= \Pr[\top\la\Vrfy[z^*](1^{\secp+\secp^{c_x}}, k^*,\psi_k):
    \begin{array}{ll}
        k\la \keygen[x^*](1^{\secp+\secp^{c_x}})  \\
        \psi_k\la \StateGen[y^*](1^{\secp+\secp^{c_x}},k)\\
        k^*\la \cA(\psi_k^{\otimes t(\secp+\secp^{c_x})})
    \end{array}
    ]\geq 1/(\secp+\secp^{c_x})^C\geq 1/\secp^{D}
\end{align}
for infinitely many $\secp$.
This contradicts that $\Sigma[\cM]$ satisfies security.
Therefore, $\Sigma[\cM^*]$ satisfies security.

Therefore, $\Sigma[\cM^*]=(\keygen[x^*],\StateGen[y^*],\Vrfy[z^*])$ is a OWSG scheme, where $x^*$, $y^*$, and $z^*$ halts within $\secp^3$ steps.

\section{Proof of \cref{lem:amp_cor_money}}\label{sec:app_money}

\begin{proof}[Proof of \cref{lem:amp_cor_money}]
We describe $\Sigma^*=(\Mint^*,\Vrfy^*)$.
\begin{description}
    \item[$\Mint^*(1^\secp)$:]$ $
    \begin{itemize}
        \item Run $\mathsf{Check}(\Sigma)$, where $\mathsf{Check}$ works as follows:
        \begin{itemize}
        \item Run $(s_i,\rho_{s_i})\la\Mint(1^\secp)$ for all $i\in[\secp]$.
        \item Run $\Vrfy(1^\secp,s_i,\rho_{s_i})$ for all $i\in[\secp]$.
        \item Output $1$ if the number of $\top\la\Vrfy(1^\secp,s_i,\rho_{s_i})$ is at least $\frac{11\secp}{12}$, and output $0$ otherwise.
        \end{itemize}
        \item If $1\la\mathsf{Check}(\Sigma)$, then run $(s_i,\rho_{s_i})\la\Mint(1^\secp)$ for all $i\in[\secp]$ and output $s^*\seteq \{s_i\}_{i\in[\secp]} $ and $\rho_{s}^*\seteq \bigotimes_{i\in[\secp]}\rho_{s_i}$. 
        \item If $0\la\mathsf{Check}(\Sigma)$, then run $(s_i,\rho_{s_i})\la\Mint(1^\secp)$ for all $i\in[\secp]$ and output $s^*\seteq \top$ and $\rho_{s}^*\seteq \bigotimes_{i\in[\secp]}\rho_{s_i}$.
    \end{itemize}
    \item[$\Vrfy^*(1^\secp, s^*,\rho)$:]$ $ 
    \begin{itemize}
        \item Let $\rho$ be a quantum state on the $\secp$ registers $R[1]\cdots R[\secp]$. 
        \item If $s^*= \top$, output $\top$.
        \item If $s^* \neq \top$, then parse $s^*= \{s_i\}_{i\in[\secp]}$, run $\Vrfy(1^\secp,s_i,\cdot)$ on the $R[i]$ register and obtains $b[i]$ for all $i\in[\secp]$.
        Output $\top$ if the number of $b[i]=\top$ is at least $\frac{3\secp}{4}$.
    \end{itemize}
\end{description}
We have the following \cref{prop:money_eff_2,prop:money_cor_2,prop:money_sec_2}.
\begin{proposition}\label{prop:money_eff_2}
    If $\Sigma$ satisfies efficiency, then $\Sigma^*$ satisfies efficiency.
\end{proposition}

\begin{proposition}\label{prop:money_cor_2}
    $\Sigma^*$ satisfies correctness.
\end{proposition}

\begin{proposition}\label{prop:money_sec_2}
    If $\Sigma$ satisfies efficiency, correctness, and security, then $\Sigma^*$ satisfies security.
\end{proposition}
We omit the proof of \cref{prop:money_eff_2}.
To show \cref{prop:money_cor_2}, we use the following \cref{lem:hoefd}.
\begin{lemma}[Hoeffding's inequality]\label{lem:hoefd}
    Let $X_i\in\bit$ be a two-outcome independent random variable,
    and let $S_n\seteq\sum_{i\in[n]} X_i$.
    Then, we have
    \begin{align}
        \Pr[\abs{S_n-\Exp[S_n]}\geq t]\leq 2\exp(-2t^2/n).
    \end{align}
\end{lemma}

\begin{proof}[Proof of \cref{prop:money_cor_2}]
First, assume that $\Pr[\top\la\Vrfy(1^\secp,s,\rho_s):(s,\rho_s)\la\Mint(1^\secp)]<5/6$, and compute $\Pr[\top\la\Vrfy^*(1^\secp,s^*,\rho_{s^*}): (s^*,\rho_{s^*})\la\Mint^*(1^\secp)]$.
\begin{align}
    &\Pr[\top\la\Vrfy^*(1^\secp,s^*,\rho_s^*): (s^*,\rho_s^*)\la\Mint^*(1^\secp)]\\
    &=\sum_{b,s^*}\bigg(\Pr[b\la\mathsf{Check}(\Sigma)]\Pr[(s^*,\rho_s^*)\la\Mint^*(1^\secp)\mid b\la\mathsf{Check}(\Sigma)]\\
    &\hspace{15mm}\cdot\Pr[\top\la\Vrfy^*(1^\secp,s^*,\rho_s^*)\mid (s^*,\rho_s^*)\la\Mint^*(1^\secp)\wedge b\la\mathsf{Check}(\Sigma)]\bigg)\\
    &\geq \bigg(\Pr[0\la\mathsf{Check}(\Sigma)] \Pr[(\top,\rho_s^*)\la\Mint^*(1^\secp)\mid 0\la\mathsf{Check}(\Sigma)]\\
    &\hspace{15mm} \cdot \Pr[\top\la\Vrfy^*(1^\secp,\top,\rho_s^*)\mid (\top,\rho_s^*)\la\Mint^*(1^\secp) \wedge 0\la\mathsf{Check}(\Sigma)]\bigg)\\
    &= \Pr[0\la\mathsf{Check}(\Sigma)]\\
    &\geq 1-2\exp(-\secp/72).
\end{align} 
Here, in the second equation, we have used that $\Vrfy^*(1^\secp,s,\rho)$ always outputs $\top$ if $s=\top$ and $(\top,\rho_s^*)\la\Mint^*(1^\secp)$ if $0\la\mathsf{Check}(\Sigma)$ and in the second inequality, we have used that $\Pr[0\la\mathsf{Check}(\Sigma)]\geq 1-2\exp(-\secp/72)$
when $\Pr[\top\la\Vrfy(1^\secp,s,\rho_s):(s,\rho_s)\la\Mint(1^\secp)]<5/6$, which we prove later.

Next, assume that $\Pr[\top\la\Vrfy(1^\secp,s,\rho_s):(s,\rho_s)\la\Mint(1^\secp)]\geq 5/6$, and compute $\Pr\allowbreak[\top\la\Vrfy^*(1^\secp,s^*,\rho_s^*): (s^*,\rho_s^*)\la\Mint^*(1^\secp)]$.
\begin{align}
    &\Pr[\top\la\Vrfy^*(1^\secp,s^*,\rho_s^*): (s^*,\rho_s^*)\la\Mint^*(1^\secp)]\\
    &=\sum_{b,s^*}\bigg(\Pr[b\la\mathsf{Check}(\Sigma)]\Pr[(s^*,\rho_s^*)\la\Mint^*(1^\secp)\mid b\la\mathsf{Check}(\Sigma)]\\
    &\hspace{15mm}\cdot\Pr[\top\la\Vrfy^*(1^\secp,s^*,\rho_s^*)\mid (s^*,\rho_s^*)\la\Mint^*(1^\secp)\wedge b\la\mathsf{Check}(\Sigma)]\bigg)\\
    &=\bigg(\sum_{s^*\neq \top}\Pr[1\la\mathsf{Check}(\Sigma)]\Pr[(s^*,\rho_s^*)\la\Mint^*(1^\secp)\mid 1\la\mathsf{Check}(\Sigma)]\\
    &\hspace{15mm} \cdot\Pr[\top\la\Vrfy^*(1^\secp,s^*,\rho_s^*)\mid (s^*,\rho_s^*)\la\Mint^*(1^\secp)\wedge 1\la\mathsf{Check}(\Sigma)] \bigg)+\Pr[0\la\mathsf{Check}(\Sigma)]\\
    &=1-\Pr[1\la\mathsf{Check}(\Sigma)]\\
    &\cdot\Bigg(1-\bigg(\sum_{s^*\neq \top}\Pr[(s^*,\rho_s^*)\la\Mint^*(1^\secp)\mid 1\la\mathsf{Check}(\Sigma)]\Pr[\top\la\Vrfy^*(1^\secp,s^*,\rho_s^*)\mid (s^*,\rho_s^*)\la\Mint^*(1^\secp)\wedge 1\la\mathsf{Check}(\Sigma)] \bigg)\Bigg)\\
    &\geq\bigg(\sum_{s^*\neq \top}\Pr[(s^*,\rho_s^*)\la\Mint^*(1^\secp)\mid 1\la\mathsf{Check}(\Sigma)]\Pr[\top\la\Vrfy^*(1^\secp,s^*,\rho_s^*)\mid (s^*,\rho_s^*)\la\Mint^*(1^\secp)\wedge 1\la\mathsf{Check}(\Sigma)] \bigg)\\
    &\geq 1-2\exp(-\secp/72)).
\end{align}
Here, in the second equation, we have used that $\Vrfy^*$ outputs $\top$ if $\mathsf{Check}(\Sigma)$ outputs 0, and
in the final inequality, we have used that 
\begin{align}
&\sum_{s^*\neq \top}\Pr\allowbreak[(s^*,\rho_s^*)\la\Mint^*(1^\secp)\mid 1\la\mathsf{Check}(\Sigma)]\Pr[\top\la\Vrfy^*(1^\secp,s^*,\rho_s^*)\mid (s^*,\rho_s^*)\la\Mint^*(1^\secp)\wedge 1\la\mathsf{Check}(\Sigma)]\\
&\geq 1-2\exp(-\secp/72)
\end{align}
when $\Pr\allowbreak[\top\la\Vrfy(1^\secp,s,\rho_s):(s,\rho_s)\la\Mint(1^\secp)]\geq 5/6$, which we will prove later.
Therefore, the $\Sigma^*$ satisfies the correctness.

Now, we will prove the part we skipped.
That is, we show $\Pr[0\la\mathsf{Check}(\Sigma)]\geq 1-2\exp(-\secp/72)$ when $\Pr[\top\la\Vrfy(1^\secp,s,\rho_s):(s,\rho_s)\la\Mint(1^\secp)]<5/6$.
We consider the random variable $X_{i}$ as $1$ if $\top\la\Vrfy(1^\secp,s_i,\rho_{s_i})$ for the $i$-th running of the verification algorithm while running $\mathsf{Check}$, and consider $X_i$ as $0$ if $\bot\la\Vrfy(1^\secp,s_i,\rho_{s_i})$. 
If we denote $S_{\secp}\seteq\sum_{i\in[\secp]}X_{i}$, then
$1\la\mathsf{Check}(\Sigma) $ if and only if $S_{\secp}\geq \frac{11\secp}{12}$.
On the other hand,
we have $\Exp[S_{\secp}]<\frac{5\secp}{6}$ because $\Pr[\top\la\Vrfy(1^\secp,s,\rho_s):(s,\rho_s)\la\Mint(1^\secp)]<5/6$.
Therefore, we need to have $\abs{S_{\secp}-\Exp[S_{\secp}]}\geq \frac{\secp}{12}$ for $S_{\secp}\geq \frac{11\secp}{12}$.
Therefore, by applying \cref{lem:hoefd}, we have
\begin{align}
    \Pr[1\la\mathsf{Check}(\Sigma)]\leq \Pr[\abs{S_{\secp}-\Exp[S_{\secp}]}\geq \frac{\secp}{12}]\leq2\exp(-\frac{\secp}{72}).
\end{align}
In the same way, we can prove that 
\begin{align}
&\bigg(\sum_{s^*\neq \top}\Pr[(s^*,\rho_s^*)\la\Mint^*(1^\secp)\mid 1\la\mathsf{Check}(\Sigma)]\Pr[\top\la\Vrfy^*(1^\secp,s^*,\rho_s^*)\mid (s^*,\rho_s^*)\la\Mint^*(1^\secp)\wedge 1\la\mathsf{Check}(\Sigma)] \bigg)\\
&\geq 1-2\exp(-\secp/72))
\end{align}
when $\Pr[\top\la\Vrfy(1^\secp,s,\rho_s):(s,\rho_s)\la\Mint(1^\secp)]\geq 5/6$.
\end{proof}

\begin{proof}[Proof of \cref{prop:money_sec_2}]
Let us introduce the following sequence of hybrids as follows.

\begin{description}
    \item[$\mathsf{Hyb_0}$:] This is the original security experiment of $\Sigma^*$.
    \begin{enumerate}
        \item The challenger first runs $\mathsf{Check}(\Sigma)$.
        \item The challenger does the following.
        \begin{itemize}
            \item If $1\la\mathsf{Check}(\Sigma)$, then run $(s_i,\rho_{s_i})\la\Mint (1^\secp)$ for all $i\in[\secp]$ and send $(\{s_i\}_{i\in[\secp]},\bigotimes_{i\in[\secp]}\rho_{s_i})$ to $\cA$.
            \item If $0\la\mathsf{Check}(\Sigma)$, then run $(s_i,\rho_{s_i})\la\Mint (1^\secp)$ for all $i\in[\secp]$ and send $(\bot,\bigotimes_{i\in[\secp]}\rho_{s_i})$ to $\cA$.
        \end{itemize}
        \item $\cA$ sends $\sigma$ consisting of $2\secp$ registers $\{R[1]\cdots R[2\secp]\}$.
        \item The challenger does the following.
        \begin{itemize}
            \item If $1\la\mathsf{Check}(\Sigma)$, then run $\Vrfy(1^\secp,s_i,\cdot)$ on the $R[i]$ register and obtain $b[i]$ for all $i\in[\secp]$. Set $A=1$ if the number of $i\in[\secp]$ such that $b[i]=\top$ is at least $3\secp/4$, and set $A=0$ otherwise. 
            Run
            $\Vrfy(1^\secp,s_i,\cdot)$ on the $R[i+\secp]$ register and obtain $b[i+\secp]$ for all $i\in[\secp]$.
            Set $B=1$ if the number of $i\in[\secp]$ such that $b[i+\secp]=\top$ is at least $3\secp/4$, and set $B=0$ otherwise. 
            If $A=B=1$, then the challenge outputs $\top$, and outputs $\bot$ otherwise.
            \item If $0\la\mathsf{Check}(\Sigma)$, then the challenger always outputs $\top$.
        \end{itemize}
    \end{enumerate}
    \item[$\mathsf{Hyb_1}$:] This is the same as $\mathsf{Hyb_0}$ except that the challenger always behaves as the case of $1\la\mathsf{Check}(\Sigma)$.
\end{description}
It is sufficient to show that
\begin{align}
    \Pr[\mathsf{Hyb_0}=1]\leq \negl(\secp).
\end{align}
Because $\Sigma$ satisfies correctness, $1\la\mathsf{Check}(\Sigma)$ occurs with overwhelming probability. Therefore, we have
\begin{align}
    \abs{\Pr[\mathsf{Hyb_0}=1]-\Pr[\mathsf{Hyb_1}=1]}\leq \negl(\secp).
\end{align}

Furthermore, we can show that 
\begin{align}
    \Pr[\mathsf{Hyb_1}=1]\leq \negl(\secp)
\end{align}
as long as $\Sigma$ satisfies security as follows.
For contradiction assume that there exists a QPT adversary $\cA$ such that 
\begin{align}
    \Pr[\mathsf{Hyb_1}=1]
\end{align}
is non-negligible, and then construct a QPT adversary $\cB$ that breaks $\Sigma$ as follows.
\begin{enumerate}
    \item $\cB$ receives $(s,\rho_s)$ from the challenger of $\Sigma$.
    \item $\cB$ samples $i^*\in[\secp]$, and sets $s_{i^*}=s$ and $\rho_{s_{i^*}}\seteq \rho_s$.
    \item $\cB$ generates $(s_i,\rho_{s_i})\la\Mint(1^\secp)$ for all $i\in[\secp]\setminus \{i^*\}$, and sends $\left(\{s_j\}_{j\in[\secp]},\bigotimes_{j\in[\secp]}\rho_{s_j}\right)$ to $\cA$.
    \item $\cB$ receives $\sigma$ consisting of $2\secp$ registers $R[1],\cdots,R[2\secp]$.
    \item  For all $i\in[\secp]\setminus \{i^*\}$, $\cB$ runs $\Vrfy(1^\secp,s_i,\cdot)$ on the $R[i]$ and $R[i+\secp]$ registers, and obtains $b[i]$ and $b[i+\secp]$, respectively.
    \item $\cB$ sends the $R[i^*]$ and $R[i^*+\secp]$ registers to the challenger, and the challenger runs $\Vrfy(1^\secp,s_{i^*},\cdot)$ on the $R[i^*]$ and $R[i^*+\secp]$ registers, and obtains $b[i^*]$ and $b[i^*+\secp]$.
\end{enumerate}
Clearly, $\cB$ simulates the challenger of $\mathsf{Hyb_1}$.
We write $\mathsf{First}$ to mean the event such that the number of $i\in[\secp]$ that satisfies $b[i]=\top$ is at least $3\secp/4$.
Similarly, we write $\mathsf{Second}$ to mean the event such that the number of $i\in[\secp]$ that satisfies $b[i+\secp]=\top$ is at least $3\secp/4$.
Because $\Pr[\mathsf{Hyb_1}=1]$ is non-negligible, both $\mathsf{First}$ and $\mathsf{Second}$ occur at the same time with non-negligible probability. 
This implies that, with non-negligible probability, the number of $i\in[\secp]$ such that $b[i]=b[i+\secp]=\top$ is at least $\secp/2$.
Because $i^*$ is uniformly random and independent from $\cA$, we have $b[i^*]=b[i^*+\secp]=\top$ with non-negligible probability.
This contradicts that $\Sigma$ satisfies security.
Therefore, we have $\Pr[\mathsf{Hyb_1}=1]\leq \negl(\secp)$.
\end{proof}

\end{proof}

\section{Proof  of \cref{lem:converting_flavor}}\label{sec:convert}
\begin{proof}[Proof of \cref{lem:converting_flavor}]
We prove that if the commitment $\{Q_0(\secp),Q_1(\secp)\}_{\secp\in\N}$ satisfies $c$-X hiding, then $\{\widetilde{Q}_0(\secp),\widetilde{Q}_1(\secp)\}_{\secp\in\N}$ satisfies $\sqrt{c}$-X binding, where $X\in\{$computational, statistical$\}$.
Because the same argument goes through, we consider the case where $X=\mbox{statistical}$.
Below, we fix a security parameter, and write $(Q_0,Q_1)$ and $(\widetilde{Q_0},\widetilde{Q_1})$ to mean $(Q_0(\secp),Q_1(\secp))$ and $(\widetilde{Q_0}(\secp),\widetilde{Q_1}(\secp))$, respectively.

First, let us introduce the following \cref{thm:swap_distinguish}.
\begin{theorem}[Equivalence between swapping and distinguishing~\cite{AAS20,EC:HMY23}]\label{thm:swap_distinguish}
Let $\ket{x_i}$, $\ket{y_i}$ be orthogonal $n$-qubit states and $\ket{\tau_i}$ be an $m$-qubit state. Let $U$ be a polynomial-time computable unitary over $(n+m)$-qubit states and define $\Gamma$ as
\begin{align}
    \Gamma\seteq\norm{(\bra{y}\otimes I^{\otimes m}) U\ket{x}\ket{\tau}+(\bra{x}\otimes I^{\otimes m})U\ket{y}\ket{\tau}}_1.
\end{align}
Then, there exists a non-uniform QPT distinguisher $\cA$ with advice 
$\ket{\tau'}=\ket{\tau}\otimes \frac{\ket{x}\ket{0}+\ket{y}\ket{1}}{\sqrt{2}}$
that
distinguishes 
$\ket{\psi}=\frac{\ket{x}+\ket{y}}{\sqrt{2}}$
and 
$\ket{\phi}=\frac{\ket{x}-\ket{y}}{\sqrt{2}}$
with advantage $\frac{\Gamma^2}{4}$.
Moreover, if $U$ does not act on some qubits, then $\cA$ also does not act on those qubits.
\end{theorem}

Let us assume that $\{\widetilde{Q_{0}}(\secp),\widetilde{Q_{1}}(\secp)\}_{\secp\in\N}$ is not $\sqrt{c}$-statistical biding, and let $d$ be some constant that satisfies $d\geq \sqrt{c}$.
Then, there exists a unitary $U_{\mathbf{\widetilde{R}Z}}$ over $\mathbf{\widetilde{R}}=\mathbf{C}$ and an ancillary register $\mathbf{Z}$ and a state $\ket{\tau}_{\mathbf{Z}}$ such that
\begin{align}
\norm{((\bra{0}\widetilde{Q_{1}}^{\dagger})_{\mathbf{\widetilde{CR}}}\otimes I_{\mathbf{Z}})(I_{\mathbf{\widetilde{C}}}\otimes U_{\mathbf{\widetilde{R}Z}}) ((\widetilde{Q_{0}}\ket{0})_{\mathbf{\widetilde{CR}}}\ket{\tau}_{\mathbf{Z}})}_1\geq d.
\end{align}
We observe that $U$ does not act on $\mathbf{D}$, and thus it cannot cause any interference between states that take $0$ and $1$ in $\mathbf{D}$. 
Therefore, we have
\begin{align}
&((\bra{0}\widetilde{Q}_1^{\dagger})_{\mathbf{\widetilde{CR}}} \otimes I_{\mathbf{Z}} )(I_{\mathbf{\widetilde{C}}}\otimes U_{\mathbf{\widetilde{R}Z}})(\widetilde{Q}_0\ket{0}_{\mathbf{\widetilde{CR}}}\ket{\tau}_{\mathbf{Z}} ) \\
&=\frac{1}{2}
\left(
\begin{array}{ll}
((\bra{0}Q_1^{\dagger})_{\mathbf{CR}}\bra{0}_{\mathbf{D}}\otimes I_{\mathbf{Z}})(I_{\mathbf{R,D}}\otimes U_{\mathbf{C,Z}})(Q_0\ket{0}_{\mathbf{CR}}\ket{0}_{\mathbf{D}}\ket{\tau}_{\mathbf{Z}})\\
-((\bra{0}Q_1^{\dagger})_{\mathbf{CR}}\bra{0}_{\mathbf{D}}\otimes I_{\mathbf{Z}})(I_{\mathbf{R,D}}\otimes U_{\mathbf{C,Z}})(Q_0\ket{0}_{\mathbf{CR}}\ket{0}_{\mathbf{D}}\ket{\tau}_{\mathbf{Z}})
\end{array}
\right).
\end{align}
Similarly, we have
\begin{align}
&((\bra{0}\widetilde{Q}_0^{\dagger})_{\mathbf{\widetilde{CR}}} \otimes I_{\mathbf{Z}} )(I_{\mathbf{\widetilde{C}}}\otimes U_{\mathbf{\widetilde{R}Z}})(\widetilde{Q}_1\ket{0}_{\mathbf{\widetilde{CR}}}\ket{\tau}_{\mathbf{Z}} ) \\
&=\frac{1}{2}
\left(
\begin{array}{ll}
((\bra{0}Q_1^{\dagger})_{\mathbf{CR}}\bra{0}_{\mathbf{D}}\otimes I_{\mathbf{Z}})(I_{\mathbf{R,D}}\otimes U_{\mathbf{C,Z}})(Q_0\ket{0}_{\mathbf{CR}}\ket{0}_{\mathbf{D}}\ket{\tau}_{\mathbf{Z}})\\
-((\bra{0}Q_1^{\dagger})_{\mathbf{CR}}\bra{0}_{\mathbf{D}}\otimes I_{\mathbf{Z}})(I_{\mathbf{R,D}}\otimes U_{\mathbf{C,Z}})(Q_0\ket{0}_{\mathbf{CR}}\ket{0}_{\mathbf{D}}\ket{\tau}_{\mathbf{Z}})
\end{array}
\right).
\end{align}

In particular, we have
\begin{align}
&((\bra{0}\widetilde{Q}_1^{\dagger})_{\mathbf{\widetilde{CR}}} \otimes I_{\mathbf{Z}} )(I_{\mathbf{\widetilde{C}}}\otimes U_{\mathbf{\widetilde{R}Z}})(\widetilde{Q}_0\ket{0}_{\mathbf{\widetilde{CR}}}\ket{\tau}_{\mathbf{Z}} )=((\bra{0}\widetilde{Q}_0^{\dagger})_{\mathbf{\widetilde{CR}}} \otimes I_{\mathbf{Z}} )(I_{\mathbf{\widetilde{C}}}\otimes U_{\mathbf{\widetilde{R}Z}})(\widetilde{Q}_1\ket{0}_{\mathbf{\widetilde{CR}}}\ket{\tau}_{\mathbf{Z}} ).
\end{align}

This implies that
\begin{align}
\norm{
\begin{array}{ll}
((\bra{0}\widetilde{Q}_1^{\dagger})_{\mathbf{\widetilde{CR}}} \otimes I_{\mathbf{Z}} )(I_{\mathbf{\widetilde{C}}}\otimes U_{\mathbf{\widetilde{R}Z}})(\widetilde{Q}_0\ket{0}_{\mathbf{\widetilde{CR}}}\ket{\tau}_{\mathbf{Z}} )+((\bra{0}\widetilde{Q}_0^{\dagger})_{\mathbf{\widetilde{CR}}} \otimes I_{\mathbf{Z}} )(I_{\mathbf{\widetilde{C}}}\otimes U_{\mathbf{\widetilde{R}Z}})(\widetilde{Q}_1\ket{0}_{\mathbf{\widetilde{CR}}}\ket{\tau}_{\mathbf{Z}} )
\end{array}
}_1\geq 2d.
\end{align}

If we set $\ket{x}\seteq \widetilde{Q_{0}}\ket{0}_{\mathbf{\widetilde{CR}}}$ and $\ket{y}\seteq \widetilde{Q_{1}}\ket{0}_{\mathbf{\widetilde{CR}}}$, then $\ket{x}$ and $\ket{y}$ are orthogonal. Then, by \cref{thm:swap_distinguish}, there exists a non-uniform distinguisher $\cA$ with a polynomial-size advice $\ket{\tau'}$ that does not act on $\mathbf{\widetilde{C}}=(\mathbf{R},\mathbf{D})$ and distinguishes
\begin{align}
    \ket{\psi}=\frac{\ket{x}+\ket{y}}{\sqrt{2}}=(Q_0\ket{0}_{\mathbf{CR}})\ket{0}_{\mathbf{D}}
\end{align}
and 
\begin{align}
    \ket{\phi}=\frac{\ket{x}-\ket{y}}{\sqrt{2}}=(Q_1\ket{0}_{\mathbf{CR}})\ket{1}_{\mathbf{D}}
\end{align}
with $d^2\geq c$.
This contradicts that $(Q_0,Q_1)$ satisfies $c$-statistical hiding, and thus $(\widetilde{Q_{0}},\widetilde{Q_{1}})$ satisfies $\sqrt{c}$-statistical binding.
\end{proof}

\section{Proof of \cref{lem:unclone_ske_cor}}\label{sec:app_unclone_comb}

We give the proof of \cref{lem:unclone_ske_cor}.
\begin{proof}[Proof of \cref{lem:unclone_ske_cor}]
For a candidate of one-time unclonable secret-key encryption $\Sigma=(\keygen,\Enc,\Dec)$ with $n(\secp)$-plaintext space, we can assume that $\Dec(1^\secp,\sk,\ct)$ works as follows without loss of generality:

For input $(1^\secp,\sk,\ct) $, run a classical Turing machine $\cM$ on $(1^\secp,\sk,\abs{\ct})$, obtain a unitary $U_{\Dec,\sk}$, append auxiliary state $\ket{0\cdots 0}\bra{0\cdots 0}$ to $\ct$, apply a unitary $U_{\Dec,\sk}$ on $\ct\otimes\ket{0\cdots 0}\bra{0\cdots 0} $, obtain $\rho_{\ct}$, and measure the first $n(\secp)$ qubit of $\rho_{\ct}$ with the computational basis and output its output.

\paragraph{Construction of one-time unclonable secret key encryption:}
We give a construction $\Sigma^*=(\keygen^*,\Enc^*,\Dec^*)$.
\begin{description}
    \item[$\keygen^*(1^\secp)$:]$ $
    \begin{itemize}
        \item Run $\sk\la \keygen(1^{\secp})$.
        \item Output $\sk^*\seteq \sk$.
    \end{itemize}
    \item[$\Enc^*(1^\secp,\sk^*,m)$:]$ $
    \begin{itemize}
        \item Parse $\sk^*=\sk$.
        \item Run $\ct\la\Enc(1^\secp,\sk,m)$.
        \item Measure the first $n(\secp)$-bit of $U_{\Dec,\sk}\left(\ct\otimes \ket{0\cdots 0}\bra{0\cdots 0}\right) U_{\Dec,\sk}^{\dagger}$ in the computational basis, and obtains $m^*$ and post-measurement quantum state $\rho_{m^*,\sk}$.
        \begin{itemize}
            \item If $m=m^*$, then output $\ct^*\seteq U_{\Dec,\sk}^{\dagger}\left(m\otimes \rho_{m,\sk}\right)U_{\Dec,\sk}\otimes \ket{1}\bra{1}$.
            \item If $m\neq m^*$, output $\ct^*\seteq m\otimes \ket{0}\bra{0}$. 
        \end{itemize}
    \end{itemize}
    \item[$\Dec^*(1^\secp,\sk^*,\ct^*)$:]$ $
    \begin{itemize}
        \item Parse $\ct^*=\rho\otimes \ket{b}\bra{b}$ and $\sk^*=\sk$.
        \item Measure the final bit of $\ct^*$ with $\{\ket{1}\bra{1},\ket{0}\bra{0}\}$.
        \begin{itemize}
            \item If the result is $1$, then measure the first $n(\secp)$-bit of $U_{\Dec,\sk}\rho U_{\Dec,\sk}^{\dagger}$ in the computational basis, and outputs its output.
            \item If the result is $0$, then measure the first $n(\secp)$-qubit of $\ct$ in the computational basis and outputs its output.
        \end{itemize}
    \end{itemize}
\end{description}
\end{proof}

\section{Unclonable PKE from One-Time Unclonable SKE and PKE with Quantum Ciphertexts}\label{sec:app_unc_pke}
It was shown that unclonable PKE can be constructed from one-time unclonable SKE and PKE with ``classical'' ciphertexts \cite{TCC:AK21}.
However, it is unclear whether we can construct unclonable PKE from one-time unclonable SKE and PKE with ``quantum'' ciphertexts based on their technique.
This is because their security proof relies on the existence of OWFs, but it is an open problem whether PKE with quantum ciphertexts implies OWFs or not.
Therefore, for the reader's convenience, we construct unclonable PKE from one-time unclonable SKE and PKE with quantum ciphertexts.

Our construction is based on the technique of \cite{Asia:HMNY21}.
First, let us introduce receiver non-committing encryption with quantum ciphertexts.
Note that in the same way as \cite{C:KNTY19,Eprint:HKMNPY23}, we can construct receiver non-committing encryption with quantum ciphertexts from PKE with quantum ciphertexts.
\begin{definition}[Receiver Non-Committing Encryption with Quantum Ciphertexts.]
    An receiver non-committing encryption is a set of algorithms $\Sigma\seteq(\keygen,\Enc,\Dec,\Fake,\Reveal)$ such that:
\begin{itemize}
    \item[$\Setup(1^\secp)$:]It takes $1^\secp$, and outputs a classical key pair $(\pk,\MSK)$.
    \item[$\keygen(1^\secp,\MSK):$]It takes $1^\secp$ and $\MSK$, and outputs a classical key $\sk$.
    \item[$\Enc(1^\secp,\pk,m)$:]It takes $1^\secp$, $\pk$ and $m$, and outputs a quantum ciphertext $\ct$.
    \item[$\Dec(1^\secp,\sk,\ct)$:] It takes $1^\secp$, $\sk$ and $\ct$, and outputs $m$.
    \item[$\Fake(1^\secp,\pk)$:]It takes $1^\secp$ and $\pk$, and outputs a fake quantum ciphertext $\widetilde{\ct}$ and an auxiliary state $\aux$.
    \item[$\Reveal(1^\secp,\pk,\MSK,\aux,m)$:]
    It takes $1^\secp$, $\pk$, $\MSK$, $\aux$, and $m$, and outputs a secret key $\widetilde{\sk}$.
\end{itemize}
\paragraph{Efficiency.}
The algorithms $(\Setup,\keygen,\Enc,\Dec,\Fake,\Reveal)$ are uniform QPT algorithms.
\paragraph{Correctness.}
\begin{align}
    \Pr[m\la\Dec(1^\secp,\sk,\ct):(\pk,\MSK)\la\Setup(1^{\secp}), \sk\la\keygen(1^\secp,\MSK),\ct\la\Enc(1^\secp,\pk,m)]\geq 1-\negl(\secp).
\end{align}
\paragraph{Security.}
Given a receiver non-committing encryption $\Sigma$, we consider a security experiment $\mathsf{Exp_{\Sigma,\cA}^{rec\mbox{-}nc}}(\secp,b)$ against $\cA$.
\begin{enumerate}
    \item The challenger runs $(\pk,\MSK)\la\Setup(1^{\secp})$ and sends $\pk$ to $\cA$.
    \item $\cA$ sends $m$ to the challenger.
    \item The challenger does the following:
    \begin{itemize}
        \item If $b=0$, the challenger generates $\ct\la\Enc(1^\secp,\pk,m)$ and $\sk\la\keygen(1^\secp,\MSK)$, and sends $(\ct,\sk)$ to $\cA$.
        \item If $b=1$, the challenger generates $(\widetilde{\ct},\aux)\la\Fake(1^\secp,\pk)$ and $\widetilde{\sk}\la\Reveal(1^\secp,\pk,\MSK,\aux,m)$, and sends $(\widetilde{\ct},\widetilde{\sk})$ to $\cA$.
    \end{itemize}
    \item $\cA$ outputs $b'\in\bit$, and the experiment outputs $1$ if $b'=b$.
\end{enumerate}
We say that $\Sigma$ is RNC secure if for all sufficiently large security parameters $\secp\in\N$, for any QPT adversary $\cA$, it holds that
\begin{align}
    \abs{\Pr[\mathsf{Exp_{\Sigma,\cA}^{rec\mbox{-}nc}}(\secp,0)=1]-\Pr[\mathsf{Exp_{\Sigma,\cA}^{rec\mbox{-}nc}}(\secp,1)=1]}\leq \negl(\secp).
\end{align}
\end{definition}

\paragraph{Construction}
We construct unclonable PKE $\Sigma=(\keygen,\Enc,\Dec)$ from one-time unclonable SKE $\Sigma_{\SKE}=\SKE.(\keygen,\Enc,\Dec)$ and receiver non-committing encryption with quantum ciphertexts $\Sigma_{\NCE}=\NCE.(\Setup,\keygen,\Enc,\allowbreak\Dec,\Fake,\Reveal)$:
\begin{description}
    \item[$\keygen(1^\secp)$:] $ $
    \begin{itemize}
        \item Run $(\nce.\pk,\nce.\MSK)\la\NCE.\Setup(1^\secp)$ and $\nce.\sk\la\NCE.\keygen(1^\secp,\nce.\MSK)$.
        \item Output $\pk\seteq\nce.\pk$ and $\sk\seteq\nce.\sk$.
    \end{itemize}
    \item[$\Enc(1^\secp,\pk,m)$:]$ $
    \begin{itemize}
        \item Parse $\pk=\nce.\pk$.
        \item Run $\ske.\sk\la\SKE.\keygen(1^\secp)$ and $\ske.\ct\la\SKE.\Enc(1^\secp,\ske.\sk,m)$.
        \item Run $\nce.\ct\la\NCE.\Enc(1^\secp,\nce.\pk,\ske,\sk)$.
        \item Output $\ct\seteq (\nce.\ct,\ske.\ct)$.
    \end{itemize}
    \item[$\Dec(1^\secp,\sk,\ct)$:]$ $
    \begin{itemize}
        \item Parse $\sk=\nce.\sk$ and $\ct=(\nce.\ct,\ske.\ct)$.
        \item Run $\ske.\sk\la\NCE.\Dec(1^\secp,\nce.\sk,\nce.\ct)$.
        \item Run $\SKE.\Dec(1^\secp,\ske.\sk,\ske.\ct)$ and outputs its output.
    \end{itemize}
\end{description}
Obviously, $\Sigma$ satisfies efficiency, correctness, and IND-CPA security.
\begin{lemma}
    $\Sigma$ satisfies unclonable IND-CPA security.
\end{lemma}

\begin{proof}
    We describe the sequence of hybrids against QPT adversaries $(\cA,\cB,\cC)$.
   \begin{description}
        \item[$\mathsf{Hyb_0}$:]This is the original security experiment of $\Sigma$.
        \begin{enumerate}
            \item The challenger samples $b\la\bit$.
            \item $\cA$ receives $\nce.\pk$, where $(\nce.\pk,\nce.\MSK)\la\NCE.\Setup(1^\secp)$.
            \item $\cA$ sends $(m_0,m_1)$ to the challenger.
            \item $\cA$ receives $(\nce.\ct,\ske.\ct_b)$ from the challenger, where $\ske.\sk\la\SKE.\keygen(1^\secp)$, $\nce.\ct\la\NCE.\Enc(1^\secp,\nce.\pk,\ske.\sk)$ and $\ske.\ct_b\la\SKE.\Enc(1^\secp,\ske.\sk,m_b)$. 
            \item $\cA$ generates $\rho_{\cB,\cC}$ and sends the $\cB$ and $\cC$ register to $\cB$ and $\cC$, respectively.
            \item $\cB$ and $\cC$ receives $\nce.\sk$ and outputs $b_\cB$ and $b_\cC$, respectively, where $\nce.\sk\la\NCE.\keygen(1^\secp,\nce.\MSK)$.
            \item The experiment outputs $1$ if $b=b_\cB=b_\cC$.
        \end{enumerate}
        \item[$\mathsf{Hyb_1}$:] This is the same as $\mathsf{Hyb_0}$ except that  $(\widetilde{\nce.\ct},\widetilde{\nce.\sk})$ is used instead of $(\nce.\ct,\nce.\sk)$, where $(\widetilde{\nce.\ct},\aux)\la\Fake(1^\secp,\nce.\pk)$ and $\widetilde{\nce.\sk}\la\Reveal(1^\secp,\nce.\pk,\nce.\MSK,\aux,\ske.\sk)$.
    \end{description}
    We have the following \cref{prop:hyb_0_hyb_1_pke,prop:hyb_1_pke}.
    \begin{proposition}\label{prop:hyb_0_hyb_1_pke}
        If $\Sigma_{\NCE}$ is RNC secure, then
        \begin{align}
            \abs{\Pr[\mathsf{Hyb_0}=1]-\Pr[\mathsf{Hyb_1}=1]}\leq \negl(\secp).
        \end{align}
    \end{proposition}
    \begin{proposition}\label{prop:hyb_1_pke}
        If $\Sigma_{\SKE}$ is one-time unclonable IND-CPA secure, then
        \begin{align}
            \Pr[\mathsf{Hyb_1}=1 ]\leq 1/2+\negl(\secp).
        \end{align}
    \end{proposition}
    \end{proof}

    \begin{proof}[Proof of \cref{prop:hyb_0_hyb_1_pke}]
    This can be shown by a standard hybrid argument.
    Assume that there a QPT adversary $(\cA,\cB,\cC)$ such that
    \begin{align}
        \abs{\Pr[\mathsf{Hyb_0}=1]-\Pr[\mathsf{Hyb_1}=1]}
    \end{align}
    is non-negligible.
    Then, construct a QPT adversary $\widetilde{\cA}$ that breaks the RNC security of $\Sigma_{\NCE}$ as follows.
    \begin{enumerate}
        \item $\widetilde{\cA}$ samples $b\la\bit$.
        \item $\widetilde{\cA}$ receives $\nce.\pk$ from the challenger of $\mathsf{Exp_{\Sigma_{\NCE},\widetilde{\cA}}^{rec\mbox{-}nc}}(\secp,b^*)$, where $(\nce.\pk,\nce.\MSK)\la\NCE.\Setup(1^\secp)$.
        \item $\widetilde{\cA}$ sends $\nce.\pk$ to $\cA$, and receives $(m_0,m_1)$ from $\cA$.
        \item $\widetilde{\cA}$ samples $\ske.\sk\la\SKE.\keygen(1^\secp)$, computes $\ske.\ct_b\la\SKE.\Enc(1^\secp,\ske.\sk,m_b)$, and sends $\ske.\sk$ to the challenger of $\mathsf{Exp_{\Sigma_{\NCE},\widetilde{\cA}}^{rec\mbox{-}nc}}(\secp,b^*)$.
        \item The challenger of $\mathsf{Exp_{\Sigma_{\NCE},\widetilde{\cA}}^{rec\mbox{-}nc}}(\secp,b^*)$ works as follows:
        \begin{itemize}
            \item If $b^*=0$, then runs $\nce.\ct^*\la\NCE.\Enc(1^\secp,\nce.\pk,\ske.\sk)$ and $\nce.\sk^*\la\NCE.\keygen(1^\secp,\nce.\MSK)$, and sends $(\nce.\ct^*,\nce.\sk^*)$ to $\widetilde{\cA}$.
            \item If $b^*=1$, then runs $(\nce.\ct^*,\aux)\la\Fake(1^\secp,\nce.\pk)$ and $\nce.\sk^*\la \Reveal(1^\secp,\nce.\pk,\nce.\MSK,\aux,\ske.\sk)$, and sends $(\nce.\ct^*,\nce.\sk^*)$ to $\widetilde{\cA}$.
        \end{itemize}
        \item $\widetilde{\cA}$ runs $\cA$ on $(\ske.\ct_b,\nce.\ct^*)$, and obtains $\rho_{\cB,\cC}$.
        \item $\widetilde{\cA}$ sends $\nce.\sk^*$ and the $\cB$ register (resp. the $\cC$ register) to $\cB$ (resp. $\cC$).
        \item $\cB$ and $\cC$ outputs $b_\cB$ and $b_\cC$, respectively.
        \item $\widetilde{\cA}$ outputs $1$ if $b=b_\cB=b_\cC$, and $0$ otherwise.
    \end{enumerate}
    From the construction of $\widetilde{\cA}$, 
    \begin{itemize}
        \item If $b^*=0$, $\widetilde{\cA}$ perfectly simulates the challenger of $\mathsf{Hyb_0}$ and thus it outputs the output of $\mathsf{Hyb_0}$.
        \item If $b^*=1$, $\widetilde{\cA}$ perfectly simulates the challenger of $\mathsf{Hyb_1}$ and thus it outputs the output of $\mathsf{Hyb_1}$.
    \end{itemize}
    Therefore, we have
    \begin{align}
        \abs{\Pr[\mathsf{Exp_{\Sigma_{\NCE},\widetilde{\cA}}^{rec\mbox{-}nc}}(\secp,0)=1]-\Pr[\mathsf{Exp_{\Sigma_{\NCE},\widetilde{\cA}}^{rec\mbox{-}nc}}(\secp,1)=1]}=\abs{\Pr[\mathsf{Hyb_0}=1]-\Pr[\mathsf{Hyb_1}=1]},
    \end{align}
    which contradicts that $\Sigma_{\NCE}$ satisfies RNC security.
    \end{proof}

    \begin{proof}[Proof of \cref{prop:hyb_1_pke}]
        This can be shown by a standard hybrid argument.
        Assume that there there exists a constant $C$ and QPT adversaries $(\cA,\cB,\cC)$ such that
        \begin{align}
            \Pr[\mathsf{Hyb_1}=1]\geq 1/2+1/\secp^C
        \end{align}
        for infinitely many security parameters $\secp\in\N$.
        Then, construct a set of QPT adversaries $(\widetilde{\cA},\widetilde{\cB},\widetilde{\cC})$ that breaks the unclonable IND-CPA security of $\Sigma_{\SKE}$ as follows.
        \begin{enumerate}
            \item The challenge of $\Sigma_{\SKE}$ samples $b\la\bit$.
            \item $\widetilde{\cA}$ samples $(\nce.\pk,\nce.\MSK)\la\NCE.\Setup(1^\secp)$ and sends $\nce.\pk$ to $\cA$.
            \item $\widetilde{\cA}$ receives $(m_0,m_1)$ from $\cA$, and sends $(m_0,m_1)$ to the challenger.
            \item $\widetilde{\cA}$ receives $\ske.\ct_b$, where $\ske.\ct_b\la\SKE.\Enc(1^\secp,\ske.\sk,m_b)$ and $\ske.\sk\la\SKE.\keygen(1^\secp)$.
            \item $\widetilde{\cA}$ runs $(\widetilde{\nce.\ct},\aux)\la\Fake(1^\secp,\nce.\pk)$, and runs $\cA$ on $(\widetilde{\nce.\ct},\ske.\ct_b)$, and obtains $\rho_{\cB,\cC}$.
            \item $\widetilde{\cA}$ sends $\aux$, $\nce.\MSK$ and the $\cB$ (resp. $\cC$) register to $\widetilde{\cB}$ (resp. $\widetilde{\cC}$).
            \item $\widetilde{\cB}$ (resp. $\widetilde{\cC}$) receives $\ske.\sk$ and runs $\widetilde{\nce.\sk}\la\Reveal(1^\secp,\nce.\pk,\nce.\MSK,\aux,\ske.\sk)$, and sends 
            $\widetilde{\nce.\sk}$ and the $\cB$ (resp. $\cC$) register to $\cB$ (resp. $\cC$).
            \item $\cB$ and $\cC$ outputs $b_\cB$ and $b_\cC$, respectively.
            \item $\widetilde{\cB}$ and $\widetilde{\cC}$ outputs $b_\cB$ and $b_\cC$ as the guess for $b$, respectively.
        \end{enumerate}
        From the construction of $(\widetilde{\cA},\widetilde{\cB},\widetilde{\cC})$, it perfectly simulates the challenger of $\mathsf{Hyb_1}$.
        Therefore, we have $b=b_\cB=b_\cC$ with non-negligible probability, which implies that $(\widetilde{\cA},\widetilde{\cB},\widetilde{\cC})$ break one-time unclonable IND-CPA security of $\Sigma_{\SKE}$.
        This is a contradiction.
        Therefore, we have
        \begin{align}
                        \Pr[\mathsf{Hyb_1}=1 ]\leq 1/2+\negl(\secp).
        \end{align}
    \end{proof}

\section{Proof of \cref{prop:perfect_unc}}\label{sec:app_perfect_unc}
We give the proof of \cref{prop:perfect_unc}.
\begin{proof}[Proof of \cref{prop:perfect_unc}]
In the same way as proof of \cref{lem:unclone_ske_cor}, we can show that if there exists a one-time unclonable secret-key encryption for single-bit plaintexts, then there exists a scheme $\Sigma^*=(\keygen^*,\Enc^*,\Dec^*)$ that satisfies perfect correctness.

Now, we construct one-time unclonable secret key encryption $\overline{\Sigma}\seteq(\overline{\keygen},\overline{\Enc},\overline{\Dec})$ with uniformly random secret-key and perfect correctness from one-time unclonable secret key encryption $(\keygen^*,\Enc^*,\Dec^*)$ with perfect correctness.
\begin{description}
    \item[$\overline{\keygen}(1^\secp)$:]$ $
    \begin{itemize}
        \item Sample $s\la\bit^{s(\secp)}$, where $s(\secp)$ is the length of the secret-key $\sk$ that $\keygen^*(1^\secp)$ generates 
        \item Output $\overline{\sk}\seteq s$.
    \end{itemize}
    \item[$\overline{\Enc}(1^\secp,\overline{\sk},m)$:]$ $
    \begin{itemize}
        \item Parse $\overline{\sk}\seteq s$.
        \item Run $\sk\la\keygen^*(1^\secp)$.
        \item Run $\ct\la\Enc^*(1^\secp,\sk,m)$.
        \item Output $\overline{\ct}\seteq(\ct,\sk+s)$.
    \end{itemize}
    \item[$\overline{\Dec}(1^\secp,\overline{\sk},\overline{\ct})$:]$ $
    \begin{itemize}
        \item Parse $\overline{\sk}=s$ and $\overline{\ct}=(\ct,\sk^*)$.
        \item Compute $\sk=\sk^*+s$.
        \item Run $\Dec^*(1^\secp,\sk,\ct)$ and output its output.
    \end{itemize}
\end{description}
From the construction, the secret key of $\Sigma^*$ is uniformly random.
Efficiency and perfect correctness of $\overline{\Sigma}$ straightforwardly follow that of $\Sigma^*$.
We can show that $\overline{\Sigma}$ satisfies unclonable IND-CPA security by a standard hybrid argument.
For clarity, we describe the proof of security.

Assume that there exists a QPT adversary $(\cA,\cB,\cC)$ that breaks the unclonable IND-CPA security of $\overline{\Sigma}$.
Then, construct a QPT adversary $(\widetilde{\cA},\widetilde{\cB},\widetilde{\cC})$ that breaks the unclonable IND-CPA security of $\Sigma^*$.
\begin{enumerate}
    \item The challenger of $\Sigma^*$ samples $b\la\bit$.
    \item $\widetilde{\cA}$ samples $s\la\bit^{s(\secp)}$.
    \item $\widetilde{\cA}$ receives $(m_0,m_1)$ from $\cA$.
    \item $\widetilde{\cA}$ sends $(m_0,m_1)$ to the challenger of $\Sigma^*$.
    \item $\widetilde{\cA}$ receives $\ct_b$, where $\sk\la\keygen^*(1^\secp)$ and $\ct_b\la\Enc^*(1^\secp,\sk,m_b)$.
    \item $\widetilde{\cA}$ runs $\cA$ on $(\ct_b,s)$, obtain $\rho_{\cB,\cC}$, and sends $s$ and the $\cB$ register (resp. $\cC$ register) to $\widetilde{\cB}$ (resp. $\widetilde{\cC}$).
    \item $\widetilde{\cB}$ (resp. $\widetilde{\cC}$) receives $\sk$ from the challenger of $\Sigma^*$, and sends $\sk+s$ and the $\cB$ register (resp. $\cC$ register) to $\cB$ (resp.$\cC$).
    \item The experiment outputs $1$ if $b=b_\cB=b_\cC$, where $b_\cB$ and $b_\cC$ are the output of $\cB$ and $\cC$,respectively.
\end{enumerate}
From the construction of $(\widetilde{\cA},\widetilde{\cB},\widetilde{\cC})$, it perfectly simulates the challenger of $\Sigma^*$. Therefore, if $(\cA,\cB,\cC)$ breaks the unclonable IND-CPA security of $\overline{\Sigma}$, it contradicts that $\Sigma^*$ satisfies unclonable IND-CPA security.

In the construction $\overline{\Sigma}$, the size of $\overline{\sk_{\secp}}$ and $\overline{\ct_{\secp,b}}$ are not necessarily equal to $\secp$, where $\overline{\sk_\secp}\la\overline{\keygen}(1^\secp)$ and $\overline{\ct_{\secp,b}}\la\overline{\Enc}(1^\secp,\overline{\sk_\secp},b)$.
By wisely choosing a security parameter and a standard padding argument, from $\overline{\Sigma}$, we can construct $\Sigma=(\keygen,\Enc,\Dec)$ such that $\abs{\sk_\secp}=\abs{\ct_{\secp,b}}=\secp$ for all $\secp\in\N$ and $b$ where $\sk_\secp\la\keygen(1^\secp)$ and $\ct_{\secp,b}\la\Enc(1^\secp,\sk_\secp,b)$.

For clarity, we describe the construction of $\Sigma$.
To describe our construction, let $c$ be a constant such that $\abs{\overline{\sk_\secp}}\leq \abs{\overline{\ct_{\secp,b}}}\leq \secp^c$ for all security parameters $\secp\in\N$ and $b\in\bit$, where $\overline{\sk_\secp}\la\overline{\keygen}(1^\secp)$ and $\overline{\ct_{\secp,b}}\la\overline{\Enc}(1^\secp,\overline{\sk_\secp},b)$.
\begin{description}
    \item[$\keygen(1^\secp)$:]$ $
    \begin{itemize}
        \item Sample $x\la \bit^{\secp}$.
        \item Output $\sk\seteq x$.
    \end{itemize}
    \item[$\Enc(1^\secp,\sk,b)$:]$ $
    \begin{itemize}
        \item Parse $\sk=x$.
        \item Let $\secp'$ be the largest integer such that $\secp'^c\leq \secp$.
        \item Let $\overline{x}$ be the first $\abs{\overline{\sk_{\secp'}}}$-bits of $x$, where $\overline{\sk_{\secp'}}\la\overline{\keygen}(1^{\secp'})$.
        \item Run $\overline{\ct}\la\overline{\Enc}(1^{\secp'},\overline{x},b)$. Note that since $\secp'^c\leq \secp$, the size of $\overline{\ct}$ is smaller than $\secp$.
        \item Output $\ct=(\overline{\ct},0^{\secp-\abs{\overline{\ct}}})$.
    \end{itemize}
    \item[$\Dec(1^\secp,\sk,\ct)$:]$ $
    \begin{itemize}
        \item Parse $\sk=x$ and $\ct=(\overline{\ct},0^{\secp-\abs{\overline{\ct}}})$.
        \item Let $\secp'$ be the largest integer such that $\secp'^c\leq \secp$.
        \item Let $\overline{x}$ be the first $\abs{\overline{\sk_{\secp'}}}$-bits of $x$, where $\overline{\sk_{\secp'}}\la\overline{\keygen}(1^{\secp'})$.
        \item Compute $\overline{\Dec}(1^{\secp'},\overline{x},\overline{\ct})$, and outputs its output.
    \end{itemize}
\end{description}
Efficiency and perfect correctness straightforwardly follow.
From the construction, it is obvious that $\sk_\secp$ is uniformly randomly sampled and $\abs{\sk_\secp}=\abs{\ct_{\secp,b}}=\secp$ for all $\secp\in\N$ and $b\in\bit$, where $\sk_{\secp}\la\keygen(1^\secp)$ and $\Enc(1^\secp,\sk_\secp,b)$.
Furthermore, we can prove its security by a standard hybrid argument.
\end{proof}

	\ifnum\llncs=1
	\newpage
	 	\setcounter{page}{1}
 	{
	\noindent
 	\begin{center}
	{\Large SUPPLEMENTAL MATERIALS}
	\end{center}
 	}
	\setcounter{tocdepth}{2}
	\fi
\fi

\ifnum\cameraready=1
\else
\ifnum\submission=1
\newpage
\setcounter{tocdepth}{1}
\tableofcontents
\else
\fi
\fi

\end{document}